\documentclass[11pt]{amsart}
\usepackage{amsmath}
\usepackage{amsfonts}
\usepackage{mathrsfs}
\usepackage {amssymb,euscript}  
\usepackage {amsmath}
\usepackage {amsthm}
\usepackage {amscd}
\usepackage{hyperref}
\usepackage{color}
\usepackage{epsfig}
\usepackage{caption}  

\usepackage{slashed}  

\usepackage{tikz}

\usepackage{graphics}

\usepackage{sansmath}
\usetikzlibrary{shadings,intersections} 

\usepackage{enumitem}
\newcommand{\subscript}[2]{$#1 _ #2$}

\title[Anisotropic Dynamical Horizon]{Anisotropic Dynamical Horizons\\ Arising in Gravitational Collapse}     
\date{\today}

\author[An]{Xinliang An}
\address{Department of Mathematics\\ National University of Singapore\\ 119076, Singapore}   \email{matax@nus.edu.sg}

\author[Han]{Qing Han}
\address{Department of Mathematics\\
University of Notre Dame\\
Notre Dame, IN 46556, USA} \email{qhan@nd.edu}

\theoremstyle{plain}
\newtheorem{lemma}{Lemma}[section]

\newtheorem{proposition}[lemma]{Proposition}
\newtheorem{theorem}[lemma]{Theorem}
\newtheorem{corollary}[lemma]{Corollary}

\theoremstyle{definition}
\newtheorem{remark}[lemma]{Remark}

\numberwithin{equation}{section}

\begin{document}

\newcommand{\ub}{\underline{u}}
\newcommand{\Cb}{\underline{C}}
\newcommand{\Lb}{\underline{L}}
\newcommand{\Lh}{\hat{L}}
\newcommand{\Lbh}{\hat{\Lb}}
\newcommand{\phib}{\underline{\phi}}
\newcommand{\Phib}{\underline{\Phi}}
\newcommand{\Db}{\underline{D}}
\newcommand{\Dh}{\hat{D}}
\newcommand{\Dbh}{\hat{\Db}}
\newcommand{\omb}{\underline{\omega}}
\newcommand{\omh}{\hat{\omega}}
\newcommand{\ombh}{\hat{\omb}}
\newcommand{\Pb}{\underline{P}}
\newcommand{\chib}{\underline{\chi}}
\newcommand{\chih}{\hat{\chi}}
\newcommand{\chibh}{\hat{\chib}}

\newcommand{\alb}{\underline{\alpha}}
\newcommand{\zeb}{\underline{\zeta}}
\newcommand{\beb}{\underline{\beta}}
\newcommand{\etb}{\underline{\eta}}
\newcommand{\Mb}{\underline{M}}
\newcommand{\oth}{\hat{\otimes}}

%%%%%%%% 
\newcommand{\ka}{\kappa}
\newcommand{\x}{\mathcal X}
\newcommand{\h}{\mathcal H}
\newcommand{\xx}{\mathcal X^*}
\newcommand{\xxx}{\mathcal X^{**}}
\newcommand{\y}{\mathcal Y}
\newcommand{\m}{\mathring}
\newcommand{\n}{\mathcal N}
\newcommand{\Om}{\Omega}
\newcommand{\muu}{\mu^*}
\newcommand{\nin}{\noindent}

\gdef\dxd{\displaystyle}
\gdef\x{\xi}
\def\al{\aligned}
\def\eal{\endaligned}
\def\be{\begin{equation}}
\def\ee{\end{equation}}

%%%%%%%%  

\def\a {\alpha}
\def\b {\beta}
\def\ab {\alphab}
\def\bb {\betab}
\def\nab {\nabla}

\def\ub {\underline{u}}
\def\th {\theta}
\def\Lb {\underline{L}}
\def\Hb {\underline{H}}
\def\chib {\underline{\chi}}
\def\chih {\hat{\chi}}
\def\chibh {\hat{\underline{\chi}}}
\def\omegab {\underline{\omega}}
\def\etab {\underline{\eta}}
\def\betab {\underline{\beta}}
\def\alphab {\underline{\alpha}}
\def\Psib {\underline{\Psi}}
\def\hot{\widehat{\otimes}}
\def\Phib {\underline{\Phi}}
\def\thb {\underline{\theta}}
\def\t {\tilde}
\def\st {\tilde{s}}
\def\h {\hat}

%%%%%%%%%
\def\rhoc{\check{\rho}}
\def\sigmac{\check{\sigma}}
\def\Psic{\check{\Psi}}
\def\kappab{\underline{\kappa}}
\def\betabc {\check{\underline{\beta}}}

%%%%%%%%%%%%
%%%%%%%%%%%%
\def\d {\delta}
\def\f {\frac}
\def\i {\infty}
\def\l {\bigg(}
\def\r {\bigg)}
\def\S {S_{u,\underline{u}}}
\def\o{\omega}
\def\O{\Omega}\
\def\be{\begin{equation}\begin{split}}
\def\en{\end{split}\end{equation}}
\def\at{a^{\frac{1}{2}}}
\def\af{a^{\frac{1}{4}}}
\def\od{\omega^{\dagger}}
\def\ombd{\underline{\omega}^{\dagger}}
\def\K{K-\frac{1}{|u|^2}}
\def\ut{\frac{1}{|u|^2}}
\def\s{\frac{\delta a^{\frac{1}{2}}}{|u|}}
\def\Kb{K-\frac{1}{(u+\underline{u})^2}}
\def\bf{b^{\frac{1}{4}}}
\def\bt{b^{\frac{1}{2}}}
%%%%%
\def\de{\delta}
\def\ls{\lesssim}
\def\om{\omega}
\def\Om{\Omega}
\def\lo{\lambda_1}
\def\lt{\lambda_2}
\def\phib{\bar{\phi}}
\def\bR{\bar{R}}
\def\tR{\tilde{R}}

%%%%%%%%%%%%%
%%%%%%%%%%%%%

\newcommand{\e}{\epsilon}
\newcommand{\et} {\frac{\epsilon}{2}}
\newcommand{\ef} {\frac{\epsilon}{4}}
\newcommand{\LH} {L^2(H_u)}
\newcommand{\LHb} {L^2(\underline{H}_{\underline{u}})}
\newcommand{\M} {\mathcal}
\newcommand{\TM} {\tilde{\mathcal}}
\newcommand{\p}{\phi}
\newcommand{\q}{\underline{\psi}\hspace{1pt}}
\newcommand{\Li}{_{L^{\infty}(S_{u,\underline{u}})}}
\newcommand{\Lt}{_{L^{2}(S)}}
\newcommand{\da}{\delta^{-\frac{\epsilon}{2}}}
\newcommand{\db}{\delta^{1-\frac{\epsilon}{2}}}
\newcommand{\D}{\Delta}

%%%%%%%%%%%%
%%%%%%%%%%%%%

\renewcommand{\div}{\mbox{div }}
\newcommand{\curl}{\mbox{curl }}
\newcommand{\trchb}{\mbox{tr} \chib}
\def\trch{\mbox{tr}\chi}
\newcommand{\tr}{\mbox{tr}}

\newcommand{\Ls}{{\mathcal L} \mkern-10mu /\,}
\newcommand{\eps}{{\epsilon} \mkern-8mu /\,}
%%%%%%%%%

\newcommand{\xib}{\underline{\xi}}
\newcommand{\psib}{\underline{\psi}}
\newcommand{\rhob}{\underline{\rho}}
\newcommand{\thetab}{\underline{\theta}}
\newcommand{\gammab}{\underline{\gamma}}
\newcommand{\nub}{\underline{\nu}}
\newcommand{\lb}{\underline{l}}
\newcommand{\mub}{\underline{\mu}}
\newcommand{\Xib}{\underline{\Xi}}
\newcommand{\Thetab}{\underline{\Theta}}
\newcommand{\Lambdab}{\underline{\Lambda}}
\newcommand{\vphb}{\underline{\varphi}}

\newcommand{\ih}{\hat{i}}

\newcommand{\tcL}{\widetilde{\mathscr{L}}}

\newcommand{\sRic}{Ric\mkern-19mu /\,\,\,\,}
\newcommand{\sL}{{\cal L}\mkern-10mu /}
\newcommand{\sLh}{\hat{\sL}}
\newcommand{\sg}{g\mkern-9mu /}
\newcommand{\seps}{\epsilon\mkern-8mu /}
\newcommand{\sd}{d\mkern-10mu /}
\newcommand{\sR}{R\mkern-10mu /}
\newcommand{\snab}{\nabla\mkern-13mu /}
\newcommand{\sdiv}{\mbox{div}\mkern-19mu /\,\,\,\,}
\newcommand{\scurl}{\mbox{curl}\mkern-19mu /\,\,\,\,}
\newcommand{\slap}{\mbox{$\triangle  \mkern-13mu / \,$}}
\newcommand{\sGamma}{\Gamma\mkern-10mu /}
\newcommand{\somega}{\omega\mkern-10mu /}
\newcommand{\somb}{\omb\mkern-10mu /}
\newcommand{\spi}{\pi\mkern-10mu /}
\newcommand{\sJ}{J\mkern-10mu /}
\renewcommand{\sp}{p\mkern-9mu /}
\newcommand{\su}{u\mkern-8mu /}

\begin{abstract}

For the study of $3+1$ dimensional Einstein vacuum equations (EVEs), substantial progress has been made recently on the problem of trapped surface formation. However, very limited knowledge of existence and associated properties is acquired on the boundary of the emerged trapped region, i.e., the apparent horizon, which is composed of marginally outer trapped surfaces (MOTSs) and is of great physical importance. In this paper, concerning this emerged apparent horizon we prove a folklore conjecture relating to both cosmic censorship and black hole thermodynamics. In a framework set up by Christodoulou and under a general anisotropic condition introduced by  Klainerman, Luk and Rodnianski, for $3+1$ EVEs we prove that \textit{in the process of gravitational collapse, a smooth and spacelike apparent horizon (dynamical horizon) emerges  from general (both isotropic and anisotropic) initial data. }This dynamical horizon censors singularities formed in gravitational collapse from non-trapped local observers near the center, and it also enables the extension of black hole thermodynamical theory along the apparent horizon to anisotropic scenarios. Our analysis builds on scale-critical hyperbolic method and non-perturbative elliptic techniques. New observations and equation structures are exploited. Geometrically, we furthermore construct explicit finger-type single and multi-valley anisotropic apparent horizons. They are the first mathematical examples of the anisotropic MOTS and the anisotropic apparent horizon formed in dynamics, which have potential applications in geometric analysis, black hole mechanics, numerical relativity and gravitational wave phenomenology.  
\end{abstract}  

\maketitle

\tableofcontents
\section{Introduction} 
We study the evolution of the $3+1$ dimensional Einstein vacuum equations (EVEs)  
\begin{equation}\label{EVEs} 
\mbox{Ric}_{\mu\nu}=0
\end{equation}
with anisotropic characteristic initial data.  {\color{black}In 2005, Ashtekar and Galloway \cite{AG} wrote, \textit{``There is a general expectation that under `favorable' circumstances a MTT (marginally trapped tube)\footnote{apparent horizon} that forms during a gravitational collapse becomes spacelike - i.e.,a DH (dynamical horizon)-soon after collapse.''} In this paper, we prove this folklore conjecture and connect it to the cosmic censorship. 

In 1965, Penrose \cite{Penrose} wrote, \textit{``We are thus presented with what is perhaps the most fundamental unanswered question of general-relativistic collapse theory, namely: does there exists a `cosmic censor' who forbids the appearance of naked singularity, clothing each one in an absolute event horizon?"} The above question raised by Penrose was later formulated into \textit{weak cosmic censorship conjecture}; namely, \textit{arising from generically regular initial data, there is no singularity which is visible from infinity}. In 1979 Penrose \cite{Penrose2} further formulated a more local version of ``cosmic censor", and some of his arguments were formulated into ``strong cosmic censorship conjecture". In \cite{Penrose2} he wrote, \textit{``It would seem, therefore, that if cosmic censorship is a principle of Nature, it should be formulated in such a way as to preclude locally naked singularities."} 

Our result shows that, with general (isotropic and anisotropic) initial data prescribed as in Theorem \ref{main thm},  a local version of weak cosmic censorship holds. The first singularity formed in gravitational collapse is completely hidden by the apparent horizon and it is invisible to any local observer outside. Hence, there is no locally naked singularity in the exterior the apparent horizon, where Penrose's stronger local form of cosmic censorship is proved to hold. In our setting, since we allow anisotropic initial data, spacetime dynamics could differentiate from spherically symmetric case,  and close to the center of gravitational collapse, some locally naked singularities may form. The causal curves emitting from these locally naked singularities might be visible to local or far-away observers. In our proof, we show that both scenarios are \textit{not} possible and no locally naked singularity exterior to the apparent horizon can get formed. Mathematically speaking, we control our spacetime up to and slightly beyond the apparent horizon. Before the apparent horizon, there is no singularity. The ``center" along the apparent horizon is a spacetime singularity. And since we prove that a \textit{spacelike} apparent horizon emerges from that ``center", the future causal curves from the ``center" are \textit{all} censored by the apparent horizon and the ``center" (singularity) is completely hidden in the black hole region! This confirms the local version of cosmic censorship.   

Moreover, in our paper we also provide explicit constructions for single and multi-valley anisotropic apparent horizons. They are of finger-type shapes and are also the \textit{first} rigorously mathematically proved examples of these kinds, which have potential applications in geometric analysis, black hole mechanics, numerical relativity, and gravitational wave phenomenology. In fluid dynamics and oceanography, finger-type configurations, such as Saffman-Taylor fingers \cite{ST} and salt fingers \cite{Melvin, Schmitt}, are describing pattern formation for the interface between two fluids in a porous medium and the mixing process of warm salty water with colder fresh water. This mixing process has large-scale consequences for the structure of the ocean. Our constructed  anisotropic apparent horizons are the finger-type interface between the black-hole region and the non-black-hole region. In particular, our multi-valley anisotropic apparent horizon is the counterpart of the $N$-finger solution observed in experiments and numerical simulations in fluid dynamics.} 

\subsection{Motivations}
In 2009, Christodoulou published a breakthrough result. In a 589-page monumental work \cite{Chr:book}, Chrsitodoulou showed that in the Einstein vacuum spacetime a trapped surface can form dynamically from arbitrarily dispersed initial data. A trapped surface is a two dimensional sphere with both future null expansions negative, and it indicates a strong field (black hole) region is formed. Together with the Penrose's incompleteness theorem, a trapped surface also predicts future geodesically incompleteness.   

Later works by Klainerman and Rodnianski \cite{KR:Trapped} and by {\color{black}the first author} \cite{An12} simplified \cite{Chr:book} in a systematical way by introducing signatures for short-pulse and signatures for decay rates. The work by {\color{black}the first author} and Luk \cite{AL} further extended Christodoulou's result to the scale-critical regime. Recently, {\color{black}the first author} \cite{An19} gave an alternative proof of \cite{Chr:book} and \cite{AL} by using only signatures for decay rates and by employing a spacetime scaling argument. 

In these aforementioned works, to form a trapped surface in the gravitational collapse, we require that the initial mass input of each angle of $\mathbb{S}^2$ satisfies a uniform non-zero lower bound and a uniform upper bound. In another word, the characteristic initial data of EVEs are not spherically symmetric,  but they are also not very anisotropic. (The details of anisotropicity will be explained in Theorem \ref{prop-anisotropic-phi-one} and Theorem \ref{prop-anisotropic-supersolution-phi-multi}.) In \cite{KLR}, Klainerman, Luk, and Rodnianski studied the case of fully anisotropic initial data and showed that a trapped surface can also form dynamically. Their initial data allow that the initial mass input of most (almost all) angles of $\mathbb{S}^2$ to be \textit{zero} and along angles only in a small region they have a uniform non-zero lower bound. 

So far, we have a relatively satisfactory theory of trapped surface formation. The next challenging question is to find the outer boundary of these trapped surfaces, which is called the apparent horizon. It is closely tied to the black hole thermodynamics in physics.

In 1970s, Bekenstein and Hawking found that black hole spacetimes are associated with an amazing property. At each fixed time $t$, the area of the black-hole event horizon could be viewed as an entropy of the spacetime (at time $t$) and it is increasing toward the future. In general relativity, we could define two types of boundaries for a black hole region: the event horizon and  the apparent horizon. For dynamical spacetimes, these two horizons are not necessarily the same. To define the event horizon, global information of the whole spacetime up to $t=\infty$ is required.\footnote{ The event horizon is defined as the future boundary of the causal past set of the future null infinity.} On the other side, an apparent horizon could be defined and detected locally, and it is composed of marginally outer trapped surfaces (MOTSs), which is the outer boundary of the trapped region and where the outgoing future null expansion is zero. The apparent horizon is widely used in the mathematical and numerical explorations of gravitational collapse. 

A natural task is to develop black hole thermodynamics along an apparent horizon. This remained to be an open question until 2003. In \cite{AK03, AK}, Ashtekar and Krishnan showed that if a smooth spacelike apparent horizon forms in gravitational collapse, then at each fixed time, one can define the area of the apparent horizon to be an entropy and it is indeed \textit{increasing} toward the future. Hence, the second law of black hole mechanics holds.  

However, a key mathematical step is missing in the argument by Ashtekar and Krishnan. \textit{Does the smooth spacelike apparent horizon exist? Can it emerge in gravitational collapse from general initial data?} These questions are important. Due to physical intuitions, it is conjectured that \textit{the smooth spacelike apparent horizon exists in gravitational collapse and it arises from general (both isotropic and anisotropic) initial data.}  

The above physical thinking motives the first author to prove the existence of the apparent horizon based on the trapped surface formation results obtained in \cite{Chr:book} and the other works aforementioned. By solving a new quasilinear elliptic equation with carefully designed gradient estimate, and by employing a limit argument to the scale-critical result in \cite{AL} established by An and Luk, the first author in \cite{An17} showed that smooth apparent horizon  emerges from not-very-anisotropic characteristic initial data. If the initial data are moreover required to be close to isotropic, the corresponding apparent horizon is indeed spacelike and the second law of black hole mechanics (area law) along the apparent horizon is established. 

In this paper,  we generalize the above conclusion in \cite{An17} to the general and fully-anisotropic scenarios. We adopt a different form of the quasilinear elliptic equation for marginally outer trapped surfaces lying on an incoming null cone. By employing \textit{non-perturbative} analytic tools, such as Moser's iteration, the Harnack inequality, Schauder estimates, and the strong maximal principle, via a method of continuity we show the existence, uniqueness, and smoothness of the apparent horizon arising from general (anisotropic) characteristic initial data and we furthermore prove that the apparent horizon obtained is always \textit{spacelike}. {\color{black}This confirms a local version of cosmic censorship!} Moreover, we construct apparent horizons with \textit{isolated valleys} and demonstrate that the corresponding apparent horizons are also \textit{anisotropic}.    

Marginally outer trapped surfaces are also studied within a 3-dimensional spacelike hypersurface. Interested readers are referred to \cite{AEM, AMS, AM, AG, AK03, AK, E1, E2, GS, Sc, SY, Yau}.

\subsection{The Trapped Surface Formation}
Our mathematical explanation starts from \cite{Chr:book} by Christodoulou. Whether a trapped surface could form in Einstein vacuum spacetime was open for a long time. In a 589-page monumental work \cite{Chr:book}, with a double null foliation Christodoulou {\color{black} provided an affirmative answer.}  

In \cite{Chr:book} characteristic initial data are prescribed along a truncated incoming cone $\Hb_0$, where Minkowskian initial data are prescribed, and a truncated outgoing cone $H_0$. The 2-sphere $S_{0,0}$ is the intersection of these two cones and it is a standard 2-sphere with radius $1$. Along the outgoing null cone $H_0$, we denote by $\chi$ its outgoing null second fundamental form. Define $\chih$ to be the traceless part of $\chi$. Then, $\chih$ will serve as the free-prescribed data along $H_0$ for $0\leq \ub\leq \delta$, where $\d$ is small. The EVEs are to be solved in the colored region in Figure \ref{Figure 1}. 

\begin{figure}[h]
\begin{center}
\begin{tikzpicture}[scale=0.75]

\draw [white](3,-1)-- node[midway, sloped, below,black]{$H_0(u=0)$}(4,0);

\draw [white](2,2)--node [midway,sloped,above,black] {$\Hb_{\delta}(\ub=\delta)$}(4,0);
\draw [white](1,1)--node [midway,sloped, below,black] {$\Hb_{0}(\ub=0)$}(3,-1);
\draw [dashed] (0, 2)--(0, -4);
\draw [dashed] (0, -4)--(4,0)--(2,2);
\draw [dashed] (0,0)--(2,2);

\draw [dashed] (0,-4)--(2,-2);
\draw [dashed] (0,2)--(3,-1);
\draw [very thick] (1,1)--(3,-1)--(4,0)--(2,2)--(1,1);
\fill[yellow!70!red] (1,1)--(3,-1)--(4,0)--(2,2)--(1,1);
\draw [white](1,1)-- node[midway,sloped,above,black]{$H_{u_*}$}(2,2);
\draw [->] (3.3,-0.6)-- node[midway, sloped, above,black]{$e_4$}(3.6,-0.3);
\draw [->] (1.4,1.3)-- node[midway, sloped, below,black]{$e_4$}(1.7,1.6);
\draw [->] (3.3,0.6)-- node[midway, sloped, below,black]{$e_3$}(2.7,1.2);
\draw [->] (2.4,-0.3)-- node[midway, sloped, above,black]{$e_3$}(1.7,0.4);
\end{tikzpicture}
\end{center}
\caption{}
	\label{Figure 1}
\end{figure}

Along $H_0$, setting $\chih$ to be large, Christodoulou designed a short-pulse hierarchy, where different geometric quantities are required to be of different sizes in term of $\d$. This carefully chosen hierarchy is later shown to be preserved in the nonlinear evolution. Thus, being a large data problem in nature, a semi-global existence theorem in the colored region can still be established. We hence could control the dynamics and geometry of EVEs from initial data quite far away to a region where a ``black hole'' is about to form.  Requiring the incoming radiation per unit solid angle is bounded uniformly from below, a trapped surface is guaranteed to form in the colored region of Figure \ref{Figure 1}.

Here is the main theorem\footnote{In Christodoulou's original proof, outgoing initial data are prescribed at past null infinity. Here, we only mention its finite-region version.} in \cite{Chr:book}.
\begin{theorem}[Christodoulou \cite{Chr:book}]\label{Chr.thm}
Assume $\Hb_0$ coincides with a backwards light cone in the Minkowskian spacetime for $0\leq u\leq 1$. Then for every $B>0$ and $u_*\leq 1$, there exists a $\de=\de(B,u_*)>0$ sufficiently small with the following property. If the initial $\chih_0$, prescribed on $H_0$ for $0\leq \ub\leq \de$, satisfies
\begin{equation}\label{Chr.upper.bound}
\sum_{i\leq 5,\,j\leq 3}\de^{\frac 12+j}\|\nab_{e_4}^j \nab^i\chih_0\|_{L^\infty_{\ub}L^2(S_{0,\ub})}\leq B, 
\end{equation}
where $e_4$ and $\nab$ are outgoing null vector and angular derivative on a 2-sphere $S_{u,\ub}$ respectively, then the solution to EVEs remains regular in $0\leq u\leq u_*$, $0\leq \ub\leq \de$. In addition, if the initial data further verify the lower bound
\begin{equation}\label{Chr.lower.bound}
\inf_{\o\in \mathbb{S}^2} \int_0^{\de}  |\chih_0|^2(\o, \ub')\,d\ub' \geq M_* > 2(1-u_*),
\end{equation}
then for $\de$ sufficiently small (depending on $B$, $u_*$ and $M_*$), the sphere $S_{u_*,\de}$ is a trapped surface.\footnote{The initial data constructed in \cite{Chr:book} satisfy both \eqref{Chr.upper.bound} and \eqref{Chr.lower.bound} at the same time. Moreover, the initial data can be chosen to obey 
$\inf_{\o\in \mathbb{S}^2} \int_0^{\de} |\chih|^2(\o, 0, \ub')\,d\ub'<2.$
Thus for $\de$ sufficiently small, it can be proved that the initial hypersurface $H_0$ does not contain any trapped surfaces.} 
\end{theorem}

After \cite{Chr:book}, various simplifications and extensions of Theorem \ref{Chr.thm} have been obtained.  Interested readers are also referred to \cite{Dafermos, KR:Scarred, L-R:Interaction}.  In all these works, similar upper and lower bounds as in \cite{Chr:book} are used and trapped surfaces formed are of radius close to $1$.

An important work \cite{KLR} relaxed the lower bound assumption in \cite{Chr:book} to allow for fully anisotropic initial data.

\begin{theorem}[Klainerman-Luk-Rodnianski \cite{KLR}]
Assume that the initial data for EVEs satisfy the condition \eqref{Chr.upper.bound} in Theorem \ref{Chr.thm} and the lower bound
\begin{equation*}
\sup_{\o\in \mathbb{S}^2} \int_0^{\de} |\chih_0|^2(\o, \ub')\,d\ub' \geq M_* > 0.
\end{equation*}
Then, for sufficiently small $\de$, a trapped surface forms in the future of the initial data.
\end{theorem}

To obtain this result, in \cite{KLR} the authors employed a deformation of the double null foliation and solved a quasilinear elliptic inequality. See also a geometric approach later by Le in \cite{Le}.

\smallskip 

Another work \cite{AL} relaxed the upper bound assumption in \cite{Chr:book}. {\color{black}Two large parameters $a,b$ satisfying $1\ll b\leq \at \leq \d^{-\f12}$ are introduced.} The authors considered the colored region in Figure \ref{Figure 1} with $u_*=1-b\d \at$ and improved \cite{Chr:book} to the scale-critical regime.

\begin{theorem}[An-Luk \cite{AL}] \label{thm1.3}
Prescribe Minkowskian data along 
the initial incoming null hypersurface $\Hb_0$ with $0\leq u \leq 1$. Given a fixed $\delta$, for every $B$, there exist $a_0=a_0 (B)$ and $b_0=b_0(B)$ sufficiently large with the following property. Choose any $a$ and $b$ verifying $a_0 \leq a \leq \delta^{-1}$ and $b_0\leq b \leq a^{\frac12} \leq \delta^{-\frac12}$.  Along $H_0$, with $0\leq \ub \leq \d$, if initially $\chih_0$ satisfies
\begin{equation}\label{AL.upper.bound}
\sum_{i\leq 5, j\leq 3}\delta^{j} a^{-\frac12}\|\nab^{j}_{e_4}\nab^{i}\chih_{0}\|_{L^{\infty}_{\ub}L^2(S_{0,\ub})}\leq B, 
\end{equation}
then the solution to EVEs \eqref{EVEs} is regular in the spacetime region $0\leq u \leq 1-b\delta \at$ and $0\leq \ub \leq \d$. In addition, if the initial data also verify a lower bound
\begin{equation*}
\inf_{\o\in \mathbb{S}^2}\int_0^{\delta}|\chih_0|^2(\o, \ub') d\ub'\geq 4\delta a,
\end{equation*}
then the $2$-sphere $S_{1-\d a, \d}$ is a trapped surface.  
\end{theorem}

{\color{black}Both parameters $a$ and $b$ play key roles in the above theorem. Parameter $a$ is a large universal constant and using it we prove the scale-critical result\,--\,the formation of a trapped surface with radius $\d a$. The parameter $b$ is employed to guarantee that all the estimates in terms of the power of $a$ are sharp. In this paper, we also find a geometric meaning of $b$: it enables a room of anisotropicity in the estimates. This is a crucial new observation.\footnote{See also Remark \ref{b anisotropicity}.}    
}

To prove Theorem \ref{thm1.3}, smooth characteristic initial data with a different hierarchy to \cite{Chr:book} are designed. Let $a$ be a large universal constant. To form a trapped surface, it is assumed that the $H^{\f32}$-norm of the metric is of size $\at$ and all other $H^s$-norms are small for $s<\f32$. Note that the homogeneous Sobolev space $\dot{H}^{\f32}(\mathbb{R}^3)$ is the scale-critical space for the $3+1$ dimensional EVEs. Hence, \cite{AL} is also the first scale-critical result for EVEs in the $3+1$ dimensions.

Later, via a different method by only using signature for decay rates, requiring $b=\at$,
in \cite{An19} with less than 60 pages the first author reproved and extended Theorem \ref{thm1.3} to past null infinity with sharper bounds. A spacetime rescaling argument is crucially used.

\subsection{Apparent Horizons Arising from Non-Anisotropic Initial Data} 

In \cite{An17}, the first author considered the following characteristic initial value problem for EVEs in the colored region in Figure \ref{Figure 3}.

\begin{figure}[h]
\begin{center}
\begin{tikzpicture}[scale=0.86]
\draw [white](0,2.3)--node [midway,sloped, below,black] {$O$}(-0.5,2.3);
\draw [white](0.6,1.9)--node [midway,sloped, above,black] {$S_{1-\tilde{\ub} a, \tilde{\ub}}$}(1.5, 2.8);
\fill[yellow!70!red](0,2)--(1,1)--(2,2)--(0,2);
%\draw (1.5,-0.5) node[very near end, sloped,below]{$H_{u_{\infty}}(u=u_{\infty})$}--(2,0); 
\draw [white](3,-1)-- node[midway, sloped, below,black]{$H_0(u=0)$}(4,0);
\draw [white](2.1,2.1)--node [midway,sloped, above,black] {$S_{1-\delta a, \delta}$}(2.8,2.8);
\draw [white](2,2)--node [midway,sloped,above,black] {$\Hb_{\delta}(\ub=\delta)$}(4,0);
\draw [white](1,1)--node [midway,sloped, below,black] {$\Hb_{0}(\ub=0)$}(3,-1);
\draw [dashed] (0, 2)--(0, -4);
\draw [dashed] (0, -4)--(4,0)--(2,2);
\draw [dashed] (0,0)--(2,2);
%\draw [dashed] (0,1)--(1.5,2.5);
\draw [dashed] (0,-4)--(2,-2);
\draw [thick] (0,2)--(3,-1);
%\draw [thick] (1,1)--(0.5,1.5)--(1.5,2.5)--(2,2)--(1,1);
\fill[yellow!70!red] (1,1)--(3,-1)--(4,0)--(2,2)--(1,1);
%\fill[yellow!30!red] (1, 2)--(3.5, -0.5)--(3,-1)--(0.5,1.5)--(1,2); 
\fill[yellow!30!red] (0.5, 2)--(3.25,-0.75)--(3,-1)--(0.25, 1.75)--(0.5, 2);
%\fill[yellow!30!red](1,1)--(0.5,1.5)--(1.5,2.5)--(2,2)--(1,1); 
%\fill[red!40!](1, 2)--(0,2)--(0.5, 1.5)--(1, 2); 
\draw [white] (1.2, 1.3)--node [midway,sloped,above,black] {$\Hb_{\tilde{\ub}}(\ub=\tilde{\ub})$}(3.25,-0.75); 

\draw [thick] (1,1)--(3,-1)--(4,0)--(2,2)--(1,1);

%\draw [thick] (1, 2)--(3.5, -0.5); 
%\draw [thick] (0.5,1.5)--(1,2); 
\draw [thick] (0.5, 2)--(3.25,-0.75); 
\draw [thick] (0.25, 1.75)--(0.5, 2); 
\draw[dashed] (0, 2)--(2.1, 1.9);
\end{tikzpicture}
\end{center}
\caption{}\label{Figure 3} 
\end{figure}

\begin{enumerate}[label=(\subscript{A}{{\arabic*}})]

\item The initial incoming null hypersurface $\Hb_0$ coincides with a backward light cone in the Minkowski space with $0\leq u \leq 1$, and $S_{0,0}$ is the standard $2$-sphere with radius $1$. 

\item Let $\delta$, $B$, and $a$ be positive constants. Along $H_0$, for $0\leq \ub \leq \delta$, prescribe $\chih_0$ such that
\begin{equation*}
\sum_{i\leq 5,  j\leq 3} a^{-\frac12}\|{\ub}^{j}\nab^{j}_{e_4}\nab^{i}\chih_0\|_{L^{\infty}_{\ub}L^2(S_{0,\ub})}\leq B.
\end{equation*}

\end{enumerate}

\begin{theorem}[An \cite{An17}]\label{thm1.5} 
Given $\delta>0$, for every $B>0$, there exist $a_0=a_0 (B)$ and $b_0=b_0(B)$ sufficiently large with the following property. Choose any $a$ and $b$ satisfying $a_0 \leq a \leq \delta^{-1}$ and $b_0\leq b \leq a^{\frac12} \leq \delta^{-\frac12}$, and assume the characteristic initial data satisfy $(A_1)$ and $(A_2)$ above. Then, the EVEs \eqref{EVEs} admit a unique solution in the region $0< \ub \leq \d$ and $0\leq u \leq 1-b\ub \at$.

{\it 1.} Moreover, assume
\begin{equation}\label{eq-definition-f}
 \int_0^{\ub}|\chih_0|^2 (\o,\ub')d\ub'=f(\o, \ub)\ub a \quad  \mbox{for each }  0<\ub\leq \d,
\end{equation}
where $f(\o,\ub)$ is a smooth function  in $\mathbb S^2\times (0,\delta]$ satisfying 
\begin{equation}\label{eq-assumption-f-near 1}20/21\leq f(\o, \ub)\leq 22/21.\end{equation}
Then along each $\Hb_{\ub}$ $(0< \ub \leq \d)$, there \textit{exists a unique} MOTS $M_{\ub}$. 

{\it 2.} Furthermore, assume, for any $i,j\ge 0$,
\begin{equation*}
a^{-\frac12}\|{{\color{black}\ub}}^{j}\nab^{j}_{e_4}\nab^{i}\chih_{0}\|_{L^{\infty}_{\ub}L^2(S_{0,\ub})}\leq B.
\end{equation*}
Then the collection of MOTSs $\{M_{\ub}\}$ forms {\color{black} a smooth apparent horizon.  

{\it 3.} In addition, assume 
\begin{equation}\label{eq-assumption-f-near 1-enhanced}1-\f{1}{c}\leq f(\o, \ub)\leq 1+\f{1}{c},\end{equation} with {\color{black} sufficiently large} $c$. 
Then, the apparent horizon is \textit{spacelike}.} 
\end{theorem}

Here we outline the approach of \cite{An17}.

\smallskip

In the original double null foliation, on $S_{1-\ub a, \ub}$ the outgoing null expansion ($\tr\chi|_{S_{1-\ub a, \ub}}$) is negative pointwisely. Hence, $S_{1-\ub a, \ub}$ is a trapped surface. For fixed $\ub$, to find a 2-sphere $M_{\ub}=\{(u, \ub, \o): u=1-R(\o,\ub)\}$ where the  
outgoing null expansion vanishes for all points (thus $M_{\ub}$ is a marginally outer trapped surface), we need to deform the foliation on $\underline{H}_{\ub}$. 
In this setting,  the null expansion $\tr\chi'$ has the following expression: 
\begin{equation}\label{deformation equation}
\begin{split}
\tr\chi'&=\tr\chi-2\O\D'_{M_{\ub}} R(\o, \ub)-4\O\eta_a\nab^a R(\o,\ub)\\
&\qquad-\O^2 \tr\chib |\nab R(\o,\ub)|^2-8\O^2\omb|\nab R(\o,\ub)|^2.
\end{split}
\end{equation}
Here,  $\D'_{M_{\ub}}$ and $\nab$ are the Laplacian operator and the angular derivative to the unknown 2-sphere $u=1-R(\o, \ub)$ at the point $(\o, 1-R(\o,\ub), \ub)$; $\Omega$ is a geometric quantity satisfying $g(\f{\partial}{\partial u}, \f{\partial}{\partial \ub})=-2\O^2$; and 
$\eta_a$, $\chib_{bc}$, and $\omb$ are certain Ricci coefficients of $g$.

On fixed $\Hb_{\ub}$, when there is no danger of confusion, for simplicity we use $R(\o)$ to substitute $R(\o,\ub)$. In \cite{KLR}, Klainerman, Luk, and Rodnianski derived (\ref{deformation equation}) and first solved the quasilinear elliptic inequality {\color{black} $\tr\chi'<0$, i.e.,}
\begin{equation*}
\begin{split}
\tr\chi-2\O\D'_{M_{\ub}} R(\o)-4\O\eta_a\nab^a R(\o)-\O^2 \tr\chib |\nab R(\o)|^2-8\O^2\omb|\nab R(\o)|^2
<0.
\end{split}
\end{equation*}
We proceed to solve the quasilinear elliptic equation $\tr\chi'=0$, i.e., 
\begin{equation}\label{1.12}
\begin{split}
&\D'_{M_{\ub}} R(\o)+2\eta_b\nab^b R(\o)+\f12 \O \tr\chib |\nab R(\o)|^2+4\O \omb |\nab R(\o)|^2-\f{\O^{-1}}{2} \tr\chi=0.
\end{split}
\end{equation}

A key step in \cite{An17} is to derive $C^1$ (gradient) estimates  by Bochner's formula. Define
\begin{equation}\label{h(R)}
h(R)=1+\f{8}{\ub^2 a^2}\Big(R-\f{\ub a}{2}\Big)^2.
\end{equation}
By a careful analysis based on Bochner's formula and the explicit form of \eqref{1.12}, in \cite{An17} we obtain
\begin{equation*}
\begin{split}
\D'_M \big( h(R)|\nab R|^2 \big)\geq& \f{1}{\ub^2 a^2}|\nab R|^4+\f{1}{8R^2}|\nab R|^2-\f{1}{R^2}\cdot o(1).
\end{split}
\end{equation*}
Note that not only the sign in front of $|\nab R|^4$ is \textit{positive}, but the sign in front of $|\nab R|^2$ is also \textit{positive}! 
{\color{black}Then, via the maximum principle} it implies 
\begin{equation}\label{C1 key estimate}
|\nab R| (\o)\ll1 \quad \mbox{for all }  \o\in \mathbb S^2. 
\end{equation}
{\color{black} We point out that the exact form of $h(R)$ in \eqref{h(R)} plays a \textit{vital} role in the derivation of \eqref{C1 key estimate}. The coefficient $8$ is closely related to the assumption \eqref{eq-assumption-f-near 1}. The gradient estimate \eqref{C1 key estimate} is essential in subsequent studies in \cite{An17}. It demonstrates that the solution $R(\o)$ is close to a constant under the condition \eqref{eq-assumption-f-near 1}.
Physically, the resulting MOTS is close to a sphere.}

To solve \eqref{1.12}, we need to analyze the corresponding linearized equation. 
Denote by $G(R)$ the expression on the left-hand side of \eqref{1.12}, for $R=R(\o)$. 
By a lengthy computation, the linearized operator $G_{R}(R)$ is given by 
\begin{equation}\label{main linearization equation in An17}
G_{R}(R)[W]
=\D'_{M}W-\f{1}{R}\big(1-|\nab R|^2\big)\big(1+o(1)\big)W.
\end{equation}
By \eqref{C1 key estimate}, the coefficient of $W$ in \eqref{main linearization equation in An17} is \textit{negative}, and thus  $G_{R}(R)[W]$ is invertible. 
Such invertibility implies the openness in the method of continuity to solve \eqref{1.12}, or to prove the existence of MOTS.
Similarly, \eqref{C1 key estimate} plays a crucial role in the proof of other assertions in Theorem \ref{thm1.5}. 

We now make a simple but an important observation. 
The coefficient of zero order term in \eqref{main linearization equation in An17} does not have a desired sign if $|\nabla R|$ is not small as asserted in \eqref{C1 key estimate}. Different methods are needed to solve \eqref{1.12} for the general anisotropic initial data! {\color{black}Also to prove the spacelikeness of the apparent horizon, as required in (\ref{eq-assumption-f-near 1-enhanced}), $f(\o,\ub)$ has to be very close to $1$, which is the close-to-homogenous case. To deal with anisotropic scenarios, novel arguments are requested. }

\subsection{Apparent Horizons Arising from General Anisotropic Initial Data}

In this paper, we will study the equation \eqref{1.12} under a much relaxed assumption on the function $f$ introduced in \eqref{eq-definition-f}. We prove the following result.  

\begin{theorem}\label{main thm} 
With characteristic initial data satisfying $(A_1)$, $(A_2)$ and with the choice of parameters $a$, $b$
as in Theorem \ref{thm1.5}, the EVEs (\ref{EVEs}) admit a unique solution in the region $0< \ub \leq \d$ and $0\leq u \leq 1-b\ub \at$. 

{\it 1.} Assume
\begin{equation*}
 \int_0^{\ub}|\chih_0|^2 (\o, \ub') d\ub'=f(\o, \ub)\ub a \quad  \mbox{for each}  \quad 0<\ub\leq \d,
\end{equation*}
where {\color{black} $f(\o,\ub)$ is a smooth function on $\mathbb S^2\times (0,\delta]$} satisfying 
\begin{equation}\label{eq-assumption-f-0-1}
0\leq f\leq 1\quad\text{on }\mathbb S^2\times (0,\delta],\footnote{The condition $0\leq f(\ub,\o)\leq 1$ is very general. With the same method as in this paper, we could replace $1$ by any fixed positive constant $M$. For simplicity of presentation, we always take $M=1$.} 
\end{equation}
and, for any $\ub\in (0,\delta]$,  
\begin{equation}\label{eq-assumption-lower-bound-f}
f(\cdot, \ub)\ge m\quad\text{on }B_{p}(\e),
\end{equation}
for some point $p\in \mathbb S^2$ and some constants $m\in (0,1)$ and $\e\in (0,\pi/2)$. Then in a large region of each $\Hb_{\ub}$ $(0< \ub \leq \d)$,  there \textit{exists a unique} MOTS $M_{\ub}$. (See the details of considered regions in Theorem \ref{theorem-existence-solutions}.) 

{\it 2.} In addition, if requiring, for any $i,j\ge 0$, 
\begin{equation}\label{eq-main-assumption-extra}
a^{-\frac12}\|{{\color{black}\ub}}^{j}\nab^{j}_{e_4}\nab^{i}\chih_{0}\|_{L^{\infty}_{\ub}L^2(S_{0,\ub})}\leq B \quad \text{and} \quad |\ub\partial_{\ub}f(\o,\ub)|\leq a^{-\f13},
\end{equation}
then the collection of MOTS $\{M_{\ub}\}$ forms a \textit{smooth} and \textit{spacelike} hypersurface, i.e., the dynamical horizon.  
\end{theorem}

{\color{black}Since we prescribe Minkowskian data for $\ub\leq 0$, we have $f(\o,\ub)=0$ when $\ub=0$. The vertex $O$ in Figure \ref{Figure 3} with coordinate $(u,\ub)=(1,0)$ is a singular point. Our above result implies that except at the vertex $O$, the collection of MOTS $\{M_{\ub}\}$ forms a smooth apparent horizon with $0<\ub\leq \d$.}

Furthermore, we will provide explicit constructions of  apparent horizons with isolated valleys, and by designing  and employing a \textit{new anisotropic lower bound} we prove that these apparent horizons are also anisotropic.   

In the below, we also use the function $\p=\p(\o,\ub)$ defined via $R(\o,\ub)=\ub a e^{-\p(\o,\ub)}$.

\begin{theorem}\label{prop-anisotropic-supersolution-phi-multi-intro} 
Let $p_1, \cdots, p_n$ be fixed $n$ points on $\mathbb S^2$, 
$\e$ be a positive constant 
much less than the distance between any two distinct $p_i$ and $p_j$ for $i, j=1, \cdots, n$, 
and $f$ be a smooth function on $\mathbb S^2$, such that $0\le f\le \nu(n)$ and
$$f\ge \nu(n)/2\,\text{ on }\cup_{i=1}^nB_{p_i}(\e/2),
\quad\quad f=0\,\text{ on }\mathbb S^2\setminus \cup_{i=1}^nB_{p_i}(\e),$$
for some constructed constant $\nu(n)\in (0,1)$ depending only on $n$. 
Then for each ${\ub}$, 
the equation $\mathcal S(\phi, \ub)=0$ admits a solution  $\phi_{\e}(\o, \ub)$  satisfying, 
for any $r_0\gg \e$ and any $i=1, \cdots, n$,  
\begin{equation}\label{eq-multi-valley-intro}
\phi_{\e}(p_i, \ub)-\inf_{\o\in\mathbb S^2\setminus \cup_{i=1}^nB_{p_i}(r_0)}\phi_{\e}(\o, \ub)
\le \f1{2n}\log\e+C,\end{equation}
where $C$ is a positive constant, independent of $\e$ and $\ub$. 
\end{theorem}

The estimate \eqref{eq-multi-valley-intro} demonstrates that the value of $\phi_{\e}$ at each $p_i$ is significantly smaller than 
values away from those $p_i$'s. In this sense, we say $\p_{\e}$ has a valley at $p_i$. 
The constant $\nu(n)$ can be computed explicitly. For example, for the case of a single valley, we can take $\nu(1)=1$.
The characteristic initial data leading to apparent horizons with valleys as stated in Theorem \ref{prop-anisotropic-supersolution-phi-multi-intro}
satisfy the assumptions in Theorem \ref{main thm}. The corresponding apparent horizons are henceforth spacelike as well.

\smallskip

Figure \ref{Figure MOTS} demonstrates a finger-shaped MOTS $M_{\ub}(\o)$ along $\Hb_{\ub}$ with two valleys. To prove Theorem \ref{prop-anisotropic-supersolution-phi-multi-intro} for $M_{\ub}(\o)$, in Section \ref{anisotropic examples} we construct explicit upper and lower barriers with valleys. And by the arguments in Section \ref{sec-existence-solutions}, within the region of consideration, we prove that a unique MOTS $M_{\ub}(\o)$ exists along $\Hb_{\ub}$ and it is bounded by these two barriers. The anisotropicity of $M_{\ub}(\o)$ can thus be calculated via computing the quotient of $R(\o,\ub)$ at points $A$ and $B$ or $A$ and $C$.

\begin{figure}[h]
\begin{center}
\begin{tikzpicture}[scale=0.85]
%%  \draw[red, densely dashed] (-1.36,1.46) arc [start angle = 170, end angle = 10,
%%   x radius = 13.8mm, y radius = 3.6mm];
%%  \draw[red] (-1.29,1.52) arc [start angle=-200, end angle = 20,
%%    x radius = 13.75mm, y radius = 3.15mm];
%%    \draw[fill] (-1.36,1.46) circle [radius=0.15];
\draw (0,0.5) ellipse (0.875cm and 0.29cm);
%%\draw (0,0) ellipse (1cm and 0.3cm);
\draw (0,-4) ellipse (2cm and 0.6cm);
%%\draw[fill] (-1,0.5) circle [radius=0.15];
%%\draw[fill] (-2,-3) circle [radius=0.15];
%%\draw[fill] (1,0.5) circle [radius=0.15];
%%\draw[fill] (2,-3) circle [radius=0.15];
%%\draw[fill] (-1.25,-0.375) circle [radius=0.15]; %useful
%%\draw[fill] (-1.5,-1.25) circle [radius=0.15]; %useful
%%\draw[fill] (-1.75,-2.125) circle [radius=0.15]; %useful
%%\draw [white](1.5, -1)-- node[midway, sloped, above,black]{$P \mbox{ where } \o=\o_0$}(4, -1);
\draw [white](1.2, -0.3)-- node[midway, sloped, above,black]{$\Hb_{\ub}$}(1.5, -0.3);
\draw [white](2.8, -2.2)-- node[midway, sloped, above,black]{upper barrier}(2.9, -2.2);
\draw [white](3.1, -3.3)-- node[midway, sloped, above,black]{MOTS $M_{\ub}(\o)$}(3.3, -3.3);
\draw [white](3.1, -3.6)-- node[midway, sloped, above,black]{lower barrier}(3.2, -3.6);
\draw [white] (1.85, -2.2)-- node[midway, sloped, below, black]{{\footnotesize A}} (1.9, -2.2);
\draw [white](2.2, -4.3)-- node[midway, sloped, above,black]{$p_1$}(2.4, -4.3);
\draw [white](-2.2, -4.3)-- node[midway, sloped, above,black]{$p_2$}(-2.4, -4.3);
\draw [white](-0.4, -5.2)-- node[midway, sloped, above,black]{$\mathbb{S}^2$}(0.4, -5.2);
\draw[thin] (-0.875,0.5)--(-2,-4);%%
\draw[thin] (0.875,0.5)--(2,-4);%%
\draw[thin] (-1,0)--(-2,-4);
%%\draw[fill] (1.25,-0.375) circle [radius=0.15]; %useful
%%\draw[fill] (1.5,-1.25) circle [radius=0.15]; %useful
%%\draw[fill] (1.75,-2.125) circle [radius=0.15]; %useful
\draw[thin] (1,0)--(2,-4);
%\draw [white](-3.5, -1.7)-- node[midway, sloped, above,black]{$u=1-R_1(\o)$}(-2.5, -1.7);
%%\draw [white](-4,-0.7)-- node[midway, sloped, above,black]{$u=1-R_1(\o)-k\e$}(-1.8,-0.7);
%%\draw [white](-3.5,-2.7)-- node[midway, sloped, above,black]{$u=1-R_2(\o)$}(-2.5,-2.7);
%\draw [white](3,-1.31)-- node[midway, sloped, above,black]{$u=1-R_2(\o)$}(2.2,-1.31);
%\draw [white](-1.5, -2.5)-- node[midway, sloped, above,black]{$P$}(-2.5, -2.5);
%\draw [white](-3, -3)-- node[midway, sloped, above,black]{$u=1-R_1(\o)-k\e$}(-4, -3);
%%\draw [thick] (-1,0) to [out=10,in=150] (1.125,-0.5); 
%%\draw [thick] (1.125,-0.5) to [out=180,in=-35] (-1,0); 
%%\draw [thick] (-1.125,-0.5) to [out=10,in=150] (1.25,-1); %%
%%\draw [thick] (1.25,-1) to [out=180,in=-35] (-1.125,-0.5);%%
%%\draw [thick] (-1.25,-1) to [out=10,in=150] (1.375,-1.5);
%%\draw [thick] (1.375,-1.5) to [out=180,in=-35] (-1.25,-1);
%%\draw [thick] (-1.375,-1.5) to [out=10,in=150] (1.5,-2);%
%\draw [thick] (1.5,-2) to [out=180,in=-35] (-1.375,-1.5);% 

\fill[white!90!black] (-0.5, -0.5) to [out=50,in=140] (0.6, -0.5) to [out=315,in=130] (1.6,-2.55) to [out=0, in=-60]  (1.625, -2.5) to [out=115, in=315] (0.6, -0.3) to [out=140, in=50] (-0.5, -0.3) to [out=230,in=65] (-1.5,-2) to [out=-20, in=180] (-1.45, -2.02) to [out=45, in=230] (-0.5, -0.5);
\draw [thick] (-0.5, -0.5) to [out=50,in=140] (0.6, -0.5); 
\draw [thick] (0.6,-0.5) to [out=315,in=130] (1.6, -2.55); 
\draw [thick] (1.6,-2.55) to [out=0, in=-60] (1.625, -2.5); %new
%\draw [dashed] (0.6,-0.3) to [out=315,in=115] (1.625, -2.5); 
\draw [dashed] (1.625, -2.5) to [out=115, in=315] (0.6, -0.3);
\draw [dashed] (0.6, -0.3) to [out=140, in=50] (-0.5, -0.3);
%\draw [dashed] (-0.5, -0.3) to [out=50,in=140] (0.6, -0.3); 
\draw [dashed] (-0.5,-0.3) to [out=230,in=65] (-1.5,-2); 
%\draw [dashed] (0.6,-0.3) to [out=315,in=115] (1.625, -2.5); 
%\draw [thick] (-0.5, -0.5) to [out=230, in=45] (-1.5, -2); 
\draw [thick] (-1.5, -2) to [out=-20, in=180] (-1.45, -2.02); %new
\draw [thick] (-1.45, -2.02) to [out=45, in=230] (-0.5, -0.5);

\fill[white!70!black] (0.4, -1) to [out=120,in=50] (-0.4, -1) to [out=230,in=15] (-1.575, -2.55) to [out=180, in=-20]  (-1.625,-2.5) to [out=50, in=230] (-0.5, -0.7)  to [out=50,in=120] (0.5, -0.7) to [out=310,in=135] (1.75, -3) to [out=30, in=-30] (1.7, -3.05) to [out=160,in=300] (0.4, -1);
%\draw [thick] (-0.4, -1) to [out=50,in=120] (0.4, -1); 
\draw [thick] (0.4, -1) to [out=120,in=50] (-0.4, -1); 
\draw [thick] (-0.4,-1) to [out=230,in=15] (-1.575,-2.55); 
\draw [thick] (-1.575, -2.55) to [out=180, in=-20] (-1.625, -2.5); %new
\draw[dashed] (-1.625, -2.5) to [out=50, in=230] (-0.5, -0.7);
%\draw [dashed] (-0.5,-0.7) to [out=230,in=50] (-1.625,-2.5); 
%\draw [thick] (0.4,-1) to [out=300,in=160] (1.75, -3); 
\draw [dashed] (-0.5, -0.7) to [out=50,in=120] (0.5, -0.7); 
%\draw [dashed] (-0.5,-0.7) to [out=230,in=50] (-1.625,-2.5); 
\draw [dashed] (0.5,-0.7) to [out=310,in=135] (1.75, -3); 
\draw [thick] (1.75,-3) to [out=30, in=-30] (1.7, -3.05); %new
\draw [thick] (1.7,-3.05) to [out=160,in=300] (0.4, -1); 

\fill[white!90!black] (0.2, -1.5) to [out=120,in=50] (-0.2, -1.5) to [out=230,in=10] (-1.8,-3.55) to [out=180, in=-20] (-1.875,-3.5) to [out=40, in=230] (-0.2, -1.1) to [out=50,in=120] (0.2, -1.1) to [out=310,in=155] (1.875, -3.5) to [out=30, in=-30] (1.8, -3.55) to [out=170,in=300] (0.2, -1.5);
%\draw [thick] (-0.2, -1.35) to [out=50,in=120] (0.2, -1.35); 
\draw [thick] (0.2, -1.5) to [out=120,in=50] (-0.2, -1.5); 
%\draw [thick] (-0.2,-1.5) to [out=230,in=0] (-1.875,-3.5); 
\draw [thick] (-0.2,-1.5) to [out=230,in=10] (-1.8,-3.55); 
\draw [thick] (-1.8, -3.55) to [out=180, in=-20] (-1.875, -3.5); %new
%\draw [dashed] (-0.2,-1.1) to [out=230,in=40] (-1.875,-3.5); 
\draw [dashed] (-1.875, -3.5) to [out=40, in=230] (-0.2, -1.1);
\draw [dashed] (-0.2, -1.1) to [out=50,in=120] (0.2, -1.1); 
\draw [dashed] (0.2,-1.1) to [out=310,in=155] (1.875, -3.5); 
%\draw [thick] (0.2,-1.35) to [out=300,in=170] (1.875, -3.5); 
\draw [thick] (1.875,-3.5) to [out=30, in=-30] (1.8, -3.55); %new
\draw [thick] (1.8, -3.55) to [out=170,in=300] (0.2, -1.5); 
%\draw [dashed] (-0.2,-1.1) to [out=230,in=40] (-1.875,-3.5); 

\draw [white](-0.05, -0.95)-- node[midway, sloped, below,black]{{\footnotesize B}}(0.05, -0.95);
\draw [white](-0.05, -1.91)-- node[midway, sloped, above,black]{{\footnotesize C}}(0.05, -1.91);

%\draw [thick] (-1.625,-2.5) to [out=50,in=190] (1.25,-1);%%%
%\draw [thick] (1.25,-1) to [out=230,in=15] (-1.625,-2.5);%%%
%\draw [thick] (-1.625,-2.5) to [out=10,in=150] (1.75,-3);%%%%%
%\draw [thick] (1.75,-3) to [out=190,in=-35] (-1.625,-2.5);%%%%%
%%\draw [thick] (-1,0) to [out=10,in=10] (1,0);
%%\draw [white](-3, -0.2)-- node[midway, sloped, below,black]{$P$}(3, -0.2);
%%\draw [white](-3, -2)-- node[midway, sloped, below,black]{${\color{red}D_0}$}(3, -2);
%%\draw [white](-3, 0)-- node[midway, sloped, above,black]{$r(u,v)=0$}(3, 0);
%%\draw [white](-4, -1)-- node[midway, sloped, above,black]{$u=U$}(-3, 0);
%%\draw [white](3, 0)-- node[midway, sloped, above,black]{$v=V$}(4, -1);
%%\draw [white](0, -5)-- node[midway, sloped, below,black]{$u=U'$}(4, -1);
%%\draw [white](0, -5)-- node[midway, sloped, below,black]{$v=V'$}(-4, -1);
%%\draw [white](0, 0)-- node[midway, sloped, above, black]{$u=U_0$}(-2.8,-2.8);
%%\draw [white](0, 0)-- node[midway, sloped, above, black]{$v=V_0$}(2.8,-2.8);
%%\draw[thick] (0,0)--(-2.5,-2.5);
%%\draw[thick] (0,0)--(2.5, -2.5);
%%\draw [thick] (-3,0) to [out=10,in=-170] (0,0);
%%\draw [thick] (0,0) to [out=10,in=-170] (3,0);
%%%\draw [thick] (-3,0)--(3,0);
%%\draw [thick] (-4,-1)--(-3,0);
%%\draw [thick] (3,0)--(4,-1);
%%\draw [thick] (-4,-1)--(0,-5);
%%\draw [thick] (4,-1)--(0,-5);
\draw[fill] (0,-1.02) circle [radius=0.05];
\draw[fill] (0,-1.43) circle [radius=0.05];
\draw[fill] (1.635, -2.55) circle [radius=0.05];
\end{tikzpicture}
\end{center}
\caption{}\label{Figure MOTS}
\end{figure}

Theorem \ref{main thm} further infers that \textit{smooth and spacelike apparent horizons arises from general and fully anisotropic initial data.} This certifies the assumption of the black hole thermodynamics theory, {\color{black}proves a folklore conjecture on gravitational collapse, and confirms a local version of cosmic censorship! Moreover, our single and multi-valley anisotropic apparent horizons constructed in Theorem \ref{prop-anisotropic-supersolution-phi-multi-intro} are of finger-type shapes and are the \textit{first} rigorously mathematical examples of these kinds. They are the counterpart of the finger and $N$-finger solutions in fluid dynamics \cite{ST} and oceanography \cite{Melvin, Schmitt} observed in experiments and numerical simulations. Our mathematical constructions have potential applications in geometric analysis, black hole mechanics, numerical relativity and gravitational wave phenomenology.}  

Following geometric computations in Section 11 of \cite{An17}, the spacelike apparent horizons constructed in Theorem \ref{main thm} also satisfy the following property. 

\begin{proposition}
\label{prop entropy}
Let $M_{\ub}$ be as in Theorem \ref{main thm}. Then, 
$$\lim_{\ub\rightarrow 0}\mbox{Area}(M_{\ub})=0, \quad \mbox{and} \quad \mbox{Area}(M_{\ub'})>\mbox{Area}(M_{\ub})  \mbox{ for } \ub'>\ub.$$
\end{proposition}

Inspired by the works of Ashtekar and Krishnan in \cite{AK03, AK}, we call the above proposition \textit{the area law} along the dynamical horizon. Taking $\ub$ as a time parameter, $\text{Area}(M_{\ub})$ could be viewed as the entropy along the apparent horizon at $\ub$.

\subsection{An Overview of the Proof of Theorem \ref{main thm}}
Under the assumptions \eqref{eq-assumption-f-0-1} and \eqref{eq-assumption-lower-bound-f}, solutions are not expected to be close to constant. In fact, anisotropic solutions do exist according to Theorem \ref{prop-anisotropic-supersolution-phi-multi-intro}. In the present far-away from close-to-isotropic case, the carefully designed gradient estimate as in \cite{An17} does not apply. We will employ more robust and powerful analytic methods to derive gradient estimates or, more generally, the $C^{2, \alpha}$-estimates. These methods are significantly different from those employed in \cite{An17}. Moreover, the gradient estimate in this paper is not as restrictive as the one established in \cite{An17}. Due to the more general form of our gradient estimate, we have to employ entirely different methods to study the invertibility of the linearized equations and the smooth dependence of solutions on parameters. 

Our proof of Theorem \ref{main thm} is non-perturbative in nature. Moser's iteration, the Harnack inequality, Schauder estimates, and the strong maximum principle are crucially used in our paper. Moreover, we can trace all constants used in our proof explicitly.  The constants $m$ and $\epsilon$ for $f$ in \eqref{eq-assumption-lower-bound-f} determine the upper bound constant of a priori estimates, which in turn determines all the rest of constants in later sections.

In the following, we outline our proof of Theorem \ref{main thm}. Key intermediate results are given by  
Theorem \ref{theorem-Schauder-phi} for a priori $C^{2,1/3}$-estimates, 
Lemma \ref{lemma-linearized-op-invert} for the invertibility of the linearized operator, 
Theorem \ref{theorem-existence-solutions} for the existence of solutions, 
Theorem \ref{thrm-smoothness-tubes} for the smooth dependence of solutions on the parameter $\ub$, 
and Theorem \ref{thrm-spacelike} for the spacelikeness of the hypersurfaces.

\subsubsection{\color{black} A New Quasilinear Elliptic Equation}
To derive a priori estimates for the fully anisotropic case, we change the form of (\ref{1.12}) and employ a more robust and powerful approach. We will use the function $\p=\p(\o,\ub)$ defined by 
$$R(\o,\ub)=\ub a e^{-\p(\o,\ub)}.$$ A straightforward computation changes (\ref{1.12}) to 
\begin{equation*}
\begin{split}
&\D_g\p+\f{e^{2\p}}{\ub^2 a^2}-\f{e^{3\p}f(\o,\ub)}{2\ub^2 a^2}
-2\O\chibh(\nab_g \p, \nab_g \p) \,\ub ae^{-\p}+2g(\eta, \nab_g \p)\\
&\qquad+\Big[-\f12\O\tr_g\chib{\color{black}-}\f{1}{\ub a e^{-\p}}-4\O\omb\Big]\ub ae^{-\p}|\nab_g \p|^2\\
&\qquad+\f{e^{\p}}{\ub a}\Big[\f12\O^{-1}\tr_g\chi-\f{1}{\ub a e^{-\p}}+\f{\ub a f(\o,\ub)}{2(\ub a)^2 e^{-2\phi}}\Big]=0.
\end{split}
\end{equation*}
Set $g={\color{black}(\ub a e^{-\p})^2}\gamma$. We further change the above equation to   
\begin{equation}\label{phi eqn}
\begin{split}
\mathcal{S}(\p)&\equiv\D_\gamma\p+1-\f12f(\o,\ub)e^{\p}
-2\O\chibh(\nab_{\gamma} \p, \nab_{\gamma} \p) \,\ub ae^{-\p}+2\gamma(\eta, \nab_{\gamma} \p)\\
&\qquad+\Big[-\f12\O\tr_g\chib{\color{black}-}\f{1}{\ub a e^{-\p}}-4\O\omb\Big]\ub ae^{-\p}|\nab_{\gamma} \p|^2\\
&\qquad+\ub ae^{-\p}\Big[\f12\O^{-1}\tr_g\chi-\f{1}{\ub a e^{-\p}}+\f{\ub a f(\o,\ub)}{2(\ub a)^2 e^{-2\phi}}\Big]=0.
\end{split}
\end{equation}
{\color{black} Under the assumptions \eqref{eq-assumption-f-0-1} and \eqref{eq-assumption-lower-bound-f},} in Section \ref{sec-Barriers} we construct upper and lower barriers for solutions $\phi(\o,\ub)$ to \eqref{phi eqn}. We then seek the solution satisfying 
\begin{equation}\label{eq-basic-bound}0\leq \p\leq \kappa.\end{equation} Here, $\kappa$ is a fixed large but specific constant related to the barriers. {\color{black} We should point out that the constants $m$ and $\epsilon$ in \eqref{eq-assumption-lower-bound-f} determine the constant $\kappa$ in \eqref{eq-basic-bound} explicitly, which then gives all the rest constants.}

We would like to emphasize that the metric $g={\color{black} (\ub a e^{-\p})^2}\gamma$ is degenerate as $\ub\to 0$, but the conformal metric $\gamma$ is uniformly bounded. As a consequence, equations in  
terms of the metric $\gamma$ are uniformly elliptic. This fact plays
a fundamental role in deriving various estimates for solutions to 
the elliptic equations of $\phi$ and $\partial_{\ub}\phi$.
We also point out that we will study three elliptic equations in this paper. The main equation is \eqref{phi eqn}, which is the equation for $\phi$. To solve
it, we need to study its linearized equation. To prove the spacelikeness of the apparent horizon,
we will also need to study the elliptic equation for $\partial_{\ub}\phi$. 

\subsubsection{Moser's Iteration and  Schauder Estimates}

Note that \eqref{phi eqn} is a quasilinear elliptic equation and $\gamma=\gamma(\ub,\p, \o)$.   
Assume $\p$ is a solution to $\mathcal S(\p)=0$.
We start our study with the $L^{\infty}$-bound for $\p$  as in \eqref{eq-basic-bound} {\color{black} and derive $C^{2,1/3}$-estimates in three successive steps.} First in Theorem \ref{theorem-Holder-solution}, we prove a $C^\alpha$-estimate of the form
$$|\p|_{C^{0,\a}(\mathbb{S}^2)}\leq Ce^{\kappa},$$ 
where $\a$ is a small positive number. This step from $L^{\infty}$ to $C^{0,\a}$ is achieved through the weak Harnack inequality and Moser's iteration together with an important observation that the equation \eqref{phi eqn} satisfies the \textit{natural growth} condition.
{\color{black} Refer to  Section 4.5 of \cite{HL} for the definition of natural growth condition.} 
Next in Theorem \ref{theorem-Holder-gradient}, we prove a $C^{1,1/3}$-estimate
\begin{equation*}
|\p|_{C^{1,1/3}(\mathbb{S}^2)}\le (Ce^{\kappa})^{Ce^{4\Lambda\kappa}},
\end{equation*}
where $\Lambda$ is the ellipticity constant of the metric component $(\gamma_{ij})$ and $C$ is a positive uniform constant.
In this step, we use {\color{black} the characterization of H\"older continuous functions by Campanato}  and the technique of freezing coefficients. 
Finally, we use the Schauder estimate and {\color{black} obtain 
\begin{equation}\label{eq-Schauder-estimates}|\p|_{C^{2, 1/3}(\mathbb{S}^2)}\leq (Ce^{\kappa})^{Ce^{4\Lambda\kappa}}.\end{equation}
Refer to Theorem \ref{theorem-Schauder-phi} for details.

\subsubsection{\color{black} The Invertibility of the Linearized Operator via the Positivity of the First Eigenvalue and the Strong Maximum Principle} We first compute the linearized operator $\partial_{\p}\mathcal{S}(\p)$. 
Let $\p$ be a given function on $\mathbb{S}^2$.  In Proposition \ref{form of linearization}, we prove that, for any function $w$ on $\mathbb{S}^2$,
\begin{equation}\label{eq-linearization-S-intro}
\partial_{\p} \mathcal S(\p)[w]=\f{d\,}{d\epsilon}\Big|_{\epsilon=0}\mathcal S(\p+\epsilon w) 
=\Delta_{\gamma}w-\f12fe^{\p}w+c^{i}\nab_i w+cw,
\end{equation}
where the connection is with respect to the metric $\gamma$, and 
$c=c(\o,\p,\ub)$ and $c^i=c^i(\o,\p,\ub)$ are functions satisfying 
\begin{equation}\label{c c^1 eqn1}
|c|+|c^i|\leq Ca^{-\f13} e^{\p}{\color{black} (|\nabla^2_{\gamma}\p|+|\nabla_\gamma\p|^2+1)}. 
\end{equation}
{\color{black} The estimate \eqref{c c^1 eqn1} plays an important role in the study of the invertibility of the linearized operator as in \eqref{eq-linearization-S-intro}, and it is based on scale-invariant estimates of some geometric quantities associated with the Einstein vacuum equations established in \cite{AL}.  In view of \eqref{eq-linearization-S-intro},} we consider
\begin{equation*}
L_0w=\D_\gamma w-\f12fe^{\phi}w,
\end{equation*}
and view $f$ and $\p$ as two fixed nonnegative functions on  
$\mathbb S^2$, with $f$ satisfying \eqref{eq-assumption-lower-bound-f}.  

In Lemma \ref{lemma-H1-estimates-L0}, we prove that {the equation $L_0w=0$ admits only the trivial solution. Moreover, the first eigenvalue $\mu_1$ of $-L_0$ is \textit{positive} and simple, and 
has a \textit{positive} eigenfunction $\psi_1$.}
By the positivity of $\mu_1$ and the positivity of $\psi_1$, we have 
\begin{align*}
\partial_{\p} \mathcal S(\p)[\psi_1]
=-\mu_1\psi_1+c^{i}\nab_i \psi_1+c\psi_1<0,
\end{align*}
{\color{black} by \eqref{eq-Schauder-estimates} and \eqref{c c^1 eqn1}, and  choosing} $a$ appropriately large depending on {\color{black}$\kappa$}. By Theorem 2.11 in \cite{HL}, the strong maximum principle then holds for the operator $\partial_{\p}\mathcal S(\p)$. In particular, for any $w\in C^2(\mathbb S^2)$, we have
\begin{equation*}
\max_{\mathbb S^2}|w|\le C\max_{\mathbb S^2}|\partial_{\p}\mathcal S(\p)[w]|,\end{equation*}
where $C$ is a positive constant independent of $w$. 
{\color{black} By the Schauder theory,  we obtain the crucial invertibility of the linearized operator
$\partial_{\p}\mathcal S(\p): C^{2, \beta}(\mathbb{S}^2)\rightarrow C^{{0},\beta}(\mathbb{S}^2)$ as in Lemma \ref{lemma-linearized-op-invert}.}

\subsubsection{The Existence, Uniqueness and Regularity of Solutions}
Now, we consider the equation $\mathcal S(\p, \ub)=0$. 

First, by the strong maximum principle,  we establish the uniqueness of solutions in Lemma \ref{lemma-comparison}.  
For the existence of solutions, we employ the method of continuity. 
We obtain the closedness by the $C^{2,1/3}$-estimates in Theorem \ref{theorem-Schauder-phi}
and the openness by the invertibility of the linearized operator in Lemma \ref{lemma-linearized-op-invert}. 
As a consequence, for each fixed constant $\ub\in(0,\delta]$, we find a smooth solution ${\p}=\p(\o,\ub)$ to $\mathcal S(\p, \ub)=0$
in Theorem \ref{theorem-existence-solutions}.

Next, we view such a solution $\p$ as a function of $\o$ and $\ub$. Using again the invertibility of the linearized operator and the implicit function theorem,  we conclude that $\p$ is smooth in both $\o$ and $\ub$ in 
Theorem \ref{thrm-smoothness-tubes}.

\subsubsection{The Spacelikeness of Apparent Horizons}
We further aim to prove that $\{u=1-\ub a e^{-\p(\o, \ub)}\}$ is a spacelike hypersurface. {\color{black} A key step is to derive an upper bound for $\partial_{\ub}\p$.} In Proposition \ref{prop-equation-partial-ub-phi}, we deduce the equation satisfied by $\partial_{\ub}\p$. For each $\ub \in (0,\delta]$, let $\p=\p(\o, \ub)$ be the solution 
to $\mathcal S(\p, \ub)=0$ as in Theorem \ref{theorem-existence-solutions}. 
Then, $\partial_{\ub}\p$ satisfies 
\begin{equation}\label{eq-ub-differentiation}
\Delta_{\gamma}(\partial_{\ub}\p)-\f12f(\o, \ub)e^{\p}\partial_{\ub}\p+c^{i}\nab_i (\partial_{\ub}\p)+c\partial_{\ub}\p=\frac{1}{\ub}h
+\f12\partial_{\ub}fe^{\p},
\end{equation}
where the connection is with respect to the metric $\gamma$, and 
$c=c(\o,\p,\ub)$, $c^i=c^i(\o,\p,\ub)$ and $h=h(\o,\p,\ub)$
are functions satisfying 
\begin{equation}\label{eq-ub-coefficients-estimates}
|c|+|c^i|+|h|\leq Ca^{-\f13} e^{\p}{\color{black} (|\nabla^2_{\gamma}\p|+|\nabla_\gamma\p|^2+1)}.\end{equation}
{\color{black} We now compare \eqref{eq-ub-differentiation} and \eqref{eq-ub-coefficients-estimates} with 
\eqref{eq-linearization-S-intro} and \eqref{c c^1 eqn1}, respectively. 
We point out that the linear operator acting on $\partial_{\ub}\p$ on the left-hand side of \eqref{eq-ub-differentiation}
is similar to the operator in \eqref{eq-linearization-S-intro}, with similar estimates of coefficients. 
However, the computation leading to \eqref{eq-ub-differentiation} and the derivation of \eqref{eq-ub-coefficients-estimates} 
are much more complicated.

With the help of the positivity of the first eigenvalue $\mu_1$ and the first eigenfunction $\psi_1$ again}, 
in Corollary \ref{cor-equation-partial-ub-phi} we prove
\begin{equation}\label{eq-estimate-partial-ub-phi 0}
\max_{\mathbb S^2}|\partial_{\ub}\p|\le C\big[\ub^{-1}a^{-\f13} \max_{\mathbb S^2}e^{\p}
(|\p|^2_{C^2(\mathbb S^2)}+1)+\max_{\mathbb S^2}|\partial_{\ub}fe^{\p}|\big],
\end{equation}
where $C$ is a positive constant independent of $a$, $\ub$ and $\p$. 
{\color{black} By  taking $a$ sufficiently large, we have 
$$\max_{\mathbb S^2}|\ub\partial_{\ub}\p|\le 1/2.$$

Let $g'$ be the induced metric to the surface $u=1-R(\o, \ub)=1-\ub a e^{-\p(\o, \ub)}$. We denote $\o=(\theta_1, \theta_2)$. Let $X=\lambda_1\partial_{\theta_1}+\lambda_2\partial_{\theta_2}+\lambda_3\partial_{\ub}$ be an arbitrary 
nonzero tangent vector with
some constants $\lambda_1, \lambda_2, \lambda_3$. A straightforward computation yields
\begin{align*}
g'(X,X)
=\ub^2 a^2 e^{-2\p}\lambda_i\lambda_j \gamma_{ij}-4\ub ae^{-\p}\lambda_i \lambda_3 \partial_{i}\p+4\lambda_3^2 a e^{-\p}\big(1-\ub \partial_{\ub}\p\big),\end{align*}
where the summation is over $i,j=1,2$. We then obtain $g'(X,X)>0$ if $a$ is sufficiently large.
In other words, the metric $g$ induced to $\{(\o, u, \ub);\,u=1-\ub a e^{-\p(\o, \ub)}\}$ is \textit{Riemannian} and the corresponding apparent horizon is spacelike, as asserted in Theorem \ref{thrm-spacelike}.

\subsubsection{A Final Remark} 
In analyzing both the quasilinear elliptic equation $\mathcal{S}(\p,\ub)=0$ and its linearization, 
the variation of Laplacian plays an important role. We compute variations of Laplacian in three different circumstances for different purposes. In analyzing the nonlinear equation $\mathcal S(\p,\ub)=0$, with the variation of Laplacian we calculate the difference between the Laplacian with respect to the general metric $\gamma$ and the Laplacian with respect to the standard spherical metric. The purpose of the comparison is to construct an explicit supersolution. Second, we need to derive a variation of Laplacian in proving the invertibility of linearized equations. Third, we need to derive the variation of Laplacian with respect to $\ub$ in studying the equation satisfied by $\partial_{\ub}\phi$. These variations are explicitly computed in terms of geometric quantities and solutions. With estimates of geometric quantities established in \cite{AL} and estimates of solutions established in Section \ref{sec-Schauder-estimates}, we are able to present precise estimates of these variations of the Laplacian operators as well as variations of other geometric quantities appearing in the equations. As a consequence, we can choose the parameter $a$ explicitly in terms of the initial data, notably $f$, through $m$ and $\epsilon$ defined in \eqref{eq-assumption-lower-bound-f}.

\smallskip 

The paper is organized as follows. In Section \ref{sec-Settings}, we demonstrate the settings. In Section \ref{sec-Barriers}, we construct subsolutions and supersolutions. In Section \ref{sec-Schauder-estimates}, we establish Schauder estimates. In Section \ref{sec-linearized-equations}, we analyze the linearized operator and prove its invertibility.  In Section \ref{sec-existence-solutions}, we employ the method of continuity to prove the existence and properties of solutions. In Section \ref{anisotropic examples}, we construct explicit anisotropic apparent horizons with valleys. In Appendix \ref{appendix-A}, we collect useful estimates for some geometric quantities  from \cite{AL}.

\subsection{Acknowledgements} We thank Demetrios Christodoulou and Sergiu Klainerman for valuable conversations.

\section{Settings}\label{sec-Settings} 

In this section, we introduce the main equation to be studied and rearrange it into a preferred form.

\subsection{The Double Null Foliation and the Coordinate System} 
All computations in this paper are based on a double null foliation, with which we decompose the spacetime with outgoing and incoming null cones, i.e., $u=\mbox{constant}$ and $\ub=\mbox{constant}$, respectively. Here $u$ and $\ub$ satisfy the eikonal equations 
$$g^{\mu\nu}\partial_{\mu}u\partial_{\nu}u=0, \quad g^{\mu\nu}\partial_{\mu}\ub\partial_{\nu}\ub=0.$$
Given $u'$ and $\ub'$ two fixed constants, the intersection of $u=u'$ and $\ub=\ub'$ is a $2$-sphere, denoted by $S_{u',\ub'}$.\footnote{For more detailed description of the double null foliation, interested readers are referred to Chapter 1 and Chapter 2 of \cite{Chr:book} and  Section 2.1 of \cite{An19}. } 

On $S_{u', \ub'}$, we will use coordinates ($\theta^1,\theta^2$). To decide $(\theta^1, \theta^2)$ for each point on $S_{u',\ub'}$, we follow Chapter 1 of \cite{Chr:book}. First,  $S_{0,0}$ is the standard $2$-sphere in the Minkowskian spacetime with radius $1$, and we define ($\theta^1, \theta^2$) to be coordinates on it via the stereographic projection. Then, we extend this $(\theta^1, \theta^2)$ coordinate system to $\Hb_0$ ($\ub=0$) by requiring
$$\f{\partial}{\partial u} \theta^A=0\mbox{ on $\Hb_0$}.\footnote{Here we assume the geometry along $\Hb_0$ to be Minkowskian.}$$
In another word,  fix a point $p$ on $S_{0, 0}$ and assume $l_p$ is the incoming null geodesics on $\Hb_0$ emitting from $p$, and we require all the points along $l_p$ have the same angular coordinate $(\theta^1, \theta^2)$.
We further extend coordinate system $(\theta^1, \theta^2)$ from $\Hb_0$ to the whole spacetime region by requiring
$$\Ls_{\f{\partial}{\partial \ub}} \th^A=0;\footnote{Here $\Ls_L$ is the restriction of the Lie derivative to $TS_{u',\ub'}$.}$$ 
namely, every point along the same outgoing null geodesics on $H_u$ shares the same angular coordinate. We hence obtain a coordinate system in a neighborhood of $S_{0, 0}$. 
With it, we take the ansatz of the Lorentzian metric $g$ to be of the form
\begin{equation}\label{equation g}
g=-2\O^2(du\otimes d\ub+d\ub\otimes du)+\gamma_{AB}(d\theta^A-d^A du)\otimes(d\theta^B-d^B du).
\end{equation}
We define $e_3, e_4$ by 
$${\color{black}e_3=\Omega^{-1}\frac{\partial}{\partial u}, \quad e_4=\Omega^{-1}\left(\frac{\partial}{\partial \ub}+d^A\frac{\partial}{\partial \th^A}\right).}$$ 
It is straightforward to check that $\{e_3, e_4\}$ forms a null pair. Here $e_3, e_4$ are along incoming and outgoing null directions and 
$$g(e_3, e_3)=0, \quad g(e_4, e_4)=0, \quad g(e_3, e_4)=-2.$$
Denote by $\{e_1, e_2\}$ a frame tangent to the $2$-sphere $S_{u,\ub}$. We then decompose curvature components and Ricci coefficients using frames $\{e_1, e_2, e_3, e_4\}$.

\subsection{The Geometric Quantities}
Set the indices $A,B$ to be $1,2$ and let $D_A=D_{e_{A}}$. We define Ricci coefficients: 
\begin{equation}\label{def Ricci coefficients}
\begin{split}
&\chi_{AB}=g(D_A e_4,e_B),\, \,\, \quad \chib_{AB}=g(D_A e_3,e_B),\\
&\eta_A=-\frac 12 g(D_3 e_A,e_4),\quad \etab_A=-\frac 12 g(D_4 e_A,e_3),\\
&\tilde\omega=-\frac 14 g(D_4 e_3,e_4),\quad\,\,\, \omegab=-\frac 14 g(D_3 e_4,e_3),\\
&\zeta_A=\frac 1 2 g(D_A e_4,e_3).
\end{split}
\end{equation}
Denote by $\slashed{g}_{AB}$ the induced metric on $\S$. We further decompose $\chi, \chib$ as
$$\chi_{AB}=\f12\tr\chi\cdot \slashed{g}_{AB}+\chih_{AB}, \quad \chib_{AB}=\f12\tr\chib\cdot \slashed{g}_{AB}+\chibh_{AB}.$$
Note that $\chih_{AB}$ and $\chibh_{AB}$ are the corresponding traceless parts of $\chi_{AB}$ and $\chib_{AB}$. In Appendix \ref{appendix-A}, we will
collect useful estimates for these quantities  from \cite{AL}. Such estimates play important roles in this paper.

\subsection{The Main Elliptic Equation} 
We now introduce the main equation to be solved and rearrange it in a preferred form. 
In Sections \ref{sec-Barriers}-\ref{sec-linearized-equations}, we will view $g$,  $\chib$, $\chi$, $\eta$, $\O$, $\omb$ defined above as functions of $\o\in\mathbb S^2$ and $u$, parametrized by $\ub$; 
while in Section \ref{sec-existence-solutions}, we will view them as functions of $\o$, $u$, and $\ub$.

We write $u=1-R$. For each fixed $\ub$ and any function $R=R(\o,\ub)$ on $\mathbb S^2$, we set 
\begin{equation}\label{eq-expression-cal-G-ub}
\begin{split}
\mathcal G(R,\ub)&=\D_{g} R+\f12 \O \tr_g\chib |\nab_g R|^2-\f12\O^{-1} \tr_g\chi\\
&\qquad+4\O \omb |\nab_g R|^2+2\Omega\chibh(\nab_g R, \nab_g R)+2g(\eta, \nab_g R).
\end{split}
\end{equation}
Here $g$ is viewed as an induced metric on the 2-sphere $M=\{u=1-R\}$ along the incoming null cone $\Hb_{\ub}$. 
Using the expression of $\tr\chi'$ in \eqref{deformation equation}, we have
$$u=1-R(\o) \text{ is a} 
\begin{cases}
\text{MOTS} & \text{if $\mathcal G(R,\ub)=0$ on }\mathbb{S}^2 ,\\
\text{trapped surface} & \text{if $\mathcal G(R,\ub)>0$ on }\mathbb{S}^2,\\
\text{untrapped surface} & \text{if $\mathcal G(R,\ub)<0$ on }\mathbb{S}^2.
\end{cases}$$
Our primary goal is to find a solution $R=R(\o, \ub)$ of $\mathcal G(R,\ub)=0$ for each $\ub$
and then study the regularity of $R=R(\o, \ub)$ in terms of $\ub$. 
When there is no danger of confusion, we suppress $\ub$ and write $R=R(\o)$ and $\mathcal G(R)$ for simplicity. 

We proceed to rearrange $\mathcal G(R,\ub)$ in a preferred form. The rearrangement is based on estimates established in  \cite{AL}. Set 
\begin{equation}\label{eq-definitions-alphas}\begin{split}
\alpha_1&=-\Big(\f12\O\tr_g\chib+\f{1}{R}+4\O\omb\Big),\\
\alpha_2&=\f12\O^{-1}\tr_g\chi-\f{1}{R}+\f{\ub a}{2R^2} f(\o,\ub),
\end{split}\end{equation}
where $f=f(\o,\ub)$ is the function introduced in \eqref{eq-definition-f}. We now rewrite $\mathcal G(R,\ub)$ as 
\begin{equation}\label{eq-expression-cal-G-ub-v1}
\begin{split}
\mathcal G(R,\ub)&=\D_{g} R-\f1R|\nab_g R|^2-\f1R+\f{\ub a}{2R^2}f\\
&\qquad+2\Omega g^{ik}g^{jl}\chibh_{kl}\nab_i R\nab_j R+2g^{ij}\eta_i \nab_j R
-\alpha_1 |\nab_g R|^2-\alpha_2.
\end{split}
\end{equation}
The metric $g$ depends on $\o$, $u$, and $\ub$. Note that  
\begin{equation*}
\partial_ug_{kl}=2\O\chib_{kl}, \quad \partial_{\ub}g_{kl}=2\O\chi_{kl}{\color{black}-\Omega d^A \f{\partial \chi_{kl}}{\partial \theta^A}}.\end{equation*}
By writing $\chib$ and $\chi$ as combinations of trace parts and traceless parts, we have 
\begin{equation}\label{eq-variation-g-u} 
\partial_ug_{kl}=\O\tr_g\chib\, g_{kl}+2\O \chibh_{kl},\end{equation}
and 
\begin{equation}\label{eq-variation-g-ub} \partial_{\ub}g_{kl}=\O\tr_g\chi\, g_{kl}+2\O \chih_{kl}{\color{black}-\Omega d^A \f{\partial \chi_{kl}}{\partial \theta^A}}.
\end{equation}
Along fixed $\Hb_{\ub}$, when no confusing would be caused, we suppress $\ub$ dependence and rewrite $g(\o,u,\ub)$ as $g(\o, u)$ with $u=1-R$. In local coordinates, the estimates of $g$ in \cite{AL} and \cite{An19} imply 
\begin{equation}\label{Lambda-g-1}
\f12R^2\leq \sqrt{|g|(\o,u)}\leq \f32R^2, \quad \Lambda^{-1}R^2I\leq g_{ij}(\o, u)\leq \Lambda R^2I,
\end{equation}
and
{\color{black}\begin{align}\label{Lambda-g-2}|D_{\o}g_{ij}(\o, u)|+R|\partial_{u}g_{ij}(\o,u)|&\leq \Lambda_1 R^2,\\
\label{Lambda-g-3} |D^2_{\o}g_{ij}(\o, u)|+R|D_{\o}\partial_{u}g_{ij}(\o,u)|+R^2|\partial^2_{u}g_{ij}(\o,u)|&\leq \Lambda_2 R^2,\end{align}
for some positive constants $\Lambda$, $\Lambda_1$, and $\Lambda_2$.} 
Estimates on $\Omega$, $\chi$, $\chibh$, $\eta$, $\alpha_1$, and $\alpha_2$ were established in \cite{AL} 
and will be recalled in Appendix \ref{appendix-A}.

Next, for each fixed $\ub$ we introduce a function $\phi=\phi(\o, \ub)$ by
$$R(\o,\ub)=\ub a e^{-\p(\o,\ub)}.$$
If there is no danger of confusion, we also write $\phi(\o)$, $f(\o)$ instead of $\phi(\o,\ub)$, $f(\o,\ub)$. We then define the operator $S(\p, \ub)$ by 
\begin{equation}\label{eq-relation-R-S}\mathcal S(\phi,\ub)=-R \mathcal G(\p, \ub).\end{equation}
With a straightforward computation, we obtain  
\begin{equation}\label{eq-expression-cal-S-ub-v1}
\mathcal S(\p,\ub)=R^2\D_{g}\phi+1-\f12fe^{\p}+F(\p, \ub),\end{equation}  
where 
\begin{equation}\label{eq-expression-cal-S-ub-nonlinear-v1}
F(\p, \ub)=-2R^3\Omega g^{ik}g^{jl}\chibh_{kl}\nab_i \p\nab_j \p+2R^2g^{ij}\eta_i \nab_j \p
+R^3\alpha_1 |\nab_g \p|^2+R\alpha_2.
\end{equation}

Last, we introduce a metric $\gamma=\gamma(\o, u)$ on $\mathbb S^2$ such that 
\begin{equation}\label{eq-relation-g-gamma} g=R^{2}\gamma.\end{equation} 
By \eqref{Lambda-g-1}, \eqref{Lambda-g-2} and \eqref{Lambda-g-3}, in local coordinates $\gamma$ satisfies 
\begin{equation}\label{Lambda-gamma-1}
\f12\leq \sqrt{|\gamma|(\o,u)}\leq \f32, \quad \Lambda^{-1}I\leq \gamma_{ij}(\o, u)\leq \Lambda I,
\end{equation}
and
{\color{black}\begin{align}\label{Lambda-gamma-2}|D_{\o}\gamma_{ij}(\o, u)|+R|\partial_{u}\gamma_{ij}(\o,u)|&\leq \Lambda_1,\\
\label{Lambda-gamma-3} |D^2_{\o}\gamma_{ij}(\o, u)|+R|D_{\o}\partial_{u}\gamma_{ij}(\o,u)|+R^2|\partial^2_{u}\gamma_{ij}(\o,u)|&\leq \Lambda_2,\end{align}
via renaming the positive constants $\Lambda_1$ and $\Lambda_2$ if necessary.} Using $\gamma$, it is straightforward to check that
\begin{equation}\label{eq-expression-cal-S-ub-v2}
\mathcal S(\p,\ub)=\D_{\gamma}\phi+1-\f12fe^{\p}+F(\p, \ub),\end{equation}
where
\begin{equation}\label{eq-expression-cal-S-ub-nonlinear-v2}
F(\p, \ub)=2R^{-1}\Omega\chibh_{kl} \gamma^{ik}\gamma^{jl}\nab_i \p\nab_j \p+2\gamma^{ij}\eta_i \nab_j \p
+R\alpha_1 |\nab_{\gamma} \p|^2
+R\alpha_2,
\end{equation}
and $\a_1, \a_2$ are defined in \eqref{eq-definitions-alphas}.

Note that $\mathcal{S}(\phi,\ub)=0$ is the main elliptic equation in the preferred form and can be expressed as 
\begin{align*}
\begin{split}
&\D_\gamma\p+1-\f12f(\o,\ub)e^{\p}
{\color{black}+}2\O\chibh(\nab_{\gamma} \p, \nab_{\gamma} \p) \,\ub ae^{-\p}+2\gamma(\eta, \nab_{\gamma} \p)\\
&\qquad+\Big[-\f12\O\tr_g\chib{\color{black}-}\f{1}{\ub a e^{-\p}}-4\O\omb\Big]\ub ae^{-\p}|\nab_{\gamma} \p|^2\\
&\qquad+\ub ae^{-\p}\Big[\f12\O^{-1}\tr_g\chi-\f{1}{\ub a e^{-\p}}+\f{\ub a f(\o,\ub)}{2(\ub a)^2 e^{-2\phi}}\Big]=0,
\end{split}
\end{align*}
which is (\ref{phi eqn}) in the introduction.

By the relation \eqref{eq-relation-R-S}, we further conclude that 
$$u=1-\ub a e^{-\p(\o)} \text{ is a} 
\begin{cases}
\text{MOTS} & \text{if $\mathcal S(\p,\ub)=0$ on }\mathbb{S}^2 ,\\
\text{trapped surface} & \text{if $\mathcal S(\p,\ub)<0$ on }\mathbb{S}^2,\\
\text{untrapped surface} & \text{if $\mathcal S(\p,\ub)>0$ on }\mathbb{S}^2.
\end{cases}$$
Note that we have
\begin{equation}\label{eq-relation-u-R-phi}u=1-R(\o,\ub), \quad R(\o,\ub)=\ub a e^{-\p(\o,\ub)}.\end{equation} 
In the rest of the paper, we will study exclusively the equation $\mathcal S(\p,\ub)=0$ and aim to find its solution $\p(\o, \ub)$.

\section{Subsolutions and Supersolutions}\label{sec-Barriers}  

In this section, we construct an untrapped surface and a trapped surface  
on $\Hb_{\ub}$ for each small $\ub>0$. In later sections, we will prove that 
there exists a unique MOTS between these two surfaces.

Throughout this section, we fix a $\ub\in (0,\delta]$. 
For any function $\phi=\phi(\o)$ on $\mathbb S^2$, let $\mathcal S(\phi)$ be the operator defined by 
\eqref{eq-expression-cal-S-ub-v2}, i.e., 
\begin{equation}\label{3.0-v1}
\mathcal S(\p)=\D_{\gamma}\phi+1-\f12fe^{\p}+F(\phi),
\end{equation}
where $F$ is given in \eqref{eq-expression-cal-S-ub-nonlinear-v2}. 
The aim of this section is to construct subsolutions and supersolutions for the operator $\mathcal S$. 
We always assume that $f$ is a nonnegative smooth function on $\mathbb S^2$, not identically zero. 
We point out that $f$ is allowed to be zero in a large portion of $\mathbb S^2$. 

\begin{remark}
Throughout this paper, we will use parameters $a$ and $\kappa$, which satisfy
\begin{equation}\label{eq-choice-a-b}
(Ce^{\kappa})^{Ce^{4\Lambda\kappa}}\leq a^{\f14}\ll a^{\f13}.
\end{equation}
Here, $C$ is a uniform positive constant defined in Theorem \ref{theorem-Schauder-phi}, and $\Lambda$ is introduced in \eqref{Lambda-gamma-1}.\end{remark}

\begin{remark}\label{b anisotropicity} 
For any $m\in (0,1)$, $\e\in(0, {\pi}/{2})$, $\tau>2.2$, and $C(\tau)$ to be specified in Lemma \ref{lemma-difference-Laplacians} , as to be discussed in Theorem \ref{theorem-existence-solutions} and Remark \ref{negative phi}, we solve for MOTS lying in $0\leq \ub \leq \delta$ and $0\leq u \leq 1-m\epsilon^{\tau+2}C(\tau)\ub a$. To guarantee that the MOTS is within the existence region of Theorem \ref{main thm}, i.e., $0\leq \ub \leq \d$ and $0\leq u\leq 1-b\ub \at$, we require $m\epsilon^{\tau+2}C(\tau)\geq ba^{-\f12}$. This requirement can easily be achieved by setting $1\ll b\ll \at$ and $a$ being sufficiently large. Here allowing $b\neq \at$ as in \cite{AL} is important.
\end{remark}

We start from constructing subsolutions. 

\begin{lemma}\label{lemma-sub-solution-phi}
Let $f$ be a nonnegative continuous function on $\mathbb S^2$, with $f\le 1$.  
Then, $\underline{\p}(\o)=\ln(3/2)$ 
satisfies $\mathcal S(\underline{\p})>0$, for $a$ sufficiently large. 
\end{lemma}  

\begin{proof} By \eqref{eq-expression-cal-S-ub-nonlinear-v2}, \eqref{eq-definitions-alphas} and \eqref{A1}, we have 
\begin{equation}\label{eq-estimate-F-v1}
|F(\p)|\le C e^{\p}a^{-\f13}(|\nab_{\gamma}\p|^2+1).
\end{equation} 
Hence, 
$$\mathcal S(\ln(3/2))=1-\f34f+F(\ln(3/2))\ge \f14+F(\ln(3/2))>0,$$
if $a$ is sufficiently large. 
\end{proof} 

\begin{remark}\label{negative phi}
There are other subsolutions. For example, the function $\underline{\p}=\ln(\ub a)<0$ satisfies $\mathcal S(\underline{\phi})>0$. With  $\underline{\p}=\ln(\ub a)$, we can prove a similar result to Theorem \ref{theorem-existence-solutions}. For simplicity of presentations,  we choose $\underline{\phi}(\o)=\ln({3}/{2})>0$ in the following sections.
\end{remark}

Before we construct supersolutions, we estimate $\D_{\gamma}\phi-\D_{\gamma_0}\phi$. 
Here, $\gamma_0$ is the standard spherical metric on $\mathbb S^2$. 
We first derive a well-known formula for the variation of the Laplacian operator.

\begin{lemma}\label{lemma-variation-Laplacians} 
Let $\{\gamma(s)\}$ be a family of metrics and $\{\p(s)\}$ be a family of functions, both parametrized by $s$. 
Then, 
\begin{equation}\label{eq-variation-Laplace-v1}
(\Delta_{\gamma}\p)^{\boldsymbol\cdot}=\Delta_\gamma\dot{\p}-\gamma^{ij}\gamma^{kl}\dot{\gamma}_{jl}\nabla_{ik}\p
-\gamma^{ij}\gamma^{kl}\nabla_i\dot{\gamma}_{jl}\nabla_k\p+\frac12\gamma^{ij}\gamma^{kl}\nabla_l\dot{\gamma}_{ij}\nabla_k\p,
\end{equation}
where we denote $\dot{\,}=\f{d\,}{ds}$. 
\end{lemma}

\begin{proof} First, we have
$$\dot{\gamma}^{ij}=-\gamma^{ik}\gamma^{jl}\dot{\gamma}_{kl}.$$ 
Next, with 
$$\Gamma_{ij}^k=\frac12\gamma^{kl}\big[\partial_i\gamma_{jl}+\partial_j\gamma_{il}-\partial_l\gamma_{ij}\big],$$ 
we get
\begin{align*}
\dot{\Gamma}_{ij}^{k}=\frac12\gamma^{kl}\big[\partial_i\dot{\gamma}_{jl}+\partial_j\dot{\gamma}_{il}-\partial_l\dot{\gamma}_{ij}\big]
+\gamma_{ml}\Gamma_{ij}^m\dot{\gamma}^{kl}.
\end{align*} 
With 
$$\Delta_{\gamma} \p=\gamma^{ij}\partial_{ij}\p-\gamma^{ij}\Gamma_{ij}^k\partial_k\p,$$
a straightforward calculation yields the desired result. 
\end{proof}

\begin{lemma}\label{lemma-difference-Laplacians}  
Let $\phi\in C^2(\mathbb S^2)$ be an arbitrary  positive function. Then, 
\begin{equation}\label{eq-estimate-difference-Laplacians}
|\D_{\gamma}\phi-\D_{\gamma_0}\phi|\le C(1+\p)(|\nab_{\gamma}^2\p|+|\nab_{\gamma} \p|+|\nab_{\gamma}\p|^2)e^{\p}a^{-\f13},\end{equation}
where $C$ is a positive constant, independently of $a$ and $\p$. 
\end{lemma}

\begin{proof} For any constants $\ub${\color{black},} $\ub'$ and a function $\phi\in C^2(\mathbb S^2)$, set 
$$S_{1-\ub a e^{-\p}, \ub'}=\{(\om, u, \ub'); \o\in \mathbb S^2, u=1-\ub ae^{-\p}\}.$$
In particular, when $\phi=0$ we have
$$S_{1-\ub a, \ub'}=\{(\om, u, \ub'); \o\in \mathbb S^2, u=1-\ub a\}.$$
We now fix a positive function $\phi\in C^2(\mathbb S^2)$. Note that the Minkowski data are prescribed along $\ub=0$. 
Hence, $S_{1-\ub a,0}$ is the standard $2$-sphere with radius $\ub a$. 
Note that $S_{0,0}$ on $\Hb_0$ is the standard $2$-sphere with radius $1$ and 
that $\gamma_0$ is the standard spherical metric on it. Let $g|_{S_{1-\ub a, 0}}$ be the metric $g$ restricted to $S_{1-\ub a, 0}$. 
Then, $g|_{S_{1-\ub a, 0}}=\ub^2a^2\gamma|_{S_{1-\ub a, 0}}$, with $\gamma|_{S_{1-\ub a, 0}}=\gamma_0$. 
Hence, 
$$\D_{\gamma}\phi-\D_{\gamma_0}\phi=\D_{\gamma}|_{S_{1-\ub a e^{-\p}, \ub}}\p-\D_{\gamma}|_{S_{1-\ub a, 0}}\p.$$
For $s\in [0,1]$, we set
\begin{equation}\label{R and s}
R(s)=\ub ae^{-s\p}, \quad u(s)=1-R(s), \quad \ub (s)=\ub s. 
\end{equation}
In the following, we denote by dot the derivative with respect to $s$, and 
we also suppress the dependence on $s$. Hence, by \eqref{R and s} it follows
$$\dot{R}=-R\p, \quad \dot{u}=-\dot{R}, \quad \dot{\ub}=\ub.$$ 
For simplicity, we also denote by $\gamma(s)$ the metric $\gamma$ on $S_{1-\ub a e^{-s\p}, \ub s}$. Then, 
\begin{equation}\label{eq-difference-Laplace-integral}\D_{\gamma}\phi-\D_{\gamma_0}\phi
=\D_{\gamma(1)}\p- \D_{\gamma(0)}\p=\int_0^1(\D_{\gamma}\p)^{\boldsymbol\cdot}ds.\end{equation}

By $g=R^2\gamma$ and $\dot{u}=-\dot{R}$, via calculating $\big(g|_{S_{1-R(s),\ub(s)}} \big)^{\boldsymbol\cdot}$, we get 
\begin{equation}\label{gamma dot}
R^2\dot{\gamma}_{kl}=-2R\dot{R}\gamma_{kl}-\dot{R}\partial_ug_{kl}+\ub\partial_{\ub}g_{kl}.  
\end{equation}
{\color{black} Recall \eqref{eq-variation-g-u} and \eqref{eq-variation-g-ub}, i.e.,} 
\begin{equation*}
\partial_ug_{kl}=\O\tr_g\chib g_{kl}+2\O \chibh_{kl}, \quad \partial_{\ub}g_{kl}=\O\tr_g\chi g_{kl}+2\O \chih_{kl}{\color{black}-\Omega d^A \f{\partial \chi_{kl}}{\partial \theta^A}}.
\end{equation*}
Hence,  using \eqref{gamma dot} and $\dot R=-R\phi$, we get
\begin{align*}
\dot{\gamma}_{kl}=\big[\p(R\O\tr_g\chib+2)+\ub\O\tr_g\chi \big]\gamma_{kl}
+\big[2{\p}R^{-1}\O \chibh_{kl}+\ub R^{-2}(2\O \chih_{kl}{\color{black}-\Omega d^A \f{\partial \chi_{kl}}{\partial \theta^A}})\big].
\end{align*}
A substitution in \eqref{eq-variation-Laplace-v1} yields 
\begin{align*}
(\D_{\gamma}\p)^{\boldsymbol\cdot}&=-(R\O\tr_g\chib+2)\phi\Delta_{\gamma}\p-\ub \O\tr_g\chi\Delta_{\gamma}\p\\
&\qquad-\gamma^{ij}\gamma^{kl}[2\phi R^{-1}\O\chibh_{jl}+\ub R^{-2}(2\O\chih_{jl}{\color{black}-\Omega d^A \f{\partial \chi_{jl}}{\partial \theta^A}})]\nab_{ik}\phi\\
&\qquad-\gamma^{ik}\big[\nab_i\p(R\O\tr_g\chib+2)+\ub\nab_i(\O\tr_g\chi)\big]\nab_k\phi-\phi\gamma^{ik}\nab_i(R\O\tr_g\chib)\nab_k\p\\
&\qquad-\gamma^{ij}\gamma^{kl}\nab_i\big[2\p R^{-1}\O\chibh_{jl}+\ub R^{-2}(2\O\chih_{jl}{\color{black}-\Omega d^A \f{\partial \chi_{jl}}{\partial \theta^A}})\big]\nab_k\p\\
&\qquad+\gamma^{kl}\nab_l\big[\phi(R\O\tr_g\chib+2)+\ub\O\tr_g\chi\big]\nab_k\p\\
&\qquad+\f12\gamma^{ij}\gamma^{kl}\nab_l\big[2\phi R^{-1}\O\chibh_{ij}+\ub R^{-2}(2\O\chih_{ij}{\color{black}-\Omega d^A \f{\partial \chi_{ij}}{\partial \theta^A}})\big]\nab_k\p.
\end{align*}
We point out that the above identity is evaluated at $(\o, 1-\ub ae^{-s\p(\o)}, \ub s)$
{\color{black} and can be viewed as a linear combination of $\p\nab_{ij}\p, \nab_i\p\nab_j\p, \p\nab_i\p, \nab_{ij}\p, \nab_i\p$ with coefficients given by $\gamma^{ij}$ and  
\begin{align*}R\O\tr_g\chib+2,\, R^{-1}\O\chibh_{kl},\, \nab_i(R\O\tr_g\chib),\, \nab_i(R^{-1}\O\chibh_{kl}),\\
\ub\O\tr_g\chi,\, \ub R^{-2}\O\chih_{kl},\, \ub\nab_i(\O\tr_g\chi),\, \ub \nab_{i}(R^{-2}\O\chih_{kl}),\\
{\color{black}\ub R^{-2}\Omega d^A \f{\partial \chi_{kl}}{\partial \theta^A}, \, \ub\nab_i(R^{-2}\Omega d^A \f{\partial \chi_{kl}}{\partial \theta^A})}. \quad \quad \, \qquad \end{align*}
By \eqref{A1}, these quantities are bounded by $Ce^{s\p}a^{-\f13}$.} 
Therefore, 
$$|(\D_{\gamma}\p)^{\boldsymbol\cdot}|\le 
C(1+\p)(|\nab^2\p|+|\nab \p|+|\nab\p|^2)e^{\p}a^{-\f13}.$$
With \eqref{eq-difference-Laplace-integral}, this implies the desired result.  
\end{proof}

Next, we assume that $f$ satisfies \eqref{eq-assumption-lower-bound-f}, i.e.,
$$f\ge m\quad\text{on }B_{p}(\e),$$
for some point $p\in \mathbb S^2$ and some constants $m\in (0,1)$ and $\e\in (0,\pi/2)$.

\begin{lemma}\label{lemma-super-solutions-phi}
Let $f$ be a nonnegative continuous function on $\mathbb S^2$, satisfying \eqref{eq-assumption-lower-bound-f}. 
Then, for any $\tau>2.2$, 
there exists  $\overline{\phi}\in C^\infty(\mathbb S^2)$ such that, for any sufficiently large $a$, 
$\mathcal S(\overline{\p})<0$ and 
\begin{align}\label{eq-estimate-super-phi}\begin{split}
\overline{\phi}_{\min}&=\overline{\phi}(p)=-\log(m\e^2)+C(\tau),\\
\overline{\phi}_{\max}&=\overline{\phi}(-p)=-\log(m\e^{\tau+2})+C(\tau),
\end{split}\end{align}
for some positive constant $C$ depending only on $\tau$. 
\end{lemma} 

\begin{proof} The proof of this lemma is essentially a similar one in \cite{KLR}.
We provide some details, which will be used in a later section.
We write $\mathcal S(\p)$ as 
\begin{equation}\label{3.0-v2}\mathcal S(\p)=
\D_{\gamma_0}\phi+1-\f12fe^{\phi}+\D_{\gamma}\phi-\D_{\gamma_0}\phi+F(\phi),\end{equation}
where $\D_{\gamma_0}$ is the Laplacian with respect to the standard metric on $\mathbb S^2$. 
We first construct a function $\phi_{\epsilon}=\phi_{\epsilon}(\o)$ such that 
\begin{equation}\label{eqn KLR}
\D_{\gamma_0}\phi_{\epsilon}+1.1<\f12 f(\o) e^{\phi_{\epsilon}}. 
\end{equation}
Here, we use $1.1$ instead of $1$ on the left-hand side 
to absorb the lower order terms in $\mathcal S$. 
Denote by $\lambda$ the distance function from $p$. 
Set 
$$w=\log\Big[\sin\Big(\f{\lambda}{2}\Big)\Big].$$ 
Then, 
$$\D_{\gamma_0}w+\f12=2\pi \d_{p},$$
and 
$$w=\log\lambda+O(1), \quad |\nab_{\gamma_0} w|=O(\lambda^{-1}), \quad |\nab_{\gamma_0}^2 w|=O(\lambda^{-2}).$$ 
Let $\chi_{\e}$  be a smooth cut-off function, satisfying 
$$\chi_{\e} = 
\begin{cases}
        0  & \text{on $B_p(\epsilon/2)$},\\
        1 & \text{on $\mathbb{S}^2\backslash B_p(\epsilon)$}. \end{cases}$$
Define $\tilde{w}_{\e}=\chi_{\e}w+(1-\chi_{\e})\log\e$. Then, 
$$ \tilde{w}_{\e}=
\begin{cases}
       \log \e  & \text{on $B_p(\epsilon/2)$},\\
       \log\e+O(1) & \text{on $B_p(\e)\backslash B_p(\e/2)$},\\
       \log\lambda+O(1)  & \text{on $\mathbb{S}^2\backslash B_p(\epsilon)$},\end{cases}$$
and
$$|\nab_{\gamma_0}\tilde{w}_{\epsilon}|=O(\e^{-1}), \quad 
|\nab_{\gamma_0}^2\tilde{w}_{\epsilon}|=O(\e^{-2}) \quad \text{on $\mathbb{S}^2\backslash B_p(\epsilon/2)$}.$$
Moreover, 
$$\D_{{\gamma_0}}\tilde{w}_{\e}+\f12=0 \quad \text{on $\mathbb{S}^2\backslash B_p(\epsilon)$}.$$
Consider $\tilde{\phi}_{\e}=\tau\tilde{w}_{\e}$. 
Then we have
$$\tilde{\phi}_{\e}=\begin{cases}\tau\log \e  & \text{on $B_p(\epsilon/2)$},\\
        \tau\log\e+O(1) & \mbox{on $B_p(\e)\backslash B_p(\e/2)$},\\
        \tau\log\lambda+O(1)  & \mbox{on $\mathbb{S}^2\backslash B_p(\epsilon)$},\end{cases}$$
and 
$$ |\nab_0\tilde{\phi}_{\epsilon}|=O(\e^{-1}), \quad
         |\nab_0^2\tilde{\phi}_{\epsilon}|=O(\e^{-2}) \quad\text{on $\mathbb{S}^2\backslash B_p(\epsilon/2)$}.$$
Moreover, it holds
$$\D_{\gamma_0}\tilde{\phi}_{\e}+1.1=-\f{\tau}{2}+1.1<0
\quad\text{on }\mathbb{S}^2\backslash B_p(\e).$$        
For a constant $s_{\e}>0$ to be determined, set 
$$\phi_{\e}=\tilde{\phi}_{\e}-\log s_{\e}.$$
It is obvious that
$\D_{\gamma_0}\phi_{\e}=\D_{\gamma_0}\tilde{\phi}_{\e}.$
Hence,  
$$\D_{\gamma_0}\phi_{\e}+1.1=\D_{\gamma_0}\tilde{\phi}_{\e}+1.1<0
\quad\text{on }\mathbb{S}^2\backslash B_p(\e).$$
In addition,  there exists a constant $C_0>0$ such that    
$$\D_{\gamma_0}\phi_{\e}=\D_{\gamma_0}\tilde{\phi}_{\e}<\f{C_0}{\e^{2}}\quad\text{on }B_p(\e).$$
By $f\geq m$ in $B_p(\e)$ in \eqref{eq-assumption-lower-bound-f}, to have  $C_0/\e^{2}\le\f12f(\o)e^{\phi_{\e}}$ in $B_p(\e)$, we need
$$\inf_{B_p(\e)}\phi_{\e}\geq \log \Big(\f{2C_0}{m \e^{2}}\Big).$$
By estimates on $\tilde{\phi}_{\e}$, we have
$$\inf_{B_p(\e)}\phi_{\e}\geq-\log s_{\e}+\tau\log\e-\log C_1,$$
for some constant $C_1>0$. Therefore, it suffices to choose 
$$s_{\e}=\f{m\e^{\tau+2}}{2C_0 C_1}.$$
Then, we obtain 
$$\D_{\gamma_0}\phi_{\epsilon}+1.1<\f12 f(\o)e^{\phi_{\epsilon}}\quad\text{on }\mathbb S^2.$$ 
This is the desired inequality \eqref{eqn KLR}. 
Recall that $\p_{\e}$ is given by 
$$\p_{\e}=\tau\big[\chi_\e w+(1-\chi_\e)\log\e\big]-\log\Big(\f{m\e^{\tau+2}}{2C_0 C_1}\Big).$$
We hence have \eqref{eq-estimate-super-phi}  and  
\begin{align*}|\nab_{\gamma_0} \phi_{\e}|=O(\e^{-1}),\quad
|\nab_{\gamma_0}^2 \phi_{\e}|=O(\e^{-2}).\end{align*}
Moreover, by \eqref{eq-estimate-F-v1} and \eqref{eq-estimate-difference-Laplacians}, 
$\overline{\p}=\p_{\e}$ satisfies $\mathcal S(\overline{\p})<0$ for any sufficiently large $a$.
\end{proof} 

The functions $\underline{\phi}$ and  $\overline{\p}$ in 
Lemma \ref{lemma-sub-solution-phi} and Lemma \ref{lemma-super-solutions-phi}
are a subsolution  and a supersolution of $\mathcal S(\p)=0$, respectively. In the following sections, we will construct a solution $\p$ of $\mathcal S(\p)=0$ such that 
\begin{equation}\label{MOTS-barriers-phi}
\underline{\p}\leq \p\leq \overline{\p} \quad\text{on }\mathbb S^2, 
\end{equation}
where the subsolution and supersolution serve as lower and upper barriers.  Furthermore, we will prove that such a $\phi$ is unique.

\section{Moser's Iteration and Schauder Estimates}\label{sec-Schauder-estimates} 

In this section, we will derive Schauder estimates for solutions.  We will present a detailed proof of $C^{1,1/3}(\mathbb{S}^2)$-estimate, {\color{black} emphasizing on the explicit dependence of the $L^\infty$-bound of solutions.} Again, we fix a $\ub$ and let $\mathcal S(\phi)$ be the operator defined by 
\eqref{3.0-v1}. We will prove the following result.

\begin{theorem}\label{theorem-Schauder-phi}
Let $\mathcal S(\phi)$ be defined by \eqref{3.0-v1} and $\kappa$ be a positive constant. 
Assume $\p$ is a solution of $\mathcal S(\p)=0$, with 
$0\le \p\le\kappa$ on $\mathbb S^2$. 
Then, 
$$|\p|_{C^{2, 1/3}(\mathbb{S}^2)}\leq (Ce^{\kappa})^{Ce^{4\Lambda\kappa}},$$
where $\Lambda$ is the ellipticity constant of the metric component $(\gamma_{ij})$ and $C$ is a 
positive constant depending only on $\gamma, \chibh, \eta, \a_1, \a_2$ and their angular derivatives.
\end{theorem}

We will prove Theorem \ref{theorem-Schauder-phi} in three steps. 
In the first step, we estimate the H\"older norms of the solutions. This step is based 
on the local boundedness and the weak Harnack inequality due to Moser. 
In the second step, we estimate the H\"older norms of the first derivatives. This step is based 
on an integral characterization of H\"older continuous functions. 
In the third step, we estimate the H\"older norms of the second derivatives. 
In the course of the proof, we keep track of the dependence on $\kappa$.

Let $F$ be given by \eqref{eq-expression-cal-S-ub-nonlinear-v2}. 
We aim to establish a priori estimates for solutions $\p=\p(\o)$ to the equation 
\begin{equation}\label{main eqn lambda}
-\D_{\gamma}\phi=1-\f12fe^{\p}+F.
\end{equation}
We point out that \eqref{main eqn lambda} is a quasilinear elliptic equation for $\p$, 
since $F$ involves the gradient of $\p$ and $\gamma$ also depends on $\p$. 
We always assume \eqref{eq-assumption-f-0-1}, i.e., 
$$0\leq f\leq 1 \,\text{ on }\mathbb S^2.$$
The equation \eqref{main eqn lambda} is defined globally on $\mathbb S^2$. 
In the following, we will prove a more general result, valid in arbitrary dimensions. 

{\color{black}In this section, to be aligned with the literature on elliptic equations, we let $B_r(p)$ be the standard ball  in $\mathbb R^n$ centered at $p$ with radius $r$, for any integer $n\ge 2$. And sometimes we write $B_r$ for short, when the center is at the origin. In other sections, to emphasize the radius $r$, we adopt the notation $B_p(r)$.} We consider the nonlinear equation of the form 
\begin{equation}\label{u equation lambda main}
-\partial_j\big(a_{ij}\partial_i\p\big)=\sqrt{|\gamma|}\Big[1-\f12fe^{\p}+F\Big].
\end{equation}
This is the equation \eqref{main eqn lambda} in local coordinates.  
First, we assume 
\begin{align*}
\Lambda^{-1} I \le \big(a_{ij}(\o, \p)\big)  \le \Lambda I,
\end{align*} 
for some positive constant $\Lambda$, and 
\begin{align*}
1/2\leq \sqrt{|\gamma|}(\o, \p)\leq 2.
\end{align*} 
By abusing notations slightly, we set, for any $(\o, \phi, p) \in B_1
\times \mathbb R \times \mathbb R^n$,  
$$F(\o, \p, p)=a^{-\f13}e^{\p}\big[c_{ij}p_ip_j+c_ip_i+c_0\big],$$ 
where $c_{ij}, c_i, c_0$ are functions of $(\o,\p)$ on $B_1\times \mathbb R_+$. 
By the Cauchy-Schwarz inequality, we have 
\begin{equation*}|F|\le C e^{\p}a^{-\f13}(|D \phi|^2+1).\end{equation*}
{\color{black} This is simply \eqref{eq-estimate-F-v1} in local coordinates.}  
Next, we take $a$ sufficiently large such that $a^{-\f13}e^{\p}\ll 1.$ 
Hence, 
\begin{equation}\label{eq-estimate-F}|F|\le c(|D \phi|^2+1),\end{equation}
for some small positive constant $c$. In fact, we can take $c=1$.

We now put \eqref{u equation lambda main} in divergence form as follows: 
\begin{equation}\label{eq-main-nonlinear}\int a_{ij} (\o,\p)D_i \p D_j \psi\, d\o
= \int \sqrt{|\gamma|}\Big[1-\f12fe^{\p} +F(\o,\p, D\p)\Big]\psi\, d\o, 
\end{equation} 
for any $\psi \in H^1_0(B_1) \cap L^\infty (B_1)$. We always assume $\p=\p(\o)\in H^1(B_1)$ and 
\begin{equation}\label{eq-main-bound-u}0\le \p\le \kappa\quad\text{on }B_1,\end{equation} 
for some positive constant $\kappa$.

In the following, we will derive estimates of the H\"older norms for the solution $\p$, its gradients, and its Hessians, successively, 
with a focus on explicit dependence in terms of $\kappa$.

\subsection{Interior  Estimates of H\"older Norms} 
Our first goal is to derive interior estimates of the H\"older norm of $\p$. 
In the proof, we will adapt some relevant results in  
\cite{GT} or \cite{HL}.

\begin{theorem}\label{theorem-Holder-solution} 
Suppose $\p\in H^1(B_1)$ is a positive solution to
\eqref{eq-main-nonlinear} with \eqref{eq-main-bound-u} satisfied.  
Then,  
\begin{equation}\label{eq-Holder-semi-norm}
[\p]_{C^\alpha(B_{1/2})}\le Ce^{\kappa},\end{equation}
where $C$ is a positive constant depending only on $n$ and $\Lambda$, 
and $\alpha$ is given by 
\begin{equation}\label{eq-definition-alpha}
\alpha=\epsilon_0e^{-4\Lambda\kappa} \in (0,1),\end{equation}
for some small constant $\epsilon_0>0$ depending on $n$ and $\Lambda$.
\end{theorem}

\begin{proof}  Set, for $s \in (0,1)$, 
$$M(s){\color{black}}=\max_{B_{s}}{\p},\quad m(s){\color{black}}=\min_{B_{s}}{\p}.$$
Fix an $s \in (0,1)$ and consider $\p$ defined in $B_{s}$. Write $m=m(s)$ for brevity and set  
$$v=\p-m.$$ 
For any $\psi \in
H_0^1(B_{s}) \cap L^\infty(B_{s})$ with $\psi \ge 0$, by \eqref{eq-main-nonlinear} and \eqref{eq-estimate-F}, we have  
\begin{align*} 
\int a_{ij} D_ivD_j \psi\, d\o\le \int \sqrt{|\gamma|}\Big[1-\f12fe^{\p}+|Dv|^2+1\Big] \psi\, d\o
\le  \int [4+2|Dv|^2] \psi\, d\o,
\end{align*}
and 
\begin{align*} 
\int a_{ij} D_ivD_j \psi\, d\o\ge \int \sqrt{|\gamma|}\Big[1-\f12fe^{\p}-|Dv|^2-1\Big] \psi\, d\o
\ge  \int [-e^{\p}-2|Dv|^2] \psi\, d\o.
\end{align*}
First, we set $\underline{v}=\underline{\xi}(v)$ for some function $\underline{\xi}$ to be determined,  
with properties $\underline{\xi}\ge 0$, $\underline{\xi}'\ge 0$, and $\underline{\xi}''\ge 0$.  
Then for $\psi \in
H_0^1(B_{s}) \cap L^\infty(B_{s})$ with $\psi \ge 0$, we have
\begin{align*} \int a_{ij} D_i\underline{v}D_j \psi\, d\o&= \int a_{ij} D_i v \underline{\xi}'(v)D_j
 \psi\, d\o\\
&= \int a_{ij} D_i v D_j (\underline{\xi}'(v) \psi)\, d\o -  \int
a_{ij}
D_i v D_j v \underline{\xi}''(v)\psi\, d\o\\
&\le \int [4+2|Dv|^2] \underline{\xi}'(v)\psi \, d\o-  \int\Lambda^{-1}
|Dv|^2\underline{\xi}''(v)\psi\, d\o\\
&= \int 4\underline{\xi}'(v)\psi d\o+\int |Dv|^2[2\underline{\xi}'(v)-\Lambda^{-1}\underline{\xi}''(v)]\psi d\o.
\end{align*}
We take $\underline{\xi}$ such that $\underline{\xi}''=2\Lambda\underline{\xi}'$ with 
$\underline{\xi}(0)=0$ and $\underline{\xi}'(0)=1$. Then, it holds
\begin{equation}\label{eq-expression-lower-eta}
\underline{\xi}({\color{black}\tau})=\f1{2\Lambda}(e^{2\Lambda {\color{black}\tau}}-1).\end{equation}
Hence, 
for $\psi \in
H_0^1(B_{s}) \cap L^\infty(B_{s})$ with $\psi \ge 0$, we have
\begin{align*} \int a_{ij} D_i\underline{v}D_j \psi\, d\o\le \int 4{\color{black}\underline{\xi}}'(v)\psi \, d\o
\le \int 4e^{2\Lambda \kappa}\psi \, d\o,
\end{align*}
where we used $\underline{\xi}'(v)=e^{2\Lambda(\p-m)}\leq e^{2\Lambda \kappa}.$  

By the local boundedness provided by Theorem 8.17 in \cite{GT} or Theorem 4.14 in \cite{HL}, we obtain, for any $p>
0$, 
\begin{align*} \sup_{B_{s/2}} \underline{v}
\le C\big[ {s}^{-\frac np}\|\underline{v}\|_{L^p(B_{s})} +
e^{2\Lambda \kappa} {s}^{2} \big].\end{align*}
With $\underline{v}=\underline{\xi}(v)$ and $1\leq \underline{\xi}'(v)\leq e^{2\Lambda \kappa}$, 
by the mean value theorem, we have $v\le \underline{v}\le e^{2\Lambda \kappa}v$, and hence, 
\begin{align}\label{eq-estimate-v1} \sup_{B_{s/2}}  {v}
\le Ce^{2\Lambda \kappa}\big[ {s}^{-\frac np}\|{v}\|_{L^p(B_{s})} + {s}^{2} \big].\end{align}

Next, we set $\overline{v}=\overline{\xi}(v)$ for some function $\overline{\xi}$ to be determined, 
with properties $\overline{\xi}\ge 0$, $\overline{\xi}'\ge 0$, and $\overline{\xi}''\le 0$.  
Then for $\psi \in
H_0^1(B_{s}) \cap L^\infty(B_{s})$ with $\psi \ge 0$, we have
\begin{align*} \int a_{ij} D_i\overline{v}D_j \psi\, d\o&= \int a_{ij} D_i v \overline{\xi}'(v)D_j
 \psi\, d\o\\
&= \int a_{ij} D_i v D_j (\overline{\xi}'(v) \psi)\, d\o -  \int
a_{ij}
D_i v D_j v \overline{\xi}''(v)\psi\, d\o\\
&\ge \int [-e^{\p}-2|Dv|^2] \overline{\xi}'(v)\psi \, d\o-  \int\Lambda^{-1}
|Dv|^2 \overline{\xi}''(v)\psi\, d\o.
\end{align*}
We take $\overline{\xi}$ such that $\overline{\xi}''=-2\Lambda\overline{\xi}'$ with 
$\overline{\xi}(0)=0$ and $\overline{\xi}'(0)=1$. Then, it holds
$$\overline{\xi}(\tau)=\f1{2\Lambda}(1-e^{-2\Lambda {\color{black}\tau}}).$$ 
Hence, 
for $\psi \in
H_0^1(B_{s}) \cap L^\infty(B_{s})$ with $\psi \ge 0$, we have
\begin{align*} \int a_{ij} D_i\overline{v}D_j \psi\, d\o\ge \int -e^{\p}\overline{\xi}'(v)\psi \, d\o
\ge \int -e^{\kappa}\psi \, d\o.
\end{align*}
By the weak Harnack inequality provided by Theorem 8.18 in \cite{GT} or Theorem 4.15 in \cite{HL}, we obtain, for any
$p\in (0, \f{n}{n-2})$,
$${s}^{-\frac np}\|\overline{v}\|_{L^p(B_{s})} \le C\Big[
\inf_{B_{s/2}}  \overline{v}  + e^{\kappa} {s}^{2}\Big].$$ 
With $\overline{v}=\overline{\xi}(v)$ and $e^{-2\Lambda\kappa}\leq \overline{\eta}'(\tau)\leq 1$,  
by the mean value theorem, we have $e^{-2\Lambda \kappa}v\le \overline{v}\le v$, and hence, 
\begin{align}\label{eq-estimate-v2}{s}^{-\frac np}\|{v}\|_{L^p(B_{s})} \le Ce^{2\Lambda\kappa}\Big[
\inf_{B_{s/2}}  {v}  +  e^{\kappa}{s}^{2}\Big].\end{align}
By combining \eqref{eq-estimate-v1} and \eqref{eq-estimate-v2}, we get 
\begin{align*}\sup_{B_{s/2}}  {v} \le Ce^{4\Lambda\kappa}
\inf_{B_{s/2}}  {v}  +  Ce^{4\Lambda\kappa+\kappa} {s}^{2}.\end{align*}
With $v=\p-m(s)$, we thus obtain
\begin{equation}\label{eq-estimate-v3}
M\big(\frac{s}{2}\big) - m(s) \le  Ce^{4\Lambda\kappa}\Big[m\big(\frac{s}{2}\big)
 -  m(s)\Big]  +  Ce^{4\Lambda\kappa+\kappa} {s}^{2}.
\end{equation}

Write $M=M(s)=\max_{B_s}\p$ for brevity and set  
$$w=M-\p.$$
For any $\psi \in
H_0^1(B_{s}) \cap L^\infty(B_{s})$ with $\psi \ge 0$, by \eqref{eq-main-nonlinear} and \eqref{eq-estimate-F}, we have  
\begin{align*} 
\int a_{ij} D_iwD_j \psi\, d\o\le \int \sqrt{|\gamma|}\Big[-1+\f12fe^{\p}+|Dw|^2+1\Big] \psi\, d\o
\le  \int [e^{\p}+2|Dw|^2] \psi\, d\o,
\end{align*}
and 
\begin{align*} 
\int a_{ij} D_iwD_j \psi\, d\o\ge \int \sqrt{|\gamma|}\Big[-1+\f12fe^{\p}-|Dv|^2-1\Big] \psi\, d\o
\ge  \int [-4-2|Dw|^2] \psi\, d\o.
\end{align*}
Proceeding similarly as above for $v$,  for $p\in(0,\f{n}{n-2})$ we obtain 
\begin{align}\label{eq-estimate-w1} \sup_{B_{s/2}}  {w}
\le Ce^{2\Lambda \kappa}\big[ {s}^{-\frac np}\|{w}\|_{L^p(B_{s})} + e^{\kappa}{s}^{2} \big]\end{align}
and 
\begin{align}\label{eq-estimate-w2}{s}^{-\frac np}\|{w}\|_{L^p(B_{s})} \le Ce^{2\Lambda\kappa}\Big[
\inf_{B_{s/2}}  {w}  +  {s}^{2}\Big].\end{align}
By combining \eqref{eq-estimate-w1} and \eqref{eq-estimate-w2}, we have
\begin{align*}\sup_{B_{s/2}}  {w} \le Ce^{4\Lambda\kappa}
\inf_{B_{s/2}}  {w}  +  Ce^{4\Lambda\kappa+\kappa} {s}^{2},\end{align*}
and using $w=M(s)-\p$, we then get
\begin{equation}\label{eq-estimate-w3}
M(s) - m \big(\frac{s}{2}\big) \le Ce^{4\Lambda\kappa}
 \Big[ M(s) - M \big(\frac{s}{2}\big)\Big]
 +  Ce^{4\Lambda\kappa+\kappa} {s}^{2}.
\end{equation}

Set, for any $s \in (0,1)$, 
\begin{equation}\label{h=M-m} 
h(s) = M(s) - m(s).
\end{equation} 
Adding \eqref{eq-estimate-v3} and \eqref{eq-estimate-w3}, we get
$$
h(s)+h\big(\frac{s}{2}\big)
\le Ce^{4\Lambda\kappa}\Big[h(s)
-h\big(\frac{s}{2}\big)\Big]+Ce^{4\Lambda\kappa+\kappa} {s}^{2}, 
$$
and hence
$$
h\big(\frac{s}{2}\big) 
\le \gamma h(s)+e^{\kappa} {s}^{2}, 
$$
where 
$$ \gamma = \frac{Ce^{4\Lambda\kappa}-1}{Ce^{4\Lambda\kappa}+1} < 1.$$
With this $\gamma$, we choose $\mu\in (0,1)$ such that 
$$\alpha:=(1-\mu)\log \gamma/\log(1/2) <2\mu.$$
Since $\gamma=1-2(1+Ce^{4\Lambda\kappa})^{-1}$, where $\kappa$ is large and  
$C$ in \eqref{eq-estimate-v3} and \eqref{eq-estimate-w3} could be chosen to be large, 
we can write $\alpha$ as in \eqref{eq-definition-alpha},  
i.e., $\alpha=\epsilon_0e^{-4\Lambda\kappa}$
for some small $\epsilon_0>0$. 
We then apply Lemma 8.23 in \cite{GT} or Lemma 4.19 in \cite{HL} and obtain, for any $s\in (0, 1/2]$,
\begin{equation}\label{omega hat 1}
h(s)\le C{s}^{\alpha}\big\{h(1)
+ e^{\kappa} s^{2\mu}\big\}\leq Ce^{\kappa} {s}^{\alpha},
\end{equation}
where  we used $h(1)\leq M(1)\leq \kappa\leq e^{\kappa}$ for the second inequality.
We then conclude \eqref{eq-Holder-semi-norm}.
\end{proof}

\subsection{Interior Gradient Estimates} 
Our next goal is to derive the gradient estimate of  $\p$.  We introduce an extra condition 
\begin{align}\label{eq-assumption-a-ij-Lambda-1}
|D_{\o}a_{kl}(\o,\p)|+|\partial_{\p} a_{kl}(\o,\p)|\le{\color{black} \Lambda_1},
\end{align} 
for some positive constant $\Lambda_1$. In the proof below, we will adapt some relevant results in \cite{HL}
{\color{black}and employ the characterization of H\"older continuous functions by Campanato.}

\begin{theorem}\label{theorem-Holder-gradient} Suppose $\p\in H^1(B_1)$ is a positive solution to
\eqref{eq-main-nonlinear} with \eqref{eq-main-bound-u} satisfied.  
Then,  
\begin{equation}\label{eq-estimate-final}
|\p|_{C^{1,1/3}(B_{1/2})}\le (Ce^{\kappa})^{Ce^{4\Lambda\kappa}},\end{equation}
where $C$ is a positive constant depending only on $n$, $\Lambda$, {\color{black} and $\Lambda_1$}. 
\end{theorem} 

\begin{proof} We adopt same notations as in the proof of Theorem \ref{theorem-Holder-solution}. 
We first proceed as in the proof of Theorem \ref{theorem-Holder-solution} to obtain an interior estimate in
the $L^2$-norm for the gradient of $\p$. 
For any $\psi \in
H_0^1(B_{1}) \cap L^\infty(B_{1})$ with $\psi \ge 0$, by \eqref{eq-main-nonlinear} and \eqref{eq-estimate-F}, we have   
\begin{align*} 
\int a_{ij} D_i\p D_j \psi\, d\o\le \int \sqrt{|\gamma|}\Big[1-\f12fe^{\p}+|D\p|^2+1\Big] \psi\, d\o
\le  \int [4+2|D\p|^2] \psi\, d\o.
\end{align*}
Taking $\underline{\xi}$ as in \eqref{eq-expression-lower-eta}, we have 
\begin{align*} \int a_{ij} D_i(\underline{\xi}(\p))D_j \psi\, d\o
\le \int 4\underline{\xi}'(\p)\psi \, d\o,
\end{align*}
which is equivalent to
\begin{align*} \int a_{ij} e^{2\Lambda\p}D_i\p D_j \psi\, d\o
\le \int 4e^{2\Lambda\p}\psi \, d\o.
\end{align*}
Take $\psi=\p \mu^2$  for some $\mu\in C^1_0(B_1)$
with $\mu=1$ on $B_{3/4}$. By the Cauchy-Schwarz inequality, we get 
\begin{equation*}
\int  e^{2\Lambda\p}|D\p|^2 \mu^2\, d\o
\le C\int  e^{2\Lambda\p}\big(\p^2|D \mu|^2+\p \mu^2\big)\, d\o,
\end{equation*} 
and hence 
\begin{equation}\label{eq-L2-estimate-gradient}
\int_{B_{3/4}}|D\p|^2\, d\o\le C\kappa^2 e^{2\Lambda\kappa},\end{equation}
where $C$ is a positive constant depending only on $n$ and $\Lambda$.  

We next derive an estimate of the H\"older norm of the gradient and prove 
\begin{equation}\label{eq-estimate-gradient}
|D\p|_{C^{1/3}(B_{1/2})}\le (Ce^{\kappa})^{Ce^{2\Lambda\kappa}}.\end{equation}
The proof below is similar to that of 
Theorem 4.24 in \cite{HL}. The difference is that in Theorem 4.24 of \cite{HL} $a_{ij}=a_{ij}(\o)$, while here $a_{ij}=a_{ij}(\o, \p(\o))$.  
We need to estimate $a_{ij}(\o_0, \p(\o_0))-a_{ij}(\o, \p(\o))$ and to control $\p(\o)-\p(\o_0)$ 
with Theorem \ref{theorem-Holder-solution}. In the following, $C$ is a positive constant depending only on $n$ and $\Lambda$.
We denote by $\p_{\o, s}$ and $(D\p)_{\o,s}$ the average of $\p$ and $D\p$ over $B_{s}(\o)$, respectively.

For any $B_{s}(\o_0) \subset B_1$,
take $\h w$ such that for any 
$\psi \in H^1_0\big(B_{ {\color{black}s}}(\o_0)\big)$ it holds
\begin{equation}\label{hat W linear}
\int_{B_{s}(\o_0)} a_{ij} (\o_0, \p(\o_0)) D_i \h w D_j\psi \, d\o= 0, 
\end{equation}
with $\h w-\p\in
H^1_0\big(B_{s}(\o_0)\big)$. The maximum principle implies 
\begin{equation}\label{W hat}
\inf_{B_{s}(\o_0)} \p \le \h w \le \sup_{B_{s}(\o_0)}\p\quad \text{in} \ B_{s}(\o_0). 
\end{equation}
For convenience, set $\h v=\p-\h w$ and recall \eqref{h=M-m}  
$$h(s)=h(\o_0;s)=\sup_{B_{s}(\o_0)}\p-\inf_{B_{s}(\o_0)} \p.$$ 
Then, $\h v \in H^1_0
\big(B_{s}(\o_0)\big)$
and by (\ref{W hat}) we get
\begin{equation}\label{eq-oscillation}\sup_{B_{s}(\o_0)} |\h v| \le h(s).\end{equation}
By  Corollary 3.11 in \cite{HL}, we have for any $0<\rho\le s$
\begin{equation}\label{eq-estimates-gradient}\int_{B_\rho(\o_0)} |D\p|^2 \, d\o\le c\Big[\Big(\frac{\rho}{s}\Big)^n
\int_{B_{s}(\o_0)}|D\p|^2\, d\o +
\int_{B_{s}(\o_0)}|D\h v|^2\Big]\, d\o,\end{equation} and
\begin{align}\label{eq-estimates-gradient-comparison}\begin{split}
\int_{B_\rho(\o_0)}|D\p-(D\p)_{\o_0,\rho}|^2 \, d\o
&\le c \Big[\Big(\frac{\rho}{s}\Big)^{n+2} \int_{B_{s}(\o_0)} |D\p - (D\p)_{\o_0,{ {\color{black}s}}}|^2\, d\o\\
&\qquad+\int_{B_{s}(\o_0)} |D\h v|^2\, d\o\Big].
\end{split}
\end{align}
Moreover, by  \eqref{hat W linear} and \eqref{eq-main-nonlinear}, we get 
for any $\psi \in H^1_0\big(B_{s}(\o_0)\big) \cap L^\infty
\big(B_{s}(\o_0)\big)$ 
\begin{align*} 
&\int_{B_{s}(\o_0)} a_{ij}(\o_0, \p(\o_0)) D_i \h vD_j\psi\, d\o\\
&\qquad=\int_{B_{s}(\o_0)} a_{ij}(\o_0, \p(\o_0)) D_i \p D_j\psi \, d\o\\
&\qquad=\int_{B_{s}(\o_0)}\sqrt{|\gamma|} \big[1-\f12f(\o)e^{\p(\o)}+F(\o,\p(\o), D\p(\o))\big]\psi\, d\o\\
&\qquad\quad+\int_{B_{s}(\o_0)}\big[a_{ij}(\o_0, \p(\o_0))-a_{ij}(\o, \p(\o))\big] D_i \p D_j\psi\, d\o.
\end{align*}
Note that, for any $\o\in B_{s}(\o_0)$ it holds
$$\big|a_{ij}(\o_0, \p(\o_0))-a_{ij}(\o, \p(\o))\big|\le C[|\o-\o_0|+ |\p(\o)-\p(\o_0)|]
\le C[s+  h(s)].$$
Taking $\psi = \h v$, together with \eqref{eq-estimate-F}, we obtain 
\begin{align*}
\int_{B_{s}(\o_0)} |D\h v|^2 \, d\o\le C\int_{B_{s}(\o_0)}\big[[s+ h(s)]|D\p||D\h v|
+ |D\p|^2|\h v| +e^{\kappa}|\h v|\big]\, d\o.
\end{align*}
We apply the Cauchy-Schwarz inequality to the first term in the right-hand side, substitute $\h v$ 
in the second term with \eqref{eq-oscillation}, and apply the Sobolev inequality to $\h v$ in the third term as follows:  
$$ \int_{B_{s}(\o_0)}|\h v|\, d\o
\le c{s}^{\frac{n+2}{2}}\Big(\int_{B_{s}(\o_0)}|D\h v|^{2}\, d\o\Big)^{\frac12}.$$
Then, we obtain
\begin{equation}\label{eq-estimates-gradient-improved}
\int_{B_{s}(\o_0)} |D\h v|^2 \, d\o\le C\big(h(s)+[h(s)]^2+s^2\big)\int_{B_{s}(\o_0)}|D\p|^2 \, d\o+Ce^{2\kappa}{s}^{n+2}.
\end{equation}
Therefore, by \eqref{eq-estimates-gradient} and \eqref{eq-estimates-gradient-improved}, 
we have for any $0<\rho\le s$ 
$$\int_{B_\rho(\o_0)} |D\p|^2 \, d\o\le C\Big[\Big({\frac{\rho}{s}}\Big)^n
+h(s)+[h(s)]^2+s^2\Big]\int_{B_{s}(\o_0)} |D\p|^2\, d\o +
Ce^{2\kappa}{s}^{n+2}.$$

Take any $\delta\in (0,1)$. For a small $\epsilon_\delta>0$, 
using \eqref{omega hat 1} we can choose some $s_\delta$ small such that $B_{{s}_\delta}(\o_0)\subset B_{3/4}$ and 
\begin{equation}\label{eq-condition-delta}
h(s_\delta)+[h(s_\delta)]^2+s_\delta^2<\epsilon_\delta. \end{equation}
By Lemma 3.4 in \cite{HL},  we have for any $s\le {s}_\delta$
\begin{equation*}
\int_{B_{s}(\o_0)} |D\p|^2\, d\o\le C{s}^{n-2+2\delta}\Big[{s}_\delta^{-(n-2+2\delta)}\int_{B_{{s}_\delta}(\o_0)} |D\p|^2\, d\o
+ e^{2\kappa}\Big].
\end{equation*}
By the choice of $\alpha$ in \eqref{eq-definition-alpha} and \eqref{eq-Holder-semi-norm}, 
using \eqref{omega hat 1} together with $\kappa\gg 1$, 
we can take  suitable positive constants $C_1$ and $C_2$ depending on $n$, $\Lambda$ and 
\begin{equation}\label{eq-definition-r}{s}_\delta=(C_1e^{\kappa})^{-C_2e^{4\Lambda\kappa}},\end{equation} 
so that \eqref{eq-condition-delta} is satisfied.  By \eqref{eq-L2-estimate-gradient}, we get 
\begin{equation}\label{eq-estimates-gradient-decay}
\int_{B_{s}(\o_0)} |D\p|^2\, d\o\le C{s}^{n-2+2\delta}\big[\kappa^2e^{2\Lambda\kappa} {s}_\delta^{-(n-2+2\delta)}
+ e^{2\kappa}\big].
\end{equation}
By the Poincar\'e inequality, we have for any $s\le {s}_\delta$ 
with $B_{{s}_\delta}(\o_0)\subset B_{3/4}$ 
\begin{equation*}
\int_{B_{s}(\o_0)} |\p-\p_{\o_0,s}|^2\, d\o\le C{s}^{n+2\delta}\big[\kappa^2e^{2\Lambda\kappa} {s}_\delta^{-(n-2+2\delta)}
+ e^{2\kappa}\big].
\end{equation*}
Hence, by Theorem 3.1 in \cite{HL}, we obtain for any $s\le {s}_\delta$ 
with $B_{{s}_\delta}(\o_0)\subset B_{3/4}$ 
\begin{equation}\label{eq-oscillation-improved}
h(\o_0; s)=\sup_{B_s(\o_0)}\p-\inf_{B_s(\o_0)}\p
\le C\big[\kappa e^{\Lambda\kappa} {s}_\delta^{-(\frac n2-1+\delta)}+e^{\kappa}\big]{s}^\delta.
\end{equation}
With \eqref{eq-estimates-gradient-improved},
\eqref{eq-estimates-gradient-decay}, and \eqref{eq-oscillation-improved}, 
we have for any $s\le {s}_\delta$ 
with $B_{{s}_\delta}(\o_0)\subset B_{3/4}$ 
\begin{align*} \int_{B_{s}(\o_0)} |D\h v|^2 \, d\o\le  C\kappa^{3}e^{3\Lambda\kappa} {s}_\delta^{-{\f32}(n-2+2\delta)}
{s}^{n-2+3\delta}+Ce^{2\kappa}{s}^{n+2}. 
\end{align*}
We take $\delta\in (2/3, 1)$ 
and write $\alpha'=3\delta/2-1$. 
Note that $\alpha'\in (0, 1/2)$ and $\alpha'$ can be as close as to 1/2 by taking $\delta$ as close as to 1. 
With \eqref{eq-estimates-gradient-comparison} we obtain, 
for any $\rho\leq s\le {s}_\delta$ 
with $B_{{s}_\delta}(\o_0)\subset B_{3/4}$, 
\begin{align*}\int_{B_\rho(\o_0)}|D\p-(D\p)_{\o_0,\rho}|^2\, d\o &\le C\Big(\frac{\rho}{s}\Big)^{n+2}\int_{B_{s}(\o_0)} |D\p - (D\p)_{\o_0,s}|^2\, d\o\\
&\qquad+C\kappa^{3}e^{3\Lambda\kappa} {s}_\delta^{-\f32(n-2+2\delta)}{s}^{n+2\alpha'}+Ce^{2\kappa} {s}^{n+2}.
\end{align*}
By Lemma 3.4 in \cite{HL}, together with \eqref{eq-L2-estimate-gradient}, we have for any $s\le {s}_\delta$ 
with $B_{{s}_\delta}(\o_0)\subset B_{3/4}$  
\begin{align*}\int_{B_{s}(\o_0)}|D\p-(D\p)_{\o_0,s}|^2 \, d\o\le 
C\kappa^{3}e^{3\Lambda\kappa} {s}_\delta^{-\f32(n-2+2\delta)}{s}^{n+2\alpha'}. 
\end{align*}
By Theorem 3.1 in \cite{HL}, together with (\ref{eq-L2-estimate-gradient}), we obtain for any $\o\in B_{1/2}$ 
\begin{equation}\label{C11}|D\p(\o)|\le C\big[\kappa^{\f32}e^{\f32\Lambda\kappa} {s}_\delta^{-\f34(n-2+2\delta)}+\kappa e^{\Lambda\kappa}\big]
\leq C \kappa^{\f32}e^{\f32\Lambda\kappa} {s}_\delta^{- \f34(n-2+2\delta)},\end{equation}
and for any $\o_1,\o_2\in B_{1/2}$ with $|\o_1-\o_2|<{s}_\delta$,
$$|D\p(\o_1)-D\p(\o_2)|\le C\kappa^{\f32}e^{\f32\Lambda\kappa} {s}_\delta^{-\f34(n-2+2\delta)}|\o_1-\o_2|^{\alpha'}.$$
Moreover, for any $\o_1,\o_2\in B_{1/2}$ it holds
\begin{equation}\label{C12}
|D\p(\o_1)-D\p(\o_2)|\le C\kappa^{\f32}e^{\f32\Lambda\kappa} {s}_\delta^{-\f34(n-2+2\delta)}{s}_\delta^{-\alpha'}
|\o_1-\o_2|^{\alpha'}.
\end{equation}
Recall ${s}_\delta$ as in \eqref{eq-definition-r} 
and note $0<\alpha'<1/2$. We can choose $\alpha'=1/3$ and, with (\ref{C11}) and (\ref{C12}), we hence conclude
\eqref{eq-estimate-gradient}. 
\end{proof}

\subsection{Interior Schauder Estimates}
Our final goal is to prove the following Schauder estimates. 
We introduce an extra assumption 
\begin{align}\label{eq-assumption-a-ij-Lambda-2}
|D^2_{\o}a_{kl}(\o,\p)|+|D_{\o}\partial_{\p} a_{kl}(\o,\p)|+|\partial^2_{\p} a_{kl}(\o,\p)|\le \Lambda_2,
\end{align} 
for some positive constant $\Lambda_2$. 

\begin{theorem}\label{theorem-Holder-second derivative} 
Let $\p\in H^1(B_1)$ be a positive solution to
\eqref{eq-main-nonlinear} with \eqref{eq-main-bound-u} satisfied, and $f\in C^{1/3}(B_1)$.  
Then,  
\begin{equation}\label{eq-estimate-second-derivative}
|\p|_{C^{2,1/3}(B_{1/2})}\le (Ce^{\kappa})^{Ce^{4\Lambda\kappa}},\end{equation}
where $C$ is a positive constant depending only on $n$, 
$\Lambda$, {\color{black} $\Lambda_1$, $\Lambda_2$,  
the $C^{1}$-norms of $\sqrt{|\gamma|}$, $c_{ij}$, $c_i$, and $c_0$ on 
$B_1\times \mathbb R_+$,}
and the $C^{1/3}$-norm of $f$ on $B_1$.
\end{theorem}

\begin{proof} 
We write \eqref{eq-main-nonlinear} in nondivergence form
\begin{equation}\label{Theorem 4.7 application}
a_{ij} (\o,\p) \partial_{ij}\p = {h}(\omega), 
\end{equation}
where 
\begin{equation*}
{h}(\o)=-\partial_{j}a_{ij}(\o,\p)\partial_i \p
-\partial_{\p} a_{ij}(\o,\p)\partial_i \p\partial_j \p- \sqrt{|\gamma|}\Big[1-\f12fe^{\p}+F(\o, \p, D\p)\Big].
\end{equation*}
We treat the equation \eqref{Theorem 4.7 application} as a linear equation of $\p$, 
with ${h}$ as a given function defined in $B_1$. 
By the explicit expression of $F$ and Theorem \ref{theorem-Holder-gradient}, we have 
$$|{h}|_{C^{1/3}(B_{3/4})}\le (Ce^{\kappa})^{Ce^{4\Lambda\kappa}}.$$
If we quote the standard interior Schauder estimate, we will get  
\begin{equation*}
|\p|_{C^{2,1/3}(B_{1/2})}\le C_*\big[|\p|_{L^\infty(B_{3/4})}+|h|_{C^{1/3}(B_{3/4})}\big]
\le C_*(Ce^{\kappa})^{Ce^{4\Lambda\kappa}},\end{equation*}
where $C_*$ is a positive constant depending only on $n$, $\Lambda$, and the $C^{1/3}$-norm of $a_{ij}$ on $B_{3/4}$. 
Since the leading coefficient $a_{ij}$ is also a function of $\p$, i.e., $a_{ij}=a_{ij}(\o, \p)$, 
then $C_*$ also depends on $\p$.  

Instead of quoting the interior Schauder estimate, we  need to 
examine its proof by the technique of {\it freezing} coefficients, 
and to keep track of the dependence on the H\"older norms of the leading coefficients.  
In this process, 
all contributions due to the dependence of $a_{ij}$ on $\p$ can be absorbed by 
$(Ce^{\kappa})^{Ce^{4\Lambda\kappa}}$. 
Therefore, we obtain the desired estimate \eqref{eq-estimate-second-derivative}. 
\end{proof} 

We are ready to prove Theorem \ref{theorem-Schauder-phi}. 

\begin{proof}[Proof of Theorem \ref{theorem-Schauder-phi}] 
Fix an arbitrary ball $B$ on $\mathbb S^2$. In local coordinates within $B$, 
we write \eqref{main eqn lambda} as \eqref{u equation lambda main}. Specifically, we set 
$$a_{ij}(\o,\p)=\sqrt{|\gamma(\o, u)|}\gamma^{ij}(\o, u),$$
with $u=1-R=1-\ub ae^{-\p}.$ 
By \eqref{Lambda-gamma-1}, \eqref{Lambda-gamma-2}, and \eqref{Lambda-gamma-3}, we obtain 
\eqref{eq-assumption-a-ij-Lambda-1} and \eqref{eq-assumption-a-ij-Lambda-2} by renaming 
$\Lambda_1$ and $\Lambda_2$ if necessary. 
{\color{black} Next, we examine the function $F(\p,\ub)$ given by \eqref{eq-expression-cal-S-ub-nonlinear-v2}. 
Note that $F$ can be viewed as a linear combination of $\nab_i\p\nab_j\p, \nab_i\p, 1$ with coefficients given by 
$\gamma^{ij}$, $R^{-1}\Omega\chibh_{kl}$, $\eta_i$, $R\alpha_1$, and $R\alpha_2$, which are functions of $\o$ and $u=1-R=1-\ub ae^{-\p}$. 
In order to get an estimate of the $C^{1/3}$-norm of $F$, we need to bound these coefficients and their derivatives with respect to $\o$ and $\p$. 
Note $\partial_{\p}=R\partial_u$. Hence, we need to bound 
\begin{align*} 
&R^{-1}\Omega\chibh_{kl},\, \eta_k,\, R\alpha_1,\, R\alpha_2,\\
&R^{-1}\nab_i(\Omega\chibh_{kl}),\, \nab_i\eta_k,\, R\nab_i\alpha_1,\, R\nab_i\alpha_2,\\
&\partial_u(\Omega\chibh_{kl}),\, R\partial_u\eta_k,\, R^2\partial_u\alpha_1,\, R^2\partial_u\alpha_2.
\end{align*}
By \eqref{A1} and \eqref{A2}, these quantities are bounded by $e^{\p}a^{-\f13}$. 
Hence, we have an estimate of the $C^{1/3}$-norm of $F$.}
Therefore, we can apply Theorem \ref{theorem-Holder-solution}, 
Theorem \ref{theorem-Holder-gradient}, and Theorem \ref{theorem-Holder-second derivative} successively, 
and obtain the desired result.\end{proof}

\section{The Linearized Equation}\label{sec-linearized-equations} 

In this section, we derive and analyze the linearized equation. 
As in previous sections, we fix $\ub$. 

We first compute  the linearized operator. The approach here is different from that in \cite{An17}, where techniques via commutator formulas are used. Here we adopt a more direct calculation. Let $\mathcal S(\phi)$
be the operator defined by  \eqref{eq-expression-cal-S-ub-v2}. 

\begin{proposition}\label{form of linearization}
Let $\p$ be a given function on $\mathbb S^2$. 
The linearized operator $\partial_{\p} \mathcal S(\p)$ is of the form
\begin{equation}\label{eq-linearization-S}
\partial_{\p} \mathcal S(\p)[w]
=\Delta_{\gamma}w-\f12f(\o)e^{\p}w+c^{i}\nab_i w+cw,
\end{equation}
where the connection is with respect to the metric $\gamma$, and 
$c=c(\o,\p,\ub)$ and $c^i=c^i(\o,\p,\ub)$ are functions satisfying 
$$|c|+|c^i|\leq Ca^{-\f13} e^{\p}{\color{black} (|\nabla^2_{\gamma}\p|+|\nabla_\gamma\p|^2+1)}.$$
\end{proposition}

\begin{proof} The proof is divided into four steps. 

{\it Step 1. Setup.} Recall the relations given by \eqref{eq-relation-g-gamma}  
and \eqref{eq-relation-u-R-phi}; namely, $g=R^{2}\gamma$, and 
\begin{equation*}u=1-R, \, R=\ub a e^{-\p}.\end{equation*} 
By \eqref{eq-expression-cal-S-ub-v2} and \eqref{eq-expression-cal-S-ub-nonlinear-v2}, we write
\begin{align}\label{eq-expression-S-I-1-5}\mathcal S(\p)=I_1+I_2+I_3+I_4+I_5,\end{align}
where 
\begin{align}\label{eq-expression-S-I-1}
I_1=\D_{\gamma}\phi+1-\f12fe^{\p}\end{align}
and 
\begin{align}\label{eq-expression-S-I-2-5}\begin{split}
I_2=2R^{-1}\Omega\chibh_{kl} \gamma^{ik}\gamma^{jl}\nab_i \p\nab_j \p, \quad
&I_3=2\gamma^{ij}\eta_i \nab_j \p,\\
I_4=R\alpha_1 |\nab_\gamma \p|^2,\quad
&I_5=R\alpha_2. 
\end{split}\end{align}

{\it Step 2. Parametrized functions.}  
Now, we fix functions $\p$, $R$ with $R=\ub a e^{-\p}, u=1-R$.  Consider $\p(\e)=\p+\e w$. Then, 
$R(\e)=Re^{-\e w}$. We proceed to compute the linearized operator 
$$\partial_{\p} \mathcal S(\p)[w]
=\f{d\,}{d\epsilon}\Big|_{\epsilon=0}\mathcal S(\p+\epsilon w).$$ 
For brevity, we denote $\dot{\,}=\f{d\,}{d\epsilon}\Big|_{\epsilon=0}$. Hence, 
$$\dot{\p}=w, \, \dot{R}=-Rw,\, \dot{u}=Rw.$$
Note that $\gamma, \chib, \chi, \eta, \O, \omb$ all depend on $\e$ through $u=1-R$.  
We now derive several general formulas in terms of $\dot{\gamma}$. 

Note that  
$$\dot{\gamma}^{ij}=-\gamma^{ik}\gamma^{jl}\dot{\gamma}_{kl}.$$ 
By $|\nabla_\gamma\phi|^2=\gamma^{ij}\nab_i\p\nab_j\p$, we obtain
\begin{equation}\label{eq-variation-gradient-square}
(|\nabla_\gamma\phi|^2)^{\boldsymbol\cdot}=2\gamma^{ij}\nab_i\p\nab_jw
-\gamma^{ik}\gamma^{jl}\dot{\gamma}_{kl}\nab_i\p\nab_j\p.\end{equation}
By applying Lemma \ref{lemma-variation-Laplacians} to $\gamma$, together with $\dot{\phi}=w$, we have 
\begin{equation}\label{eq-variation-Laplace}
(\Delta_{\gamma}\p)^{\boldsymbol\cdot}=\Delta_\gamma w-\gamma^{ik}\gamma^{jl}\dot{\gamma}_{kl}\nabla_{ij}\p
-\gamma^{ij}\gamma^{kl}\nabla_i\dot{\gamma}_{jl}\nabla_k\p+\frac12\gamma^{ij}\gamma^{kl}\nabla_l\dot{\gamma}_{ij}\nabla_k\p.
\end{equation}

Recall \eqref{eq-variation-g-u}, i.e.,  
\begin{equation*}
\partial_ug_{kl}=\O\tr_g\chib g_{kl}+2\O \chibh_{\,kl}.
\end{equation*}
Then, with $g=R^2\gamma$, 
\begin{equation*}
\partial_u\gamma_{kl}=R^{-1}(R\O\tr_g\chib+2)\gamma_{kl}+2R^{-2}\O \chibh_{\,kl},
\end{equation*}
and hence 
\begin{equation*}
\dot{\gamma}_{kl}=\dot{u}\partial_u \gamma_{kl}
=w(R\O\tr_g\chib+2)\gamma_{kl}+2wR^{-1}\O \chibh_{\,kl}.
\end{equation*}
A substitution in \eqref{eq-variation-gradient-square} and \eqref{eq-variation-Laplace} yields 
\begin{equation}\label{eq-variation-gradient-square1}
(|\nabla_\gamma\phi|^2)^{\boldsymbol\cdot}=2\gamma^{ij}\nab_j\p\nab_iw-w(R\O\tr_g\chib+2) |\nab_\gamma\p|^2-2
wR^{-1}\O \chibh_{\,kl}
\gamma^{ik}\gamma^{jl}\nab_i\p\nab_j\p,\end{equation}
and 
\begin{equation}\label{eq-variation-Laplace1}\begin{split}
(\Delta_{\gamma}\p)^{\boldsymbol\cdot}&=\Delta_\gamma w-2R^{-1}\O\chibh_{kl}\gamma^{ik}\gamma^{jl}\nab_j \p\nab_iw\\
&\qquad -w\big[(R\O\tr_g\chib+2)\Delta_\gamma\p+2R^{-1}\nab_i(\O\chibh_{kl})\gamma^{ik}\gamma^{jl} \nab_j\p\\
&\qquad+2R^{-1}\O \chibh_{kl}\gamma^{ik}\gamma^{jl}(\nab_{ij}\p
+\nab_i\p \nab_j \p)\big].
\end{split}
\end{equation}

{\it Step 3. Computations of $\dot{I}_i$, for $i=1, \cdots, 5$.}
We first compute $\dot{I}_1$. By \eqref{eq-expression-S-I-1} and \eqref{eq-variation-Laplace1}, we have
\begin{equation}\label{eq-variation-I1}
\begin{split}
\dot{I}_1&=\Delta_{\gamma}w-\f12fe^{\p}w-2R^{-1}\O\chibh_{jl}\gamma^{ij}\gamma^{kl}\nab_k \p\nab_iw
-w(2+R\O\tr_g\chib)\Delta_{\gamma}\p\\
&\qquad -2w\big[R^{-1}\nab_i(\O\chibh_{jl})\gamma^{ij}\gamma^{kl} \nab_k\p
+R^{-1}\O \chibh_{kl}\gamma^{ik}\gamma^{jl}(\nab_{ij}\p
+ \nab_i\p \nab_j \p)\big].
\end{split}
\end{equation}
We note that the first two terms in $\dot{I}_1$ are the first two terms on the right-hand side of \eqref{eq-linearization-S}. 
We will  demonstrate that the rest terms in $\dot{I}_1$ and all other $\dot{I}_i$, for $i=2, \cdots, 5$, 
are in the forms of the last two terms on the right-hand side of \eqref{eq-linearization-S}. 

Without presenting details of standard calculation process, here we list the expressions of $\dot{I}_i$, for $i=2, \cdots, 5$. We have 
\begin{align}\label{eq-dot-I2-5}\begin{split}
\dot{I}_2&=4R^{-1}\Omega\chibh_{kl} \gamma^{ik}\gamma^{jl}\nab_j\p\nab_i w\\
&\qquad+2w\big(\partial_u(\Omega\chibh_{kl})+R^{-1}\Omega\chibh_{kl}\big) \gamma^{ik}\gamma^{jl}\nab_i\p\nab_j \p\\
&\qquad -4wR^{-1}\Omega\chibh_{kl}\big((R\O\tr_g\chib+2)\gamma^{ki}
+2R^{-1}\Omega\chibh_{pq}\gamma^{ip}\gamma^{kq}\big)\gamma^{lj}\nab_i\p\nab_j \p,\\
\dot{I}_3&=2\gamma^{ij}\eta_j \nab_i w+2wR\partial_u\eta_i \gamma^{ij}\nab_j \p\\
&\qquad -2w\big((R\O\tr_g\chib+2)\gamma^{ij}
+2R^{-1}\Omega\chibh_{kl}\gamma^{ik}\gamma^{jl}\big)\eta_i \nab_j \p,\\
\dot{I}_4&=2R\alpha_1\gamma^{ij}\nab_i\p\nab_jw+w(R^2\partial_u\alpha_1-R\alpha_1)|\nab_{\gamma}\phi|^2\\
&\qquad -wR\alpha_1\big((R\O\tr_g\chib+2) \gamma^{ij}+2R^{-1}\O \chibh_{\,kl}\gamma^{ik}\gamma^{jl}\big)\nab_i\p\nab_j\p,\\
\dot{I}_5 
&=w(R^2\partial_u\alpha_2-R\alpha_2). 
\end{split}\end{align}

{\it Step 4. Expressions of coefficients.} 
We now collect $\dot{I}_i$, for $1\le i\le 5$, and write the sum in the form of \eqref{eq-linearization-S}. We obtain 
\begin{equation}\label{eq-expression-ci}
c_i=2R^{-1}\Omega\chibh_{kl} \gamma^{ik}\gamma^{jl}\nab_j\p+2\gamma^{ij}\eta_j
+2R\alpha_1\gamma^{ij}\nab_j\p,
\end{equation}
and 
\begin{equation}\label{eq-expression-c}
\begin{split}
c&=-(2+R\O\tr_g\chib)\Delta_{\gamma}\p\\
&\qquad-2\big[R^{-1}\O \chibh_{kl}\gamma^{ik}\gamma^{jl}(\nab_{ij}\p+\nab_i\p \nab_j \p)
+R^{-1}\nab_i(\O\chibh_{jl})\gamma^{ij}\gamma^{kl} \nab_k\p\big]\\
&\qquad -\big((R\O\tr_g\chib+2) \gamma^{ij}+2R^{-1}\O \chibh_{\,kl}\gamma^{ik}\gamma^{jl}\big)\\
&\qquad\qquad\cdot\big(4R^{-1}\Omega\chibh_{ip}\gamma^{pq}\nab_q\p\nab_j \p+2\eta_i \nab_j \p+R\alpha_1\nab_i\p\nab_j\p)\\
&\qquad+2\big(\partial_u(\Omega\chibh_{kl})+R^{-1}\Omega\chibh_{kl}\big) \gamma^{ik}\gamma^{jl}\nab_i\p\nab_j \p
+2R\partial_u\eta_i\gamma^{ij}\nab_j \p\\
&\qquad+(R^2\partial_u\alpha_1-R\alpha_1)|\nab_{\gamma}\phi|^2+R^2\partial_u\alpha_2-R\alpha_2.
\end{split}
\end{equation}
We point out that $c_1, c_2$, and $c$ can be viewed as linear combinations of  $\nab_{ij}\p$, $\nab_i\p\nab_j\p$, $\nab_i\p$, and $1$,
with coefficients given by $\gamma^{ij}$ and 
\begin{equation}\label{eq-expressions-coefficients}\begin{split}
&2+R\O\tr_g\chib,\, R^{-1}\Omega\chibh_{kl},\,  \eta_i,\, R\alpha_1,\, R\alpha_2,\, R^{-1}\nab_i(\Omega\chibh_{kl}), \\
&\partial_u(\Omega\chibh_{kl}),\, R\partial_u\eta_i,\, R^2\partial_u\alpha_1,\, 
R^2\partial_u\alpha_2. 
\end{split}\end{equation}
By \eqref{A1} and \eqref{A2}, 
all expressions in \eqref{eq-expressions-coefficients} are controlled by $a^{-\f13} e^{\p}$.
\end{proof} 

Next, we discuss the invertibility of the linearized operator $\partial_{\p} S(\p)$. 
We first examine the linear operator introduced by the first two terms on the left-hand side of \eqref{eq-linearization-S}. 
Set
\begin{equation}\label{eq-linearization-model}
L_0w=\D_\gamma w-\f12fe^{\phi}w.
\end{equation}
Here, we view $f$ and $\p$ as two fixed nonnegative functions on 
$\mathbb S^2$.  

\begin{lemma}\label{lemma-H1-estimates-L0} 
Let $f$ and $\p$ be given nonnegative functions on $\mathbb S^2$, with \eqref{eq-assumption-lower-bound-f} satisfied. 
Then the equation $L_0w=0$ admits only the trivial solution.  Moreover, the first eigenvalue $\mu_1$ of $-L_0$ is positive and simple, and 
has a positive eigenfunction.
\end{lemma}

\begin{proof} Integrating $-wL_0w$ by parts on $\mathbb S^2$, we have 
$$\int_{\mathbb S^2}\Big(|\nab_\gamma w|^2+\f12fe^{\phi}w^2\Big)d\o= -\int_{\mathbb S^2}wL_0wd\o.$$
By restricting the second integral to $B_{p}(\e)$, using $f\geq m$ in $B_p(\e)$ in \eqref{eq-assumption-lower-bound-f},  we get 
$$\int_{\mathbb S^2}|\nab_\gamma w|^2d\o+\f12{m}\int_{B_{p}(\e)}w^2d\o\le -\int_{\mathbb S^2}wL_0wd\o.$$
It is easy to verify 
$$\int_{\mathbb S^2}w^2d\o\le C\Big[\int_{\mathbb S^2}|\nab_\gamma w|^2d\o+\int_{B_{p}(\e)}w^2d\o\Big].$$
Hence, 
$$\int_{\mathbb S^2}w^2d\o\le C\int_{\mathbb S^2}|wL_0w|d\o.$$
A simple application of the Cauchy-Schwarz inequality implies 
$$\int_{\mathbb S^2}w^2d\o\le C\int_{\mathbb S^2}|L_0w|^2d\o.$$
Hence, $L_0w=0$ implies $w=0$.

Set 
$$\langle w_1, w_2\rangle =\int_{\mathbb S^2}\Big(\nab_\gamma w_1\cdot\nab_\gamma w_2+\f12fe^{\phi}w_1w_2\Big)d\o.$$
The arguments above also demonstrate that $\langle w_1, w_2\rangle$ is an inner-product on $H^1(\mathbb S^2)$, 
which induces a norm equivalent to the standard $H^1(\mathbb S^2)$-norm. 
{\color{black} The standard $L^2$-theory of elliptic equations applies to the operator $L_0$. 
In particular, the first eigenvalue of $-L_0$ is characterized by
$$\mu_1=\inf\Big\{\int_{\mathbb S^2}\Big(|\nab_\gamma w|^2+\f12fe^{\phi}w^2\Big)d\o; \int_{\mathbb S^2}w^2d\o=1\Big\}.$$
Hence, $\mu_1>0$. Moreover, $\mu_1$ is simple and has a positive eigenfunction. 
Refer to Theorem 8.38 in \cite{GT}. 
We need to point out that Theorem 8.38 in \cite{GT} is formulated for self-adjoint uniformly elliptic operators 
on bounded domains in the Euclidean space with zero Dirichlet data. 
However, the same conclusion holds for self-adjoint uniformly elliptic operators on closed manifolds.} 
\end{proof} 

We now prove the crucial invertibility of the linearized operators.  

\begin{lemma}\label{lemma-linearized-op-invert}
{\color{black} Let $\beta\in (0,1)$ be a fixed constant, and 
$f$ and $\p$ be given nonnegative functions on $\mathbb S^2$, 
with \eqref{eq-assumption-lower-bound-f} satisfied and 
$|\p|_{C^2(\mathbb S^2)}\le K$ for some positive constant $K$.
Then,  
$\partial_{\p}\mathcal S(\p): C^{2, \beta}(\mathbb{S}^2)\rightarrow C^{{0},\beta}(\mathbb{S}^2)$
is invertible, for $a$ sufficiently large depending on $K$.}
\end{lemma}

\begin{proof} Let $L_0$ be the operator as in \eqref{eq-linearization-model}. 
By Lemma \ref{lemma-H1-estimates-L0},  we can take a positive 
eigenfunction $\psi_1$  corresponding to the first eigenvalue $\mu_1>0$; namely, 
$$\D_{\gamma} \psi_1-\f12fe^{\p}\psi_1=-\mu_1\psi_1.$$
By Proposition \ref{form of linearization}, we have 
\begin{align*}
\partial_{\p} \mathcal S(\p)[\psi_1]
&=\Delta_{\gamma}\psi_1-\f12fe^{\p}\psi_1+c^{i}\nab_i \psi_1+c\psi_1\\
&=-\mu_1\psi_1+c^{i}\nab_i \psi_1+c\psi_1.
\end{align*}
Hence, $\partial_{\p} \mathcal S(\p)[\psi_1]<0$ on $\mathbb S^2$ by choosing $a$ appropriately large depending on $K$.
By Theorem 2.11 and its proof in \cite{HL}, the strong maximum principle holds for the operator $\partial_{\p}\mathcal S(\p)$. 
Thus,   for any $w\in C^2(\mathbb S^2)$,
\begin{equation}\label{eq-estimate-linearization}
\max_{\mathbb S^2}|w|\le C\max_{\mathbb S^2}|\partial_{\p}\mathcal S(\p)[w]|,\end{equation}
where $C$ is a positive constant independent of $w$. 
By the Schauder theory, we further obtain for any $w\in C^{2,\beta}(\mathbb S^2)$ 
\begin{equation*}
|w|_{C^{2,\beta}(\mathbb S^2)}\le C|\partial_{\p}\mathcal S(\p)[w]|_{C^{0,\beta}(\mathbb S^2)},\end{equation*}
for some positive constant $C$ independent of $w$. Therefore,   
$\partial_{\p}\mathcal S(\p): C^{2, \beta}(\mathbb{S}^2)\rightarrow C^{{0},\beta}(\mathbb{S}^2)$
is invertible.
\end{proof} 

In \eqref{eq-estimate-linearization}, the constant $C$ may depend on the given function $\p$. 
The dependence of $C$ on $\p$ arises from the dependence of the eigenfunction $\psi_1$ on $\p$. This is because the operator $L_0$ depends on $\p$ via the factor $e^{\p}$ and the metric $\gamma$, 
which is parametrized by $u=1-\ub ae^{-\p}$. 
Moreover, this constance $C$ may also depend on $\ub$ through the function $f=f(\o, \ub)$. In the next result, we will eliminate such dependence on specific $\p$ and $\ub$.

\begin{lemma}\label{lemma-linearized-apriori-estimates}
Let $f$ and $\p$ be given nonnegative functions on $\mathbb S^2$, 
with \eqref{eq-assumption-lower-bound-f} satisfied and 
$|\p|_{C^2(\mathbb S^2)}\le K$ for some positive constant $K$.
Then, for $a$ sufficiently large depending on $K$ and for any $w\in C^2(\mathbb S^2)$, it holds
 \begin{equation}\label{eq-estimate-linearization-v2}
\max_{\mathbb S^2}|w|\le C\max_{\mathbb S^2}|\partial_{\p}\mathcal S(\p)[w]|,\end{equation}
where $C$ is a positive constant depending only on {\color{black} $m$ and $\e$ in \eqref{eq-assumption-lower-bound-f} and $K$, 
independent of specific $\p$ and $\ub\in (0,\delta]$.}
\end{lemma}

\begin{proof} 
We first choose a smooth nonnegative function $f_0=f_0(\o)$ on $\mathbb S^2$ such that $f_0\ge m/2$ on $B_p(\e)$ and $f(\cdot, \ub)\ge f_0$ on $\mathbb S^2$ for any $\ub\in (0,\delta]$.
We take the standard spherical metric $\gamma_0$ on $\mathbb S^2$.  
{\color{black} Instead of  $L_0$ in \eqref{eq-linearization-model}, we consider 
$$L_*w=\D_{\gamma_0} w-\f12f_0w.$$
By applying Lemma \ref{lemma-H1-estimates-L0} to $L_*$,}  we can take a positive 
eigenfunction $\psi_1$  corresponding to the first eigenvalue $\mu_1>0$; namely, 
$$\D_{\gamma_0} \psi_1-\f12f\psi_1=-\mu_1\psi_1.$$
{\color{black} We point out that $\mu_1$ and $\psi_1$ depend only on $m$ and $\e$, independent of $\ub\in (0,\delta]$.} By $\p\ge 0$, {\color{black} $f\ge f_0$}, Lemma \ref{lemma-difference-Laplacians}, and Proposition \ref{form of linearization}, we have 
\begin{align*}
\partial_{\p} \mathcal S(\p)[\psi_1]
&=\Delta_{\gamma_0}\psi_1-\f12fe^{\p}\psi_1+\Delta_{\gamma}\psi_1-\Delta_{\gamma_0}\psi_1+c^{i}\nab_i \psi_1+c\psi_1\\
&\le\Delta_{\gamma_0}\psi_1-\f12f_0\psi_1+\Delta_{\gamma}\psi_1-\Delta_{\gamma_0}\psi_1+c^{i}\nab_i \psi_1+c\psi_1\\
&=-\mu_1\psi_1+\Delta_{\gamma}\psi_1-\Delta_{\gamma_0}\psi_1+c^{i}\nab_i \psi_1+c\psi_1\le -{\color{black} c(m, \e, K)}<0,
\end{align*}
via choosing $a$ appropriately large depending on {\color{black} $m$, $\e$, and $K$}.
By Theorem 2.11 and its proof in \cite{HL}, the strong maximum principle holds, and in particular we obtain \eqref{eq-estimate-linearization-v2}.
\end{proof}

We further prove the following result.   

\begin{lemma}\label{lemma-comparison}
Let $f$ and $\p_1, \p_2$ be given nonnegative functions on $\mathbb S^2$, 
with \eqref{eq-assumption-lower-bound-f} satisfied and 
$|\p_i|_{C^2(\mathbb S^2)}\le K$ for some positive constant $K$
and $i=1,2$.
If $\mathcal S(\p_1)\ge \mathcal S(\p_2)$, 
then either $\p_1<\p_2$ or $\p_1\equiv\p_2$ on $\mathbb S^2$. 
\end{lemma}

\begin{proof} For the given $\p_1$ and $\p_2$, we write 
$$\mathcal L(\p_1-\p_2)\equiv\mathcal S(\p_1)-\mathcal S(\p_2)=\int_0^1\partial_{\p} \mathcal S(t\p_1+(1-t)\p_2)[\p_1-\p_2]dt.$$
Here,  $\mathcal L$ is viewed as a linear operator acting on the difference $\p_1-\p_2$. The operator $\int_0^1 \partial_{\phi}\mathcal S(t\p_1+(1-t)\p_2)dt$ has a similar structure to the linearized operator $\partial_{\phi}\mathcal S(\phi)$ in Proposition \ref{form of linearization}. The corresponding Lemmas \ref{lemma-H1-estimates-L0}-\ref{lemma-linearized-apriori-estimates} can be proved in the same manner. By the assumption, we get $\mathcal L(\p_1-\p_2)\ge 0$. As in the proof of Lemma \ref{lemma-linearized-op-invert}, by Theorem 2.11 in \cite{HL}, the strong maximum principle holds for the operator $\mathcal L$. We hence obtain the desired result.
\end{proof}

\begin{remark}
As a consequence, the equation $\mathcal S(\p)=0$ admits at most one solution $\p$ with the property 
$|\p|_{C^2(\mathbb S^2)}\le K$.  
\end{remark}

\section{The Existence and Properties of Solutions}\label{sec-existence-solutions}  

In this section, for each small $\ub>0$ we construct a MOTS along $\Hb_{\ub}$, and we
further prove that, with $\ub$ as a parameter, these MOTSs form a smooth apparent horizon. 
Mathematically, this is equivalent to first solving the elliptic equation $\mathcal S(\phi, \ub)=0$
with the operator $\mathcal S$ defined in \eqref{eq-expression-cal-S-ub-v2}, and then to proving 
that the resulted solution is also smooth in terms of $\ub$.

We first fix some positive constants. 
Let $\kappa$ be the maximum of the supersolution $\overline{\phi}$ constructed in Lemma  \ref{lemma-super-solutions-phi}
and take an arbitrary constant $K$ such that 
$$K\ge\max\{C(\kappa), |\overline{\p}|_{C^2(\mathbb S^2)}\},$$
where $C(\kappa)$ is the bound established in Theorem \ref{theorem-Schauder-phi}; namely, 
$$C(\kappa)=(Ce^{\kappa})^{Ce^{4\Lambda\kappa}}.$$
Then, we choose $a$ sufficiently large depending on $\kappa$, according to Lemma \ref{lemma-linearized-op-invert}.

We are now ready to prove the existence of MOTSs on each fixed $\Hb_{\ub}$. 

\begin{theorem}\label{theorem-existence-solutions}
Let $\ub$ be a fixed constant in $(0,\delta]$, $f$ be a given smooth function on $\mathbb S^2$, 
satisfying \eqref{eq-assumption-f-0-1} and \eqref{eq-assumption-lower-bound-f}, 
and $\underline{\p}$ and  $\overline{\p}$ 
be as in Lemma \ref{lemma-sub-solution-phi} and Lemma \ref{lemma-super-solutions-phi}, respectively.  
Then, there exists a unique smooth solution ${\p}=\p(\o,\ub)$ of $\mathcal S(\p, \ub)=0$, satisfying 
$\underline{\p}< \p(\cdot, \ub)<\overline{\p}$ on $\mathbb S^2$. 
\end{theorem} 

\begin{proof} We suppress the dependence on $\ub$ in this proof. 
For any constant $\lambda\in [0,1]$ and 
any $C^2$-function $\p$ on $\mathbb S^2$, set  
$$\mathcal S_\lambda(\p)=\Delta_{\gamma}\p+1-\f12\big[\lambda f+(1-\lambda)]e^{\p}+\lambda(I_2+I_3+I_4+I_5),$$
where $I_i$, for $2\le i\le 5$, are defined in \eqref{eq-expression-S-I-2-5}. 
We have $\mathcal S_1(\phi)=\mathcal S(\p)$ by \eqref{eq-expression-S-I-1-5}. 
We will prove the existence of the solution to $\mathcal S_1(\p)=0$ via the method of continuity. 
For $\lambda=0$, the operator $\mathcal S_0(\p)$ is reduced to 
$$\mathcal S_0(\p)=\Delta_{\gamma}\p+1-\f12e^{\p}.$$
It is obvious that $\p_0=\ln 2$ is a solution to $\mathcal S_0(\p)=0$. For suitably small $\e>0$ {\color{black} as in \eqref{eq-assumption-lower-bound-f}}, note that $\underline{\p}< \p_0<\overline{\p}$ on $\mathbb S^2$. 

Before we proceed to prove the openness and the closedness, we point out 
that results established in Sections \ref{sec-Barriers}-\ref{sec-linearized-equations}  for $\mathcal S(\phi)$ also
hold for $\mathcal S_\lambda (\phi)$ with all $\lambda\in [0,1]$. The proofs are exactly the same with $f$ replaced by $\lambda f+(1-\lambda)$ and with $I_i$ replaced by $\lambda I_i$, where $\lambda\in [0,1]$ and $2\leq i \leq 5$.
In particular, the estimate in Theorem \ref{theorem-Schauder-phi} holds for solutions to $\mathcal S_\lambda(\p)=0$, 
uniformly for $\lambda\in [0,1]$. 

For the openness, we fix a constant $\beta\in (0,1)$, say $\beta=1/3$. 
For some $\lambda\in [0,1)$, assume there exists a solution $\p_{\lambda}\in C^{2,\beta}(\mathbb{S}^2)$
to $\mathcal S_\lambda(\p_{\lambda})=0$, satisfying $\underline{\p}< \p_\lambda<\overline{\p}$ on $\mathbb S^2$.
By Theorem \ref{theorem-Schauder-phi}, we have $|\phi_\lambda|_{C^2(\mathbb S^2)}\le K$. By Lemma \ref{lemma-linearized-op-invert}, 
the operator $\partial_{\p}\mathcal S_\lambda(\p_\lambda): C^{2, \beta}(\mathbb{S}^2)\rightarrow C^{{0},\beta}(\mathbb{S}^2)$ is invertible. 
By the implicit function theorem, the equation $\mathcal S_{\lambda'}(\p)=0$ admits a $C^{2,\beta}$-solution 
$\p=\p_{\lambda'}$ for every $\lambda'$ close to $\lambda$. Moreover, 
$\underline{\p}< \p_{\lambda'}<\overline{\p}$ on $\mathbb S^2$, for any $\lambda'$ even closer to $\lambda$. 

The closedness follows from the estimate established in Theorem \ref{theorem-Schauder-phi}.
Lemma \ref{lemma-comparison} guarantees that the strict inequalities $\underline{\p}< \p<\overline{\p}$ hold in the limiting process of $\lambda$, under the assumption $\underline{\p}\le \p\le\overline{\p}$. 
The openness and closeness together prove the existence of the solution to $\mathcal S_\lambda(\p)=0$, satisfying 
$\underline{\p}< \p<\overline{\p}$ on $\mathbb S^2$, for any $\lambda\in [0,1]$, and in particular for $\lambda=1$. 

Last, we prove the uniqueness of the solution to $\mathcal S(\p)=0$ with $\phi$ satisfying 
$\underline{\p}< \p<\overline{\p}$ on $\mathbb S^2$. 
Let $\p$ be one such solution. 
By Theorem \ref{theorem-Schauder-phi}, we have $|\phi|_{C^2(\mathbb S^2)}\le K$. As a consequence of Lemma \ref{lemma-comparison}, such a solution is unique. 
\end{proof} 

Next, we view the solution $\p$ in Theorem \ref{theorem-existence-solutions} as a function of $\o$ and $\ub$. We now study its regularity in terms of  $\ub$.

\begin{theorem}\label{thrm-smoothness-tubes}  
Under the assumptions in Theorem \ref{theorem-existence-solutions}, 
for each $\ub \in (0,\delta]$, let $\p=\p(\o, \ub)$ be the solution 
of $\mathcal S(\p, \ub)=0$ as in Theorem \ref{theorem-existence-solutions}. 
Then, the function $\p=\p(\o, \ub)$ is smooth in terms of the variables $(\o, \ub)\in \mathbb S^2\times(0,\d]$.
\end{theorem}

\begin{proof} We restore $\ub$ back to our notations. 
By Theorem \ref{theorem-existence-solutions}, 
there exists a unique solution ${\p}=\p(\o,\ub)$ to $\mathcal S(\p, \ub)=0$, smooth in $\o$ and satisfying 
$\underline{\p}< \p(\cdot, \ub)<\overline{\p}$ on $\mathbb S^2$.
By Theorem \ref{theorem-Schauder-phi}, we have $|\phi(\cdot, \ub)|_{C^2(\mathbb S^2)}\le K$. We now let $\ub$ as a new parameter and will employ the method of continuity for one more time. Assume $\phi(\o,\ub)$ is a solution to $\mathcal S(\p,\ub)=0$, we will use the implicit function theorem to obtain a solution $\p(\o,\ub')$ to $\mathcal S(\p,\ub')=0$ with $\ub'$ close to $\ub$. The \textit{key} is the invertibility of $\partial_{\p}\mathcal{S}(\p,\ub): C^{2,\beta}(\mathbb{S}^2)\rightarrow C^{0,\beta}(\mathbb{S}^2)$ with any $\phi$ satisfying $|\phi|_{C^2(\mathbb{S}^2)}\leq K$. And this was \textit{already} proved in Lemma \ref{lemma-linearized-op-invert}. Hence,  $\p(\o, \ub)$ is a smooth function of both $\o$ and $\ub$. 
 \end{proof}

In the rest of this section, we show that $\{u=1-\ub a e^{-\p(\o, \ub)}\}$ is a spacelike hypersurface. A crucial step in this part is to derive the equation satisfied by $\partial_{\ub}\p$ {\color{black} and to obtain an estimate of $\partial_{\ub}\p$}. 

\begin{proposition}\label{prop-equation-partial-ub-phi}  
For each $\ub \in (0,\delta]$, let $\p=\p(\o, \ub)$ be the solution 
to $\mathcal S(\p, \ub)=0$ as in Theorem \ref{theorem-existence-solutions}. 
Then, $\partial_{\ub}\p$ satisfies 
\begin{equation}\label{eq-equation-partial-ub-phi}
\Delta_{\gamma}(\partial_{\ub}\p)-\f12f(\o, \ub)e^{\p}\partial_{\ub}\p+c^{i}\nab_i (\partial_{\ub}\p)+c\partial_{\ub}\p=\frac{1}{\ub}h
+\f12\partial_{\ub}fe^{\p},
\end{equation}
where the connection is with respect to the metric $\gamma$, and 
$c=c(\o,\p,\ub)$, $c^i=c^i(\o,\p,\ub)$ and $h=h(\o,\p,\ub)$
are functions satisfying 
\begin{equation}\label{eq-estimates-partial-ub-coefficients}
|c|+|c^i|+|h|\leq Ca^{-\f13} e^{\p}{\color{black} (|\nabla^2_{\gamma}\p|+|\nabla_\gamma\p|^2+1)}.\end{equation}
\end{proposition}

\begin{proof} The proof is similar to that of Proposition \ref{form of linearization}. 
By $\mathcal S(\p, \ub)=0$, we have 
\begin{equation}\label{eq-equation-partial-ub-phi-v1}
\f{d\,}{d\ub}\mathcal S(\p, \ub)=0.\end{equation}
We proceed to compute the left-hand side and express it in terms of $\partial_{\ub}\p$. 
Throughout the proof, we denote $\dot{\,}=\f{d\,}{d\ub}$. 

Recall the relations given by \eqref{eq-relation-g-gamma} and \eqref{eq-relation-u-R-phi}, i.e., 
$$g=R^{2}\gamma, \quad u=1-R, \quad R=\ub a e^{-\p}. $$ 
Then, 
\begin{equation}\label{dot for ub}
\dot{u}=-\dot{R}, \, \dot{R}=R(\ub^{-1}-\dot{\p}).
\end{equation}
Note that $\gamma, \chib, \chi, \eta, \O, \omb$ are functions of $\o, u, \ub$ with $u=1-R$. 
We first derive a general formula. Consider a smooth function $\psi=\psi(\o, u, \ub)$ defined along the apparent horizon. Then, 
$$\dot{\psi}=\dot{u}\partial_u\psi+\partial_{\ub}{\psi}=R(\dot{\p}-\ub^{-1})\partial_u\psi+\partial_{\ub}{\psi},$$
and hence
\begin{align}\label{eq-dot-ub-general}
\dot{\psi}=R\partial_u\psi\dot{\phi}-\ub^{-1}(R \partial_u\psi-\ub \partial_{\ub}{\psi}).
\end{align}
By $g=R^2\gamma$ and $\dot{u}=-\dot{R}$, we get 
\begin{equation*}
R^2\dot{\gamma}_{kl}=-2R\dot{R}\gamma_{kl}-\dot{R}\partial_ug_{kl}+\partial_{\ub}g_{kl}.  
\end{equation*}
Recall \eqref{eq-variation-g-u} and \eqref{eq-variation-g-ub}, i.e., 
\begin{align*}
\partial_ug_{kl}=\O\tr_g\chib g_{kl}+2\O \chibh_{kl},\quad 
\partial_{\ub}g_{kl}=\O\tr_g\chi g_{kl}+2\O \chih_{kl}{\color{black}-\Omega d^A \f{\partial \chi_{kl}}{\partial \theta^A}}.
\end{align*}
Hence,  
\begin{align}\label{eq-dot-ub-metric}\begin{split}
\dot{\gamma}_{kl}&=\dot{\p}\big[(R\O\tr_g\chib+2)\gamma_{kl}+ 2R^{-1}\O \chibh_{kl}\big]\\
&\qquad-\ub^{-1}\big[(R\O\tr_g\chib+2)\gamma_{kl}+ 2R^{-1}\O \chibh_{kl}\\
&\qquad \quad \quad \quad \, \, \,  {\color{black}-\ub\O\tr_g\chi \gamma_{kl}
-\ub R^{-2}(2\O \chih_{kl}-\Omega d^A \f{\partial \chi_{kl}}{\partial \theta^A})}\big].
\end{split}\end{align}
We also note 
\begin{align}\label{eq-dot-ub-metric-1}
\dot{\gamma}^{ij}=-\gamma^{ik}\gamma^{jl}\dot{\gamma}_{kl}. 
\end{align}

In \eqref{eq-expression-S-I-1-5}, \eqref{eq-expression-S-I-1}, and \eqref{eq-expression-S-I-2-5}, 
we decompose $\mathcal S(\p, \ub)$ as
\begin{align*}
\mathcal S(\p, \ub)=\sum_{l=1}^5I_{l}.\end{align*}
By \eqref{eq-equation-partial-ub-phi-v1}, we have 
\begin{equation}\label{eq-zero-sum-dot}\sum_{l=1}^5 \dot{I}_{l}=0.\end{equation}

In the following, we compute each of $\dot{I}_{l}$, for $l=1, \cdots, 5$. We start with $\dot{I}_1$. 
Recall that $I_1$ is given by \eqref{eq-expression-S-I-1}, i.e., 
\begin{align*}
I_1=\D_{\gamma}\phi+1-\f12fe^{\p}.\end{align*}
Then, 
\begin{equation}\label{I1 dot}
\dot{I}_1=(\Delta_{\gamma}\p)^{\boldsymbol\cdot}-\f12fe^{\p}\dot{\p}-\f12\dot{f}e^{\p}.
\end{equation}
By Lemma \ref{lemma-variation-Laplacians}, we have 
\begin{equation}\label{eq-variation-Laplace-ub-phi}
(\Delta_{\gamma}\p)^{\boldsymbol\cdot}=\Delta_\gamma \dot{\p}-\gamma^{ik}\gamma^{jl}\dot{\gamma}_{kl}\nabla_{ij}\p
-\gamma^{ij}\gamma^{kl}\nabla_i\dot{\gamma}_{jl}\nabla_k\p+\frac12\gamma^{ij}\gamma^{kl}\nabla_l\dot{\gamma}_{ij}\nabla_k\p.
\end{equation}
{\color{black} By substituting \eqref{eq-variation-Laplace-ub-phi} in \eqref{I1 dot},} we have 
\begin{equation}\label{eq-variation-I1-ub} 
\dot{I}_1=\Delta_\gamma \dot{\p}-\f12fe^{\p}\dot{\p}-\f12\partial_{\ub}fe^{\p}
+{c}_1^{i}\nab_i \dot{\p}+{c}_1\dot{\p}-\frac{1}{\ub}{h}_1,
\end{equation}
where 
\begin{align*}
{c}^{i}_1&=-2R^{-1}\O\chibh_{jl}\gamma^{ij}\gamma^{kl}\nab_k \p,\\
{c}_1&=-(2+R\O\tr_g\chib)\Delta_{\gamma}\p
-2\gamma^{ij}\gamma^{kl}\big[R^{-1}\nab_i(\O\chibh_{jl}) \nab_k\p
+R^{-1}\O \chibh_{jl}(\nab_{ik}\p
+ \nab_i\p \nab_k \p)\big],
\end{align*}
and 
\begin{align*}
{h}_1&=-(2+R\O\tr_g\chib)\Delta_{\gamma}\p+\ub\O\tr_g\chi\Delta_{\gamma}\p\\
&\qquad-2\gamma^{ij}\gamma^{kl}\big[R^{-1}\nab_i(\O\chibh_{jl}) \nab_k\p
+R^{-1}\O \chibh_{jl}(\nab_{ik}\p
+ \nab_i\p \nab_k \p)\big]\\
&\qquad+\ub\gamma^{ij}\gamma^{kl}\big[R^{-2}\nab_i(2\O\chih_{jl}{\color{black}-\Omega d^A \f{\partial \chi_{jl}}{\partial \theta^A}}) \nab_k\p\\
&\qquad \qquad \qquad \quad
{\color{black}+R^{-2}(2\O \chih_{jl}-\Omega d^A \f{\partial \chi_{jl}}{\partial \theta^A})(\nab_{ik}\p
+ 2\nab_i\p \nab_k \p)\big]}.
\end{align*}
{\color{black} We point out that $c_1^i, c_1$, and $h_1$ can be viewed as linear combinations of  $\nab_{ij}\p$, $\nab_i\p\nab_j\p$, and $\nab_i\p$, 
with coefficients given by $\gamma^{ij}$ and 
\begin{equation}\label{eq-expressions-coefficients-z}\begin{split}
&2+R\O\tr_g\chib,\, R^{-1}\Omega\chibh_{kl},\,  R^{-1}\nab_i(\Omega\chibh_{kl}), \\
&\ub\O\tr_g\chi,\, \ub R^{-2}\O \chih_{kl}, \, \ub R^{-2}\nab_i(\O \chih_{kl}),\\
&{\color{black}\ub R^{-2}\Omega d^A \f{\partial \chi_{kl}}{\partial \theta^A}, \, \ub R^{-2}\nab_i(\Omega d^A \f{\partial \chi_{kl}}{\partial \theta^A})}.
\end{split}\end{equation}
By \eqref{A1}, all expressions in \eqref{eq-expressions-coefficients-z} are controlled by $e^{\p}a^{-\f13}$.

Next, we compute $\dot{I}_2, \cdots, \dot{I}_5$. Recall that $I_2, \cdots, I_5$ are given by \eqref{eq-expression-S-I-2-5}, i.e., 
\begin{align*}
I_2=2R^{-1}\Omega\chibh_{kl} \gamma^{ik}\gamma^{jl}\nab_i \p\nab_j \p, \quad
&I_3=2\gamma^{ij}\eta_i \nab_j \p,\\
I_4=R\alpha_1 \gamma^{ij}\nab_i\p\nab_j\p,\quad
&I_5=R\alpha_2. 
\end{align*}
Straightforward computations yield 
\begin{align*}
\dot{I}_2&=-2R^{-2}\dot{R}\Omega\chibh_{kl} \gamma^{ik}\gamma^{jl}\nab_i \p\nab_j \p
+2R^{-1}(\Omega\chibh_{kl})^{\boldsymbol\cdot} \gamma^{ik}\gamma^{jl}\nab_i \p\nab_j \p \\
&\qquad +4R^{-1}\Omega\chibh_{kl} \dot{\gamma}^{ik}\gamma^{jl}\nab_i \p\nab_j \p
+4R^{-1}\Omega\chibh_{kl} \gamma^{ik}\gamma^{jl}\nab_i \dot{\p}\nab_j \p,\\
\dot{I}_3&=2\dot{\gamma}^{ij}\eta_i \nab_j \p+2\gamma^{ij}\dot{\eta}_i \nab_j \p+2\gamma^{ij}\eta_i \nab_j \dot{\p},
\end{align*} 
and 
\begin{align*}
\dot{I}_4&=\dot{R}\alpha_1 \gamma^{ij}\nab_i\p\nab_j\p+R\dot{\alpha}_1 \gamma^{ij}\nab_i\p\nab_j\p\\
&\qquad+R\alpha_1 \dot{\gamma}^{ij}\nab_i\p\nab_j\p+2R\alpha_1 \gamma^{ij}\nab_i\dot{\p}\nab_j\p,\\
\dot{I}_5&=\dot{R}\alpha_2+R{\dot\alpha}_2. 
\end{align*}
In the above expressions, we substitute $\dot{R}$ by \eqref{dot for ub}, $\dot{\gamma}^{ij}$ by \eqref{eq-dot-ub-metric} and \eqref{eq-dot-ub-metric-1}, and $\dot\psi$ by \eqref{eq-dot-ub-general} for $\psi=\Omega\chibh_{kl}$, $\eta_i$, $\alpha_1$, and $\alpha_2$. 
For each $l=2, 3, 4, 5$, we can write 
\begin{equation}\label{eq-variation-I-i-ub} 
\dot{I}_{l}={c}^{i}_{l}\nab_i \dot{\p}+{c}_{l}\dot{\p}-\frac{1}{\ub}{h}_{l},
\end{equation}
where $c_{l}^i, c_{l}$, and $h_{l}$ can be viewed as linear combinations of  $\nab_{ij}\p$, $\nab_i\p\nab_j\p$, and $\nab_i\p$,
with coefficients given by $\gamma^{ij}$ and 
\begin{equation}\label{eq-expressions-coefficients-y}\begin{split}
&R\O\tr_g\chib+2,\, R^{-1}\O \chibh_{kl},\,  \eta_i,\, R\alpha_1, \, R\alpha_2,\,\\
&\ub\O\tr_g\chi, \, \ub R^{-2}\O \chih_{kl}, \, {\color{black}\ub R^{-2}\Omega d^A \f{\partial \chi_{kl}}{\partial \theta^A}},\\
&\partial_u(\O \chibh_{kl}),\,  R\partial_u\eta_i,\, R^2\partial_u\alpha_1, \, R^2\partial_u\alpha_2, \\ 
&\ub R^{-1}\partial_{\ub}(\O \chibh_{kl}),\,  \ub \partial_{\ub}\eta_i,\, \ub R\partial_{\ub}\alpha_1, \, \ub R\partial_{\ub}\alpha_2. 
\end{split}\end{equation}
By \eqref{A1},  \eqref{A2},  and \eqref{A3}, 
all expressions in \eqref{eq-expressions-coefficients-y} are controlled by $e^{\p}a^{-\f13}$. 
}

By combining \eqref{eq-zero-sum-dot},  \eqref{eq-variation-I1-ub}, and \eqref{eq-variation-I-i-ub},  
we then have the desired result.
\end{proof} 

As a consequence, we have the following estimate.   

\begin{corollary}\label{cor-equation-partial-ub-phi}
Under the assumptions in Theorem \ref{theorem-existence-solutions}, 
for each $\ub \in (0,\delta]$, let $\p=\p(\o, \ub)$ be the solution 
of $\mathcal S(\p, \ub)=0$ as in Theorem \ref{theorem-existence-solutions}. 
Then, 
\begin{equation}\label{eq-estimate-partial-ub-phi}
\max_{\mathbb S^2}|\partial_{\ub}\p|\le C\big[\ub^{-1}a^{-\f13} \max_{\mathbb S^2}e^{\p}
(|\p|^2_{C^2(\mathbb S^2)}+1)+\max_{\mathbb S^2}|\partial_{\ub}fe^{\p}|\big],
\end{equation}
where $C$ is a positive constant independent of $a$, $\ub$, and $\p$. 
\end{corollary}

\begin{proof} The left-hand side of \eqref{eq-equation-partial-ub-phi} can be viewed as a linear operator acting on 
$\partial_{\ub}\phi$. This linear operator has the same structure as $\partial_{\p}\mathcal S(\p, \ub)$
in \eqref{eq-linearization-S}. Applying Lemma \ref{lemma-linearized-apriori-estimates} with $\partial_{\phi}\mathcal{S}(\p,\ub)$ replaced by the linear operator on the left hand side of \eqref{eq-equation-partial-ub-phi}, we then obtain the desired result. 
We emphasize that $C$ can be chosen independent of $a$, $\ub$, and $\p$. 
\end{proof}

Now, we are ready to prove another main result in this paper.  

\begin{theorem}\label{thrm-spacelike} 
Under the assumptions in Theorem \ref{theorem-existence-solutions},
for each $\ub \in (0,\delta]$, let $\p=\p(\o, \ub)$ be the solution 
of $\mathcal S(\p, \ub)=0$ as in Theorem \ref{theorem-existence-solutions}. 
Assume 
\begin{equation}\label{eq-assumption-ub-f}
|\ub\partial_{\ub}f(\o,\ub)|\leq a^{-\f13}.\end{equation} 
Then, the metric $g$ restricted to $\{(\o, u, \ub);\,u=1-\ub a e^{-\p(\o, \ub)}\}$ is Riemannian. \end{theorem}

\begin{proof} With $R=\ub a e^{-\p(\o, \ub)}$, we have 
$$\partial_{\ub}R=a e^{-\p(\o, \ub)}\big(1-\ub \partial_{\ub}\p\big).$$
By \eqref{eq-estimate-partial-ub-phi} and \eqref{eq-assumption-ub-f},  taking $a$ sufficiently large, we have 
$$\max_{\mathbb S^2}|\ub\partial_{\ub}\p|\le 1/2.$$
Hence, $\partial_{\ub}R\ge a e^{-\p(\o, \ub)}/2$. 

Let $g'$ be the induced metric to the hypersurface
$u=1-R(\o, \ub)$. Denote $\o=(\theta_1, \theta_2)$. Then, 
\begin{align*}
g'_{\theta_i \theta_j}&=g_{\theta_i \theta_j}+\partial_{\theta_i} (1-R)\partial_{\theta_j}(1-R) \cdot g(\partial u, \partial u)=g_{\theta_i \theta_j},\\
g'_{\ub\, \ub}&=g_{\ub\, \ub}+2\partial_{\ub} (1-R)\partial_{\ub}\ub \cdot g(\partial_u, \partial_{\ub})=4\partial_{\ub} R,\\
g'_{\theta_i \ub}&=g_{\theta_i \ub}+\partial_{\theta_i} (1-R)\partial_{\ub}\ub\cdot g(\partial_{u}, \partial_{\ub})=2\partial_{\theta_i}R,
\end{align*}
where we used $g(\partial_u, \partial_u)=0$, $g(\partial_u, \partial_{\ub})=-2$, and $g(\partial_u, \partial_{\theta_i})=0$. 

Let $X=\lambda_1\partial_{\theta_1}+\lambda_2\partial_{\theta_2}+\lambda_3\partial_{\ub}$ be an arbitrary 
nonzero tangent vector, for 
some constants $\lambda_1, \lambda_2, \lambda_3$. Then, we have
\begin{align*}
g'(X,X)
=\lambda_i\lambda_jg_{ij}+4\lambda_i \lambda_3 \partial_{i}R+4\lambda_3^2 \partial_{\ub} R,\end{align*}
{\color{black} where $i,j$ are summed over $1,2$.}
If $\lambda_3=0$, then $g'(X,X)
=\lambda_i\lambda_jg_{ij}>0$. If $\lambda_3\neq0$, then by using  $g_{ij}=R^2\gamma_{ij}$ in \eqref{eq-relation-g-gamma}, we obtain
\begin{align*}
g'(X,X)
&\ge\ub^2a^2e^{-2\p}\lambda_i\lambda_j\gamma_{ij}-4\ub ae^{-\p}\lambda_i \lambda_3 \partial_{i}\p+2\lambda_3^2ae^{-\p}\\
&\ge2\lambda_3^2ae^{-\p}-4\lambda_3^2|\nab_\gamma \p|^2>0,
\end{align*}
if $a$ is sufficiently large.\end{proof}  

Therefore, the corresponding apparent horizon is spacelike, 
and  hence a dynamical horizon defined in \cite{AG, AK03, AK}.

\section{Anisotropic Apparent Horizons}\label{anisotropic examples}  

In previous sections, we constructed a unique MOTS in a certain range along each $\Hb_{\ub}$, and gathering all the MOTSs  
we further proved that it is a smooth apparent horizon (except at the center). 
In this section, we will construct anisotropic apparent horizons. 
Specifically, we will construct MOTSs with isolated {\it valleys}. Via Theorem \ref{thrm-smoothness-tubes}, these MOTSs also form a smooth apparent horizon with valleys.  We now focus on the operator $\mathcal S(\p, \ub)$ defined by \eqref{eq-expression-cal-S-ub-v2}  
and construct single-valley and multi-valley solutions, 
which translates to solutions of $\mathcal G(R, \ub)=0$ with isolated {\it peaks}.

We allow $f$ depends on $\ub$, but has to satisfy assumptions imposed earlier on its $\ub$-derivative in Theorem \ref{main thm}. In this section, when there is no danger of confusion, we suppress the department of $f$ on $\ub$. For fixed $\ub$, we treat $f(\ub, \o)$ as a function of $\o$.

We start from constructing solution with one valley.
 
\begin{theorem}\label{prop-anisotropic-phi-one}
Let $p$ be a point on $\mathbb S^2$, $\e$ be a sufficiently small positive constant, 
and $f$ be a smooth function on $\mathbb S^2$, such that $0\le f\le 1$ and
$$f\ge 1/2\,\text{ on }B_p(\e/2),\quad \quad f=0\,\text{ on }\mathbb S^2\setminus B_p(\e).$$
Then, for each ${\ub}$,  the equation $\mathcal S(\phi, \ub)=0$ admits a solution  $\phi_{\e}(\o, \ub)$  satisfying, 
for any $r_0\gg \e$, 
\begin{equation}\label{eq-estimate-anisotropic-max-min-one}
\phi_{\e}(p, \ub)-\inf_{\o\in\mathbb S^2\setminus B_p(r_0)}\phi_{\e}(\o, \ub)  
\le \f12\log\e+C,\end{equation}
where $C$ is a positive constant depending on $r_0$, independent of $\e$ and $\ub$. 
\end{theorem}

\begin{proof} We fix a $\ub$ and consider the operator $\mathcal S$  as in \eqref{3.0-v1}. 
Note that $f$ satisfies \eqref{eq-assumption-lower-bound-f}, with $m=1/2$. 
Let $\underline{\p}$ and  $\overline{\p}$ 
be the subsolution and the supersolution as in Lemma \ref{lemma-sub-solution-phi} and Lemma \ref{lemma-super-solutions-phi}, 
with $m=1/2$ and $\e$ there replaced by $\e/2$. 
By \eqref{eq-estimate-super-phi}, we have 
\begin{align}\label{eq-estimate-super-phi-special-one}\begin{split}
\overline{\phi}_{\min}&=\overline{\phi}(p)=-2\log\e+O(1),\\
\overline{\phi}_{\max}&=\overline{\phi}(-p)=-(\tau+2)\log\e+O(1).
\end{split}\end{align}
By Theorem  \ref{theorem-existence-solutions}, 
there exists a unique solution  $\phi$ of $\mathcal S(\phi)=0$, satisfying 
\begin{equation}\label{eq-estimate-solution-upper-lower} 
\underline{\phi}< \phi< \overline \phi.\end{equation}
Recall that $\underline{\p}$ constructed in Lemma \ref{lemma-sub-solution-phi} is constant. 
With constant $\underline{\phi}$, the estimate \eqref{eq-estimate-solution-upper-lower} is not sufficient to prove the desired result. 
In the following, we construct a new anisotropic subsolution $\underline{\phi}_*$
to improve the lower bound of $\p$ such that 
$\underline{\phi}_*$, like $\overline{\phi}$, attains its minimum and maximum at two different points $p$ and $-p$, respectively. Moreover, we will show $\overline{\phi}(p)$ is significantly less than $\underline{\phi}_*$ valued away from $p$.

For the given  point $p\in\mathbb S^2$,  denote by $d$ the distance on $\mathbb S^2$ from $p$.  
Then, $d(p)=0$ and $d(-p)=\pi$. For some positive constant $\mu \in (0.1, 0.9)$ 
and a small parameter $\delta$ to be determined, construct
\begin{equation}\label{eq-expression-phi-one}
\underline{\phi}_*=\mu\log\Big[\f{1+\delta-\cos d}{\delta} \Big].\end{equation}
Then, 
\begin{equation*}
\underline{\phi}_{*\min}=\underline{\phi}_*(p)=0, 
\quad \underline{\phi}_{*\max}=\underline{\phi}_*(-p)=-\mu\log\Big[\f{\delta}{2+\delta} \Big]
=-\mu\log\delta+O(1).
\end{equation*}
In fact, for any fixed $r_0\in (0,\pi)$ and any $q\in \mathbb S^2\setminus B_p(r_0)$, we have 
\begin{equation}\label{eq-subsol-phi-bound-outside-one}
\underline{\phi}_*(q)
=-\mu\log\delta+O(1). 
\end{equation}
We claim
\begin{equation}\label{eq-identity-Laplacian-one}
\D_{\gamma_0}\underline{\phi}_*+\mu=\f{\mu\delta(2+\delta)}{(1+\delta-\cos d)^2}, 
\end{equation}
where $\D_{\gamma_0}$ is the Laplacian with respect to the spherical metric ${\gamma_0}$ on $\mathbb S^2$.  
We postpone the proof of \eqref{eq-identity-Laplacian-one} after \eqref{eq-subsol-phi-bound-outside-v2-one} and it is a direct calculation. Then, 
\begin{equation*}
(\D_{\gamma_0}\underline{\phi}_*+\mu)e^{-\underline{\phi}_*}=\frac12 \underline{f}, 
\end{equation*}
where 
\begin{equation}\label{eq-expression-f-one}
\underline{f}=\f{2\mu(2+\delta)\delta^{1+\mu}}{(1+\delta-\cos d)^{2+\mu}}.\end{equation}
We now claim $\underline{f}> f$ on $\mathbb S^2$. 
Thus back to \eqref{3.0-v2}, with $a$ sufficiently large, 
we have $\mathcal S(\underline{\phi})>0$, since $\mu<0.9$. 
To prove the claim, we only need to prove $\underline{f}>1$ on $B_p(\e)$, since $f\le 1$ on $B_p(\e)$ 
and $f=0$ outside $B_p(\e)$. By ${\color{black} \cos d\ge 1-d^2/2}$  and \eqref{eq-expression-f-one}, we have 
\begin{equation*}
\underline{f}\geq
2\mu(2+\delta)\cdot\f{\delta^{1+\mu}}{(\delta+\e^2/2)^{2+\mu}}\quad\text{on }B_p(\e).\end{equation*}
For some $\beta>0$ to be determined, set $\delta=\e^\beta$. Then, 
\begin{equation*}
\underline{f}\geq 2\mu(2+\e^\beta)\cdot\f{\e^{\beta(1+\mu)}}{(\e^\beta+\e^2/2)^{2+\mu}}\quad\text{on }B_p(\e).\end{equation*}
If $\beta>2$, then 
\begin{equation*}
\underline{f}\geq \mu2^{4+\mu}\cdot\f{1+\e^\beta/2}{(1+2\e^{\beta-2})^{2+\mu}} 
\cdot\e^{\beta(1+\mu)-2(2+\mu)}\quad\text{on }B_p(\e).\end{equation*}
Thus, $\underline{f} \ge \mu2^3$  on $B_p(\e)$ if $\beta(1+\mu)=2(2+\mu)$ and $\e$ is small. In the following,  we take
\begin{equation}\label{eq-expression-beta-one}\beta=\f{2(2+\mu)}{1+\mu}.\end{equation}
Then, \eqref{eq-subsol-phi-bound-outside-one} reduces to, for any $q\in \mathbb S^2\setminus B_p(r_0)$,  
\begin{equation}\label{eq-subsol-phi-bound-outside-v2-one} 
\underline{\phi}_*(q)
=-\mu\beta\log\e+O(1).
\end{equation}

To prove \eqref{eq-identity-Laplacian-one}, we position $p$ at the north pole
and denote by $(\theta, \psi)$ the spherical coordinates on $\mathbb S^2$, with $\theta\in [0,\pi]$ and $\psi\in [0,2\pi)$. 
Then, $d=\theta$. 
For a function $\phi=\phi(\theta)$, we have 
\begin{equation*}
\begin{split}
\D_{0}\phi(\theta)
=\f{1}{\sin\theta}\partial_{\theta}(\sin\theta \partial_\theta\phi(\theta)).
\end{split}
\end{equation*}
Then, a straightforward computation yields \eqref{eq-identity-Laplacian-one}. 

We now compare $\underline{\phi}_*$ and $\overline{\phi}$. 
By \eqref{eq-estimate-super-phi-special-one} and \eqref{eq-subsol-phi-bound-outside-v2-one}, we have 
for any $q\in \mathbb S^2\setminus B_p(r_0)$ 
\begin{equation*}
\overline{\phi}(p)-\underline{\phi}_*(q)
\le(\mu\beta-2)\log\e+C,
\end{equation*} 
where we replaced $O(1)$ by $C$, a positive constant, independent of $\epsilon$ and $\ub$.
With $\beta$ as in \eqref{eq-expression-beta-one}, we can take $\mu$ such that $\mu\beta>2$, say $\mu\beta=2.5$. 
This yields $2\mu (2+\mu)/(1+\mu)=2.5$, and we have $\mu=(\sqrt{89}-3)/8\approx 0.80$.

By $\mathcal S(\p)=0$, $\mathcal S(\underline{\p}_*)>0$, and Lemma \ref{lemma-comparison}, 
we have $\underline{\p}_*<\p$. Here, in order to apply Lemma \ref{lemma-comparison}, 
we may need to enlarge $K$ so that $|\underline{\p}_*|_{C^2(\mathbb S^2)}\le K$, if necessary. 
Then, for any $q\in \mathbb S^2\setminus B_p(r_0)$, we obtain
$$\phi(p)-\phi(q)\le \overline{\phi}(p)-\underline{\phi}_*(q)
\le \f12\log\e+C.$$
We now restore $\e$ and $\ub$, and write $\p_{\e}(\cdot, \ub)$ for $\phi$. Then $\phi_{\epsilon}(\o,\ub)$ is our desired solution.
\end{proof} 

In Theorem \ref{prop-anisotropic-phi-one}, the solution $\p$
has a value at $p$ significantly less than the values away from $p$. In this sense, we say that 
$\p$ has a valley at $p$.

Theorem  \ref{prop-anisotropic-phi-one} can be reformulated easily with the operator $\mathcal G(R, \ub)$ defined by  
\eqref{eq-expression-cal-G-ub}. 

\begin{proposition}\label{prop-anisotropic-R-one}
Let $p$ be a point on $\mathbb S^2$, $\e$ be a sufficiently small positive constant, 
and $f$ be a smooth function on $\mathbb S^2$, such that $0\le f\le 1$ and
$$f\ge 1/2\,\text{ on }B_p(\e/2),\quad\quad f=0\,\text{ on }\mathbb S^2\setminus B_p(\e).$$
Then, for each ${\ub}$, the equation $\mathcal G(R, \ub)=0$ admits a unique solution  $R_{\e}(\o, \ub)$  satisfying 
\begin{equation}\label{eq-estimate-anisotropic-R-max-min}
\sup_{\o\in \mathbb S^2\setminus B_p(r_0)}R_{\e}(\o, \ub)\le C\e^{{\color{black}\f12}} R_{\e}(p, \ub),\end{equation} 
where $C$ is a positive constant, independent of $\e$ and $\ub$. 
\end{proposition}

As a consequence, the corresponding MOTS and apparent horizon are \textit{anisotropic}.

In the rest of the section, we construct solutions with multiple valleys. 
We first modify the arguments in the proof of Lemma \ref{lemma-super-solutions-phi}
to construct a supersolution with multiple valleys. 

\begin{lemma}\label{lemma-super-solutions-phi-multi}
Let $p_1, \cdots, p_n$ be fixed $n$ points on $\mathbb S^2$, 
$\e$ be a positive constant 
much less than the distance between any two distinct $p_i$ and $p_j$ for $i, j=1, \cdots, n$, 
and $f$ be a nonnegative continuous function on $\mathbb S^2$ such that 
$f\ge m$ on $\cup_{i=1}^nB_{p_i}(\e)$, for some constant  $m>0$. 
Then, for any $\tau>2.2/n$, there exists  $\overline{\phi}\in C^\infty(\mathbb S^2)$ such that $\mathcal S(\overline{\p})<0$ and, 
for any $i=1, \cdots, n$,  
\begin{align}\label{eq-estimate-super-phi-multi}\begin{split}
\overline{\phi}(p_i)&=-\log(m\e^2)+C,\\
\overline{\phi}_{\max}&=-\log(m\e^{\tau+2})+C,
\end{split}\end{align}
for some positive constant $C$ depending only on $n$ and $\tau$. 
\end{lemma}

\begin{proof} For each $i=1, \cdots, n$,  denote by $d_i$ the distance on $\mathbb S^2$ from $p_i$
and take a cutoff function
$$\chi_{i,\e} = 
\begin{cases}
        0  & \text{on $B_{p_i}(\epsilon/2)$},\\
        1 & \text{on $\mathbb{S}^2\backslash B_{p_i}(\epsilon)$}. \end{cases}$$
Set 
$$\tilde{\phi}_{i,\e}=\tau\Big[\chi_{i,\e}\log\sin\f{d_i}{2}+(1-\chi_{i,\e})\log\e\Big].$$
Then, we have
$$\tilde{\phi}_{i,\e}=\begin{cases}\tau\log \e  & \text{on $B_{p_i}(\epsilon/2)$},\\
        \tau\log\e+O(1) & \mbox{on $B_{p_i}(\e)\backslash B_{p_i}(\e/2)$},\\
        \tau\log d_i+O(1)  & \mbox{on $\mathbb{S}^2\backslash B_{p_i}(\epsilon)$},\end{cases}$$
and 
$$ |\nab_0\tilde{\phi}_{i,\epsilon}|=O(\e^{-1}), \quad
         |\nab_0^2\tilde{\phi}_{i,\epsilon}|=O(\e^{-2}) \quad\text{on $\mathbb{S}^2\backslash B_{p_i}(\epsilon/2)$}.$$
Moreover, it holds
$$\D_{0}\tilde{\phi}_{i,\e}=-\f{\tau}{2}
\quad\text{on }\mathbb{S}^2\backslash B_{p_i}(\e).$$     
Set 
$$\tilde{\phi}_{\e}=\sum_{i=1}^n\tilde{\phi}_{i,\e}.$$
Then, it follows
$$\D_{0}\tilde{\phi}_{\e}+1.1=-\f{n\tau}{2}+1.1<0
\quad\text{on }\mathbb{S}^2\backslash \cup_{i=1}^nB_{p_i}(\e),$$ 
as long as $\tau>2.2/n$. Next, for each $i=1, \cdots, n$, we have 
$$e^{\t\phi_{i,\e}}\ge \e^\tau/C_1, \quad e^{\t\phi_{j,\e}}\ge 1/C_1 \,(j\neq i) \quad\text{on }B_{p_i}(\e),$$ 
for some positive constant $C_1$. 
Hence, 
$$e^{\t\phi_{\e}}\ge \e^\tau/C_1^n \quad\text{on }\cup_{i=1}^nB_{p_i}(\e).$$
We can now follow the same argument as in the proof of Lemma \ref{lemma-super-solutions-phi}. Set 
$$\phi_{\e}=\tilde{\phi}_{\e}-\log s_{\e},$$
where $s_\e$ is a positive constant given by 
$$s_{\e}=\f{m\e^{\tau+2}}{2C_0 C_1^n}.$$
We then have the desired result. 
\end{proof}

\begin{theorem}\label{prop-anisotropic-supersolution-phi-multi}
Let $p_1, \cdots, p_n$ be fixed $n$ points on $\mathbb S^2$, 
$\e$ be a positive constant 
much less than the distance between any two distinct $p_i$ and $p_j$ for $i, j=1, \cdots, n$, 
and $f$ be a smooth function on $\mathbb S^2$, such that $0\le f\le \nu(n)$ and
$$f\ge \nu(n)/2\,\text{ on }\cup_{i=1}^nB_{p_i}(\e/2),
\quad\quad f=0\,\text{ on }\mathbb S^2\setminus \cup_{i=1}^nB_{p_i}(\e),$$
for some constructed constant $\nu(n)\in (0,1)$ depending only on $n$. 
Then for each ${\ub}$, 
the equation $\mathcal S(\phi, \ub)=0$ admits a solution  $\phi_{\e}(\o, \ub)$  satisfying, 
for any $r_0\gg \e$ and any $i=1, \cdots, n$,  
\begin{equation}\label{eq-estimate-anisotropic-max-min}
\phi_{\e}(p_i, \ub)-\inf_{\o\in\mathbb S^2\setminus \cup_{i=1}^nB_{p_i}(r_0)}\phi_{\e}(\o, \ub)   
\le \f1{2n}\log\e+C,\end{equation}
where $C$ is a positive constant, independent of $\e$ and $\ub$. 
\end{theorem}

{\color{black} This is simply Theorem \ref{prop-anisotropic-supersolution-phi-multi-intro}.} 

\begin{proof} We follow the proof of Theorem \ref{prop-anisotropic-phi-one}. 
We fix a $\ub$ and consider the operator $\mathcal S$  as in \eqref{3.0-v1}. We take the subsolution $\underline{\p}$ as in Lemma \ref{lemma-sub-solution-phi}
and the supersolution $\overline{\p}$ as in Lemma \ref{lemma-super-solutions-phi-multi}, 
with $m=\nu(n)/2$ and $\e$ there replaced by $\e/2$. 
By \eqref{eq-estimate-super-phi-multi}, we have, for each $i=1, \cdots, n$ 
\begin{align}\label{eq-estimate-super-phi-special-multi}\begin{split}
\overline{\phi}(p_i)=-2\log\e+O(1).
\end{split}\end{align}
By Theorem  \ref{theorem-existence-solutions}, 
there exists a unique solution  $\phi$ of $\mathcal S(\phi)=0$, satisfying 
\begin{equation}\label{eq-estimate-solution-upper-lower-multi} 
\underline{\phi}< \phi< \overline \phi.\end{equation}
We now construct an anisotropic subsolution $\underline{\phi}_*$  so that 
the value of $\overline{\phi}$ at each $p_i$ is significantly less than the value of 
$\underline{\p}_*$ at any point away from all $p_i$.
For each $i=1, \cdots, n$,  denote by $d_i$ the distance on $\mathbb S^2$ from $p_i$.  
Then, $d_i(p_i)=0$. 

For some positive constant $\mu \in (0, 1)$ 
and a small positive constant $\delta$ to be determined, construct
\begin{equation}\label{eq-expression-phi-multi}
\underline{\phi}_*=\mu\sum_{i=1}^n\log\Big[\f{1+\delta-\cos d_i}{\delta} \Big].\end{equation}
Then, 
for any fixed $r_0\in (0,\pi)$ and any $q\in \mathbb S^2\setminus \cup_{i=1}^nB_{p_i}(r_0)$, we have 
\begin{equation}\label{eq-subsol-phi-bound-outside-multi}
\underline{\phi}_*(q)
=-n\mu\log\delta+O(1). 
\end{equation}
Moreover,  
\begin{equation*}
\D_0\underline{\phi}_*+n\mu=\sum_{i=1}^n\f{\mu\delta(2+\delta)}{(1+\delta-\cos d_i)^2}.
\end{equation*}
Then, it holds
\begin{equation*}
(\D_0\underline{\phi}_*+n\mu)e^{-\underline{\phi}_*}=\frac12 \underline{f}, 
\end{equation*}
where  
\begin{equation}\label{eq-expression-f-multi}
\underline{f}=2\Big[\sum_{i=1}^n\f{\mu\delta(2+\delta)}{(1+\delta-\cos d_i)^2}\Big]
\prod_{i=1}^n\Big[\f{\delta}{1+\delta-\cos d_i}\Big]^{\mu}.\end{equation}
In the following, we will take $\mu$ such that $n\mu< 1$.

We now claim $\underline{f}> f$ on $\mathbb S^2$. 
Thus by \eqref{3.0-v2}, we have $\mathcal S(\underline{\phi}_*)>0$, since $n\mu<1$. To prove the claim, we fix an $i=1, \cdots, n$ and  prove $\underline{f}> \nu(n)$ on $B_{p_i}(\e)$, 
since $f\le \nu(n)$ on $B_{p_i}(\e)$ 
and $f=0$ outside $\cup_{i=1}^nB_{p_i}(\e)$. 
In the summation in the expression \eqref{eq-expression-f-multi}, we keep the term involving $d_i$ and drop other terms; 
while in the product, we single out the factor involving $d_i$. Then, 
\begin{equation*}
\underline{f}\ge 2\mu(2+\delta)\cdot
\f{\delta^{1+\mu}}{(1+\delta-\cos d_i)^{2+\mu}}
\prod_{j\neq i}\Big[\f{\delta}{1+\delta-\cos d_j}\Big]^{\mu}.\end{equation*}
Note that there are $n-1$ factors in the product. By $\cos d_i\ge 1-d_i^2/2$ and $\cos d_j\ge -1$ for $j\neq i$, we have 
\begin{equation*}
\underline{f}\ge \frac{2\mu(2+\delta)}{(2+\delta)^{\mu(n-1)}}\cdot
\f{\delta^{1+n\mu}}{(\delta+d_i^2/2)^{2+\mu}}.\end{equation*}
For some $\beta>0$ to be determined, we set $\delta=\e^\beta$. Then, 
\begin{equation*}
\underline{f}\ge 2\mu(2+\e^\beta)^{1-\mu(n-1)}\cdot\f{\e^{\beta(1+n\mu)}}{(\e^\beta+\e^2/2)^{2+\mu}}
\quad\text{on }B_{p_i}(\e).\end{equation*}
If $\beta>2$, we write 
$$\e^\beta+\e^2/2=\f12\e^2(1+2\e^{\beta-2}),$$
and we hence have
\begin{equation*}
\underline{f}\ge \mu 2^{4+2\mu-n\mu} \cdot\f{(1+\e^\beta/2)^{1-\mu(n-1)}}{(1+2\e^{\beta-2})^{2+\mu}} 
\cdot\e^{\beta(1+n\mu)-2(2+\mu)}\quad\text{on }B_{p_i}(\e).
\end{equation*}
We can take
\begin{equation}\label{kappa n}
\nu(n)=\mu 2^{1+2\mu-n\mu}.
\end{equation}
Thus, $\underline{f}> \nu(n)$ on $B_{p_i}(\e)$ 
if $\beta(1+n\mu)=2(2+\mu)$ and $\e$ is small. In the following, we take  
\begin{equation}\label{eq-expression-beta-multi}\beta=\f{2(2+\mu)}{1+n\mu}.\end{equation}
Then, 
\eqref{eq-subsol-phi-bound-outside-multi} reduces to, 
for any $q\in \mathbb S^2\setminus \cup_{i=1}^nB_{p_i}(r_0)$,  
\begin{equation}\label{eq-subsol-phi-bound-outside-v2-multi} 
\underline{\phi}_*(q)
=-n\mu\beta\log\e+O(1).
\end{equation}
At the moment, $\nu(n)$ also depends on $\mu$. Later on, we will fix $\mu$ depending only on $n$.  

We now compare $\underline{\phi}_*$ and $\overline{\phi}$. 
By \eqref{eq-estimate-super-phi-special-multi} and \eqref{eq-subsol-phi-bound-outside-v2-multi}, we have,  
for any $i=1, \cdots, n$ and any $q\in \mathbb S^2\setminus \cup_{i=1}^nB_{p_i}(r_0)$, 
\begin{equation*}
\overline{\phi}(p_i)-\underline{\phi}_*(q)
\le(n\mu\beta-2)\log\e+C,
\end{equation*} 
where we replaced $O(1)$ by $C$, a positive constant independent of $\epsilon$.
Note that the maximal $\mu$ is given by $1/n$. 
With $\beta$ as in \eqref{eq-expression-beta-multi}, if $\mu=1/n$, then $n\mu\beta=2+1/n$. 
Take any $\gamma\in (0, 1/n)$, say $\gamma=1/(2n)$, and solve the algebraic equation 
$n\mu\beta=2+\gamma$, i.e., 
$$2n\mu^2+(2-\gamma)n\mu-(2+\gamma)=0.$$
Its positive solution is given by 
$$\mu=\f{1}{4n}\big[\sqrt{(2-\gamma)^2n^2+8n(2+\gamma)}-(2-\gamma)n\big].$$
We take such $\mu$ and fix $\nu(n)$ accordingly. For $n$ large, we note that 
$$\f1n-\f{1}{n^2}+o\big(\f1{n^2}\big)<\mu<\f1n,$$
where the expression on the left-hand side corresponds to the solution for $\gamma=0$. 

By $\mathcal S(\p)=0$, $\mathcal S(\underline{\p}_*)>0$, and Lemma \ref{lemma-comparison}, 
we have $\underline{\p}_*<\p$. Here, in order to apply Lemma \ref{lemma-comparison}, 
we may need to enlarge $K$ so that $|\underline{\p}_*|_{C^2(\mathbb S^2)}\le K$, if necessary. 
Hence, for any $i=1, \cdots, n$ and $q\in \mathbb S^2\setminus \cup_{i=1}^nB_{p_i}(r_0)$, we prove have
$$\phi(p_i)-\phi(q)\le \overline{\phi}(p_i)-\underline{\phi}_*(q)
\le \gamma\log\e+C.$$
We then restore $\e$ and $\ub$, and write $\p_{\e}(\cdot, \ub)$ for $\phi$. This $\phi_{\epsilon}(\o,\ub)$ is our desired multi-valley solution. 
\end{proof} 

\appendix

\section{Estimates of Geometric Quantities}\label{appendix-A}

In this appendix, we collect useful estimates for geometric quantities $\O, \chi, \chib, \eta, \omb$  from \cite{AL}, which play important roles in this paper. All these quantities are defined and pointwisely estimated as listed below at 
$(\o, u, \ub)$ with $u=1-R$ and $R=\ub a e^{-\p}$. Recall $g=R^2\gamma$ in \eqref{eq-relation-g-gamma},  
$\Omega$ defined in \eqref{equation g}, and $\chi, \chib, \eta, \omb$ defined in \eqref{def Ricci coefficients}. Moreover, 
$\alpha_1$, $\alpha_2$ are defined in \eqref{eq-definitions-alphas}, i.e, 
\begin{equation*}\begin{split}
\alpha_1&=-\Big(\f12\O\tr_g\chib+\f{1}{R}+4\O\omb\Big),\\
\alpha_2&=\f12\O^{-1}\tr_g\chi-\f{1}{R}+\f{\ub a}{2R^2} f(\o,\ub).
\end{split}\end{equation*}
With $C$ as a uniform constant, we have the following estimates. 

\begin{itemize}
\item Estimates of geometric quantities and their angular derivatives:
\end{itemize}
\begin{equation}\label{A1}
\begin{split}
&{\color{black}R^2}|\chibh_{ij}|_g+{\color{black}R^2}|\eta_i|_g +{\color{black}R}|\O-1|\le {Ce^{\p}a^{-\f13}}, 
\\
&|R\O\tr_g\chib+2|+|\nab_{\gamma}(R\O\tr\chi_g \chib)|_{\gamma}+R^{-1}|\Omega\chibh_{kl}|_{\gamma}+R^{-1}|\nab_{\gamma}(\Omega\chibh_{kl})|_{\gamma}\le Ce^{\p}a^{-\f13},\\
&{\color{black}\ub}|\O\tr_g\chi|+\ub|\nab_{\gamma}(\O\tr_g \chi)|_{\gamma}+{\color{black}\ub}R^{-2}|\O\chih_{kl}|_{\gamma}+{\color{black}\ub}R^{-2}|\nab_{\gamma}(\O\chih_{kl})|_{\gamma}\leq {C}e^{\p}a^{-\f12},\\
&{\color{black}\ub R^{-2}|\Omega d^A \f{\partial \chi_{kl}}{\partial \theta^A}|_{\gamma}+\ub R^{-2}|\nab_{\gamma}(\Omega d^A \f{\partial \chi_{kl}}{\partial \theta^A})|_{\gamma}\leq Ce^{\phi}a^{-\f12},}\\
&|\eta_i|_{\gamma}+|\nab_{\gamma}\eta_i|_{\gamma}+R|\alpha_1|+R|\nab_{\gamma}\a_1|_{\gamma}+R|\alpha_2|+R|\nab_{\gamma}\a_2|_{\gamma} \le Ce^{\p}a^{-\f13}.
\end{split}
\end{equation}
These estimates can be read from norms defined in Section 2.6 and the statements of Theorem 3.1 and Proposition 5.1 in \cite{AL}.  
\smallskip

\begin{itemize}
\item Estimates of $u$-derivatives:
\end{itemize}
\begin{equation}\label{A2}
\begin{split}
&|\partial_u(\Omega\chibh_{kl})|_{\gamma}+R|\partial_u\eta_i|_{\gamma}
+R^2|\partial_u\alpha_1|+R^2|\partial_u\alpha_2|\le Ce^{\p}a^{-\f13}.\\
\end{split}
\end{equation}
We can derive the estimates for $\partial_u (\O\chibh_{kl}), \partial_u \eta_i, \partial_u\a_1$ via using equations in (2.4)  and the statement of Theorem 3.1 in \cite{AL}. For the estimate for $\partial_u \a_2$, we apply $(4.9)$ in \cite{An17} and get 
\begin{equation*}
R^2|\partial_u \a_2|=\Big|\f12 R^2 \partial_u(\O^{-1}\tr_g\chi)-1+R^{-1}\ub a f(\o,\ub)\Big|\le Ce^{\p}a^{-\f13}.
\end{equation*}

\begin{itemize}
\item Estimates of $\ub$-derivatives: 
\end{itemize}
\begin{equation}\label{A3}
\begin{split}
&\ub R^{-1}|\partial_{\ub}(\Omega\chibh_{kl})|_{\gamma}+\ub|\partial_{\ub}\eta_i|_{\gamma}
+\ub R|\partial_{\ub}\alpha_1|
+\ub R|\partial_{\ub}\alpha_2|\le {Ce^{\p}a^{-\f13}}.
\end{split}
\end{equation}
The estimates for  $\partial_{\ub}(\O\chibh_{kl}), \partial_{\ub}\eta, \partial_{\ub}\a_1$ can be derived via using equations in (2.4) and the statement of Theorem 3.1 in \cite{AL}. To obtain the estimate for $\partial_{\ub} \a_2$, recall the definition of $f$ in the statement of Theorem \ref{main thm}, i.e., 
\begin{equation*}
\ub a f(\o,\ub)=\int_0^{\ub}|\chih_{0}|^2_g(\o, \ub')d\ub'.
\end{equation*}
Hence, with $R=1-u$ and $\partial_{\ub}u=0$, it deduces
$$\ub R\partial_{\ub}\a_2=\f{\ub R}{2}\partial_{\ub}(\O^{-1}\tr_g\chi)+\f{\ub}{2R}|\chih_0|_g^2(\o, \ub).$$
Using the first and the last equations of (2.4) and the estimates ensured by Theorem 3.1 in \cite{AL}, we get
\begin{equation*}
\begin{split}
&\big|\partial_{\ub}(\O^{-1}\tr_g\chi)(\o, 1-R, \ub)+|\chih|_g^2(\o, 1-R, \ub)\big|\leq \f{Ca}{R^2}a^{-\f13},\\
&\big|-|\chih|_g^2(\o, 1-R, \ub)+{R^{-2}}|\chih_0|_g^2(\o, \ub)\big|\leq \f{Ca}{R^2}a^{-\f13}.
\end{split}
\end{equation*}
Via the triangle inequality, it follows
\begin{equation*}\label{A4}
\big|\partial_{\ub}(\O^{-1}\tr_g\chi)(\o, 1-R, \ub)+R^{-2}|\chih_0|_g^2(\o, \ub)\big|\leq \f{Ca}{R^2}a^{-\f13}.
\end{equation*}
Multiplying  both sides of the above inequality by $\ub R/2$, we hence obtain
\begin{equation*}
\ub R|\partial_{\ub} \a_2|=\Big|\f{\ub R}{2}\partial_{\ub}(\O^{-1}\tr_g\chi)+\f{\ub}{2R}|\chih_0|_g^2(\o, \ub)\Big|\leq \f{C\ub a}{R^2}a^{-\f13}=Ce^{\p}a^{-\f13}. 
\end{equation*}

\end{document}